\documentclass[10pt,a4paper,reqno]{amsart}
\usepackage{geometry}
\usepackage{enumerate}
\usepackage{amsmath}
\usepackage{amssymb}
\usepackage{amsthm}
\usepackage{xcolor}
\usepackage{verbatim}
\usepackage{tabularx}
\usepackage{mathrsfs}
\usepackage{hyperref}
\hypersetup{colorlinks=true, linktoc=page, linkcolor=red, citecolor=blue}
\usepackage{empheq}
\newcommand{\RR}{\mathbb{R}}

\DeclareMathOperator{\tr}{tr}
\DeclareMathOperator{\Ricc}{Ric}

\theoremstyle{plain}
\newtheorem{theorem}{Theorem}[section]
\newtheorem{lemma}[theorem]{Lemma}
\newtheorem{proposition}[theorem]{Proposition}
\newtheorem{corollary}[theorem]{Corollary}
\newtheorem{question}[theorem]{Question}
\theoremstyle{remark}
\newtheorem{remark}[theorem]{Remark}
\numberwithin{equation}{section}

\theoremstyle{definition}
\newtheorem{definition}{Definition}[section]

\DeclareMathOperator{\diver}{div}

\newcommand{\Riem}{\mathrm{Riem}}
\newcommand{\Ric}{\mathrm{Ric}}
\newcommand{\ric}{\mathrm{Ric}}

\newcommand{\di}{\mathrm{d}}

\newcommand{\p}{\varphi}
\newcommand{\pa}[1]{{\left(#1\right)}}                  
\newcommand{\sq}[1]{{\left[#1\right]}}                  
\newcommand{\abs}[1]{{\left|#1\right|}}                 
\newcommand{\hess}{\operatorname{Hess}}

\renewcommand{\div}{\diver}

\makeatletter
\newcommand*\owedge{\mathpalette\@owedge\relax}
\newcommand*\@owedge[1]{
	\mathbin{
		\ooalign{
			$#1\m@th\bigcirc$\cr
			\hidewidth$#1\m@th\wedge$\hidewidth\cr
		}
	}
}
\makeatother

\newcommand{\KN}{\mathbin{\bigcirc\mspace{-15mu}\wedge\mspace{3mu}}}

\title{On the Geometry of Cotton Gravity}

\date{\today}

\author{Giulio Colombo}
\address{Dipartimento di Matematica ``F. Enriques'', Universit\`a degli Studi di Milano, Via Saldini 50, I-20133 Milano, Italy}
\email{giulio.colombo@unimi.it}
\author{Filippo Mastropietro}
\thanks{The second author figures as corresponding author}
\address{Dipartimento di Matematica ``F. Enriques'', Universit\`a degli Studi di Milano, Via Saldini 50, I-20133 Milano, Italy}
\email{filippo.mastropietro@unimi.it}

\author{Marco Rigoli}
\address{Dipartimento di Matematica ``F. Enriques'', Universit\`a degli Studi di Milano, Via Saldini 50, I-20133 Milano, Italy}
\email{marco.rigoli@unimi.it}

\begin{document}
	
\maketitle

\begin{abstract}
	We analyze the geometry of the field equations of Cotton gravity (for a quite general energy-momentum tensor) on a static space-time. In particular, we describe the local structure of the spatial Riemannian factor. This structure, that we call \emph{Cotton-$\p$-perfect fluid} (C-$\p$-PF, for short) is a generalization to the regime of Cotton Gravity of the recently introduced notion of $\p$-static perfect fluid space-time ($\p$-SPFST). After discussing the variational origin of this system,  we provide sufficient conditions for a C-$\p$-PF to reduce to a $\p$-SPFST. We also study the geometry of the level sets of the lapse function $f$ and  we provide a rigidity result for C-$\p$-PFs under some curvature conditions. The role that Codazzi tensors hold in this theory is highlighted.
\end{abstract}

\bigskip

\noindent \textbf{MSC 2020} {
	Primary: 	53C25, 
	53C24, 
	Secondary: 	53C50, 
	53C80. 
}

\noindent \textbf{Keywords} {
	Cotton tensor$\cdot$
	Static space-times $\cdot$
	Harmonic maps  $\cdot$
	General Relativity
}

\section{Introduction}
Despite  its mathematical beauty, elegance and numerous successes in the description of cosmological phenomena on the scale of the solar system and the nearby universe, Einstein theory of relativity has encountered observational challenges at the cosmic level which, in the last years, have been confirmed by the Hubble telescope data. 
The most important paradigm that has been proposed to solve these challenges has been the introduction 
of dark matter and dark energy.
However, up today, we still lack sound and convincing explanations for these new fundamental concepts. Indeed, since the sixties, in the literature have appeared a plethora of what are called \emph{extended theories of gravity} in order to find motivated answers to the physical problems coming from observations.
In this paper we focus our attention on an approach recently proposed by J. Harada, \cite{Harada}, which, since its appearance, has been  called \emph{Cotton gravity}. In this theory, the gravitational dynamic is described not by the Einstein tensor, but by the Cotton tensor, notably a 3-covariant tensor, given by third order derivatives of the metric.
For instance, the field equations are
\begin{equation} \label{Intro: Har_eq}
	\hat C = -(n-2)\div_1\hat{\mathcal T} \, ,
\end{equation}
where $\hat{C}$ is the Cotton tensor of the $n$-dimensional space-time $(\hat{M},\hat{g})$, and $\div_1\hat{\mathcal T}$ is a certain combination of the components of the covariant derivative of the stress-energy tensor $\hat{T}$.
As a matter of fact, equation \eqref{Intro: Har_eq} is obtained performing  a peculiar  variation of the action functional, in which the connection, instead of the metric, is the fundamental entity. Solutions of Einstein equations also satisfy the Cotton gravity equations, but the converse is generally false.
From the physical point of view, Cotton gravity has been  first put in doubt by Bargueño, \cite{PhysRevD.104.088501}, but Harada was able to show that the galactic rotation curves  of 84 far away galaxies can be explained by his theory without the need for dark matter contributions, \cite{Harada2022}.
Some other criticisms, always regarding the physical validity of the theory, have been raised in \cite{Clement:2023tyx} and \cite{Altas2025} (see also \cite{Altas2025Killing}, where a similar framework, called \emph{Conformal Killing Gravity}, is addressed) and then confuted in \cite{mantica2025noteconservedcurrentsstatic}, \cite{Feng2024} and \cite{sussman2024responsecritiquecottongravity}.
In any case, Cotton gravity is presently an active field of research and, besides these controversial disputes, it reveals to be certainly interesting  from the purely geometric point of view.
The purpose of the present work is to use its field equations  as a stepping stone to derive meaningful generalizations of mathematical structures that, despite of their physical origin, have been vastly analysed by the mathematical community.
These include, most notably, \emph{vacuum static spaces} (see, for example, \cite{Ambrozio}, \cite{BorghiniMazzieri1}, \cite{Chrusciel}, \cite{HwangYun2}, \cite{Shen}, \cite{QingWan}) and \emph{static perfect fluid space-times} (see \cite{Ribeiro1}, \cite{Ribeiro2}, \cite{Lindblom}, \cite{Masood}).

It is known that these two classes of spaces have a mathematical importance which is independent from their relation to physics: the former class is tightly related to the theory of the deformation of the scalar curvature (see \cite{FischerMarsden1}, \cite{FischerMarsden2}), while the latter is a special instance of an Einstein-type structure, \cite{CatMastMontRig}.
The interplay between these mathematical objects is very rich: think, for example, about the work of Kobayashi and Obata, \cite{KO}, (see also the successive \cite{Kobayashi}, \cite{Lafontaine}) on the integrability conditions of vacuum static spaces and their impact on the geometry of the level sets of the static potential $f$. The same technique has been later applied to Ricci solitons, \cite{Cao2013}, and more general classes of spaces, \cite{CatMastMontRig}, \cite{Anselli_2021}. This circle of ideas will also be crucial to the present work. The structures that we are going to study encompass  the previous ones, and have, as a most distinctive feature, that of including terms depending on \emph{third order} derivatives of the metric tensor.

Our investigation is concerned with the geometry of the field equations in case our Lorentzian manifold $\hat{M}$ is a static space-time. More specifically, if $\hat{M}$ splits as $\RR\times M$ with $(M,g)$ a Riemannian $m$-manifold, then, for some $f\in C^\infty(M),$ we consider on $\hat{M}$ the Lorentzian static space-time metric $\hat{g}$ given by
\begin{equation*}
	\hat{g}=-e^{-2\hat{f}}\di t\otimes \di t+g
\end{equation*}
where $t$ is the canonical coordinate on $\RR$, $\hat{\pi}_M: \hat{M}\to M$ is the projection and $\hat{f}= \hat{\pi}_M\circ f$.
We will consider physical contributions given by two functions $\mu,p\in C^\infty(M)$, that are to be thought as density and pressure so that $\hat{\mu}= \hat{\pi}_M\circ \mu$ and $\hat{p}= \hat{\pi}_M\circ p$ will be the density and pressure of a static perfect fluid on $\hat{M}$, and by $\hat{\p}: \hat{M}\to (N,h)$, a non-linear field of the form $\hat{\p}=\hat{\pi}_M\circ \p, \p:M\to (N,h)$ with a self-interacting potential $U:N\to \RR$. We will study how, in this context, the equations of Cotton gravity,  from now on called C-$\p$-PF equations, reduce on the time-slice $(M,g)$. Before doing so, let us set some notations. Since in the following we will make extensive use of iterated covariant derivatives, metric contractions and skew-symmetrizations, we did not find it convenient to adopt modern, global notation \emph{à la} Koszul and we instead opted for an ``abstract index'' notation. Therefore, tensors and equations between them are often presented component-wise. Due to the manifest tensorial nature of all the given quantities, it will be clear to the reader that all of the expressions that we will introduce will be, indeed, globally defined and independent from the given local frame. Computations will be carried out with respect to a local orthonormal frame. More precisely, given an $m$-dimensional semi-Riemannian manifold $(M,g)$ of signature $(k,m-k)$ and a small open subset $U$ of $M$, we can always find a family of orthonormal one-forms $\{\theta^i\}_{i=1}^m$ on $U$ such that
\begin{align*}
	g=-\sum_{i=1}^k \theta^i\otimes \theta^i+\sum_{j={k+1}}^m\theta^j\otimes \theta^j \ \ \text{ on } U.
\end{align*}
When $g$ is a Riemannian metric, we deduce that the components $g_{ij}$ of the metric tensor on the dual frame to $\{\theta^i\}$ coincide with the Kronecker delta $\delta_{ij}$. Due to this fact, the operation of lowering or raising indexes does not change the value of the components of a tensor. Therefore, when we will work with a Riemannian metric, we will consider every tensor to be purely covariant and use the musical isomorphism implicitly. From now on, Einstein summation convention is on force, unless otherwise stated. In Section \ref{Sect 2: phi-curv} we will introduce some special tensors, called $\p$-\emph{curvatures}, that encode both the contributions of the classical curvature tensors and that of the map $\p$.
The most important of these tensors are the $\p$-Ricci tensor
\begin{align*}
	\ric^\p:=\ric-\alpha\p^*h,
\end{align*}
where $\alpha$ is a fixed, non-zero, coupling constant, and the $\p$-scalar curvature
\begin{align*}
	S^\p:=\tr_g \ric^\p.
\end{align*}
Combining them, we can define the $\p$-Schouten tensor
\begin{align*}
	A^\p:=\ric^\p-\frac{1}{2(m-1)}S^\p g
\end{align*}
and the $\p$-Cotton tensor, of components
\begin{align*}
	C^\p_{ijk}:=A^\p_{ij,k}-A^\p_{ik,j};
\end{align*}
here and in the following commas denote covariant derivatives.
The $\p$-Weyl tensor is defined by
\begin{align*}
	W^\p:=\Riem-\frac{1}{m-2}A^\p\KN g,
\end{align*}
where $\KN$ denotes the Kulkarni-Nomizu product. The $\p$-Bach tensor is given, in local coordinates, by the expression
\begin{align*}
	(m-2)B^\p_{ij} & = C^\p_{ijk,k}+R^\p_{lk}W^\p_{likj}-\alpha R^\p_{lj}\p^a_l\p^a_i +\alpha\pa{\p^a_{ij}\p^a_{kk}-\p^a_{kkj}\p^a_i-\frac{\abs{\tau(\p)}^2}{m-2}\delta_{ij}}, \nonumber
\end{align*}
where $\tau(\p)$ is the tension of the map $\p$,  as defined in \cite{EL}.
We will also need the following two tensors, that resemble the $D$ tensor introduced by Cao and Chen in \cite{Cao2013}:
\begin{align*}
	D^A_{ijk} & := \frac{1}{m-2}\bigg[f_kR^\p_{ij}-f_jR^\p_{ik}+\frac{1}{m-1}f_t\pa{R^\p_{tk}\delta_{ij}-R^\p_{tj}\delta_{ik}} -\frac{S^\p}{m-1}\pa{f_k\delta_{ij}-f_j\delta_{ik}}\bigg]
\end{align*}
and
\begin{align*}
	D^B_{ijk} & := \frac{1}{m-2}\bigg[ f_j f_{ik}-f_k f_{ij}+\frac{1}{m-1}f_t \pa{f_{tj}\delta_{ik}-f_{tk}\delta_{ij}} -\frac{\Delta f}{m-1}\pa{f_j \delta_{ik}-f_k\delta_{ij}}\bigg]. 	
 \end{align*}
 We will comment on their importance and meaning in a second.
 With these preparations in mind, the equations of C-$\p$-PF are
\begin{subequations}\label{Intro: C-phi-PF}
 \begin{empheq}[left={\empheqlbrace}]{alignat=2}
 		0 & = C^\p_{ijk}+f_lW^\p_{lijk}-D^A_{ijk}-(m-2)D^B_{ijk}+\frac{1}{m-1}U^a\pa{\p^a_j\delta_{ik}-\p^a_k\delta_{ij}}\label{Intro: C-phi-PF 1}\\
 		0 & = -\frac{m-1}{m}f_{lli}+\frac{m-2}{m}f_lf_{li}+\pa{\Delta f}f_i+\frac{1}{2m}S^\p_i-f_lR^\p_{il}+\frac{1}{m}U^a\p^a_i-\frac{m-1}{m}\mu_i,\label{Intro: C-phi-PF 2}\\[0.1cm]
 		0 & = \nabla p-(\mu+p)\nabla f,\label{Intro: C-phi-PF 3}\\
 		0 & = \tau(\p)-\di\p(\nabla f)-\frac{1}{\alpha}  \nabla^h U (\p) \label{Intro: C-phi-PF 4}.
 	\end{empheq}
\end{subequations}

In \eqref{Intro: C-phi-PF}, indexes refer to two local orthonormal co-frames $(U, \{\theta^i\}), (V, \{\omega^a\})$, respectively in $M$ and $N$, with the property that $\p(U)\subset V$. Clearly, $1\leq i,j,... \leq m$ and $1\leq a,b,... \leq d=\dim N$. $\nabla ^h U$ denotes the covariant derivative of $U$ with respect to the metric $h$ of $N$, so that $\nabla ^h U (\p)$ will be the corresponding section  of the pullback bundle $\p^*TN$.\\

It is useful to compare the solutions to \eqref{Intro: C-phi-PF} to those of the Einstein Field Equations corresponding to the same stress-energy tensor: these are called $\p$-\emph{Static Perfect Fluid Space-Times} ($\p$-SPFST, for shorts) and are studied in our forthcoming work \cite{phiSPFST}, jointly with L. Branca  and P. Mastrolia. In this case, the equations write as
\begin{align}\label{Intro: phi SPFST}
	\begin{cases}
		\ric^\p+\hess f-\di f\otimes \di f-\dfrac{1}{m-1}\pa{\dfrac{S^\p}{2}-p+U(\p)}g=0,\\[0.3cm]
		\Delta_f f=-\dfrac{1}{m-1}\sq{mp-mU(\p)+\dfrac{m-2}{2}S^\p},\\[0.3cm]
		\tau(\p)-\di \p(\nabla f)=\dfrac{1}{\alpha} \nabla U,\\[0.2cm]
		\mu+U(\p)=\dfrac{1}{2}S^\p,\\[0.2cm]
		\nabla p-(\mu+p)\nabla f=0.
	\end{cases}
\end{align}
	Here, we are using the notation \begin{equation*}
	\Delta_ f u:=\Delta u-g(\nabla u, \nabla f),
\end{equation*}
for any $u\in C^\infty(M)$. The operator $\Delta_f$ is often called the $f$-\emph{laplacian}.

When comparing the analysis of system \eqref{Intro: C-phi-PF} to what has been done in \cite{phiSPFST}, a question comes to mind: there, we were able to describe, for $m\geq 3$, the geometry of the regular level sets of a solution $f$ of \eqref{Intro: phi SPFST} \emph{via} the previously introduced  tensor $D^A$; to understand why this tensor is a natural object, it suffices to note its appearance in the first of the two integrability conditions of system \eqref{Intro: phi SPFST}, that is,
\begin{align}\label{Intro: prima cond int phi spfst}
	(m-1)D^A_{ijk}=C^\p_{ijk}+f_l W^\p_{lijk}+\frac{1}{m-1}U^a\pa{\p^a_j\delta_{ik}-\p^a_k\delta_{ij}}.
\end{align}
In doing so, we were inspired by the work of Cao and Chen \cite{Cao2013} on Ricci solitons, see also \cite{CatMastMontRig}, and by the (much older) work of Kobayashi and Obata \cite{KO} on vacuum static spaces.
\begin{question}\label{Question: domanda 1}
	Are similar considerations possible for \eqref{Intro: C-phi-PF}, i.e. for Cotton gravity?
\end{question}
Due to the shape of equations \eqref{Intro: C-phi-PF}, an affirmative answer seems possible; indeed,  it is easy to prove that, when \eqref{Intro: phi SPFST} holds, then $D^A=D^B$ so that \eqref{Intro: C-phi-PF 1} reduces to \eqref{Intro: prima cond int phi spfst}. Therefore, \eqref{Intro: C-phi-PF 1} is a nice generalization of the first integrability condition of a $\p$-SPFST.
%
In light of the work done in \cite{phiSPFST}, it seems now reasonable to study the geometry of the non-singular level hypersurfaces of $f$ in case $D^A\equiv D^B\equiv 0$. 
This is done in Theorem \ref{local struct: teorema warp product} below, under the further assumption that $\p$ is $\frac{U}{\alpha}$-harmonic, that is, it solves
\begin{align*}
	\tau(\p)=\frac{1}{\alpha} (\nabla^h U) (\p), 
\end{align*} 
see Remark \ref{Deduction: remark: U armonica} for more details.
The key concept here is that of a \emph{warped product metric}. Given a Riemannian manifold $(\Sigma,g_\Sigma)$, a real interval $I$ and a positive function $\rho\in C^\infty(I)$, the latter is given by the metric $g=\di r\otimes \di r+\rho^2(r)g_\Sigma$ on the manifold $I\times \Sigma$, where $r:I\times \Sigma \to I$ denotes the projection.

Pushing this reasoning further, in the next result we give sufficient conditions (of which some are also necessary) for the validity of the requirement $D^A\equiv D^B\equiv 0$. This is the main Theorem of the present work and its statement is the following:
\begin{theorem}\label{Introduzione: codazzi: teo: teoremone}
	Let $(M,g)$ be a complete $m$-dimensional Riemannian manifold, $m\geq 3$, and let $\p:(M,g)\to (N,h)$ be a smooth map of Riemannian manifolds. Let $\alpha \in \RR,\alpha > 0, f\in C^\infty(M), U\in C^\infty(N)$. Assume that $(M,g)$ satisfies
	\begin{align*}
		\begin{cases}
			0=C^\p_{ijk}+f_lW^\p_{lijk}-\dfrac{1}{m-1}U^a\pa{\p^a_k\delta_{ij}-\p^a_j\delta_{ik}}-D^A_{ijk}-(m-2)D^B_{ijk},\\[0.2cm]
			\tau(\p)-\di \p(\nabla f)=\dfrac{1}{\alpha} \nabla^h U.
		\end{cases}
	\end{align*}
	Let $S^2(M)$ be the space of $2$-covariant, symmetric tensors on $M$ and define a linear map $F:S^2(M)\to C^\infty(M)$ by setting, for $\beta\in S^2(M), \beta=\beta_{lk}\theta^l\otimes \theta^k$ locally,
	\begin{align}\label{Introduzione: codazzi: teoremone: funzionale F}
		F(\beta):=\pa{\frac{m}{m-1}U^a\p^a_lf_k+D^A_{lik}f_i-W^\p_{likj}f_if_j}\beta_{lk}.
	\end{align}
	Assume that
	\begin{enumerate}[1.]
		\item $f$ is proper;\label{main: f propria}
		\item $B^\p(\nabla f,\nabla f)=0$;\label{main: phi bach piatta}
		\item $\p$ is $\frac{1}{\alpha} U$-harmonic;\label{main: U amonic}
		\item for all $ p\in M$ regular for $f$, we have that $\nabla f$ is an eigenvector of $\ric^\p$ at $p$;\label{main: autovet}
		\item $\ric^\p+\hess f\in \ker (F)$.\label{main nucleo di F}
	\end{enumerate}
	Then, for each regular level set $\Sigma$ of $f$ and for every $p\in \Sigma$, there exists $A\subset M$ open such that $p\in A$ and $g_{|_A}$ is a warped product metric.  Moreover, $(\Sigma,g_{\Sigma})$ satisfies
	\begin{align*}
		\begin{cases}
			\ric^{\p_{|_{\Sigma}}}=\dfrac{S^{\p_{|_{\Sigma}}}}{m-1}g_\Sigma,\\[0.2cm]
			\tau(\p_{|_{\Sigma}})=0.
		\end{cases}
	\end{align*}
\end{theorem}
To prove Theorem \ref{Introduzione: codazzi: teo: teoremone}, we will construct a vector field on $M$ with non-negative divergence and integrate it on a family of subsets of the form $\Omega_{\delta,\eta}=\{x\in M: \delta\leq f(x)\leq \eta\}$, for some $\delta,\eta\in \RR, \delta<\eta$. The assumptions of Theorem \ref{Introduzione: codazzi: teo: teoremone} are at first glance unappealing, but they are motivated as follows:
\begin{itemize}
	\item[1)] that $f$ is proper is needed to obtain the compactness of the sets $\Omega_{\delta,\eta}$ in order to apply the divergence theorem;
	\item[2)] assumption \ref{main: phi bach piatta} is inspired by the work of Cao and Chen, \cite{Cao2013};
	\item[3)] assumption \ref{main: U amonic} is needed to make our analysis of the regular level sets of $f$ effective, as explained in Section \ref{Sect 4: local struct};
	\item[4)] the last two assumptions are rather technical; we note that \ref{main nucleo di F} is automatically satisfied by a $\p$-SPFST if $\p$ is $\frac{1}{\alpha}U$-harmonic, so that the present setting is more general. Moreover, as explained in Remark \ref{Remark U=0} below, in the situation of a vanishing potential $U\equiv 0$, we have that assumptions \ref{main: U amonic} and \ref{main: autovet} can be removed after a small modification of the functional $F$.
\end{itemize}

With a little effort, it can be proved that solutions of \eqref{Intro: phi SPFST} are also solutions of \eqref{Intro: C-phi-PF} so that every $\p$-SPFST is a C-$\p$-PF. In this case, Theorem \ref{Introduzione: codazzi: teo: teoremone} reduces to \cite[Theorem 4.17]{phiSPFST}.
In particular, we have proved that techniques used to obtain rigidity results for solutions to Einstein field equations can sometimes be readapted to the present setting.
In light of this fact, one might ask how ``far'' solutions to \eqref{Intro: C-phi-PF} are to solutions of \eqref{Intro: phi SPFST}.
\begin{question}\label{question: domanda 2}
	When a solution of \eqref{Intro: C-phi-PF} is a solution of \eqref{Intro: phi SPFST}?
\end{question}

In Section \ref{Sect 5: riem comp} we address this issue, giving sufficient conditions under which a C-$\p$-PF solution is ``almost'' (that is, up to two integration constants) a $\p$-SPFST solution. The precise meaning of the above statement is given in Theorem \ref{riem comp: risultato principale} below. The result is obtained by studying how the existence of a Codazzi tensor on $M$ affects the algebraic structure of the Weyl tensor. Similar considerations, for a 4-dimensional space-time, are presented in Section 6 of \cite{Mantica2012WeylCT}, where the authors study the vanishing of a certain tensor, often called the ``magnetic component'' of the Weyl tensor.
More in depth, recall that a 2-covariant, symmetric tensor $P$ on a semi-Riemannian manifold $(M,g)$ is said to be a \emph{Codazzi tensor} if its covariant derivative is a totally symmetric 3-covariant tensor.	Then it can be proved that $P$  satisfies the following, more general, condition:
\begin{align}\label{Intro: weyl comp}
	W^l_{\ iks}P_{jl}+W^l_{\ isj}P_{kl}+W^l_{\ ijk}P_{sl}=0,
\end{align}
which is called \emph{Weyl compatibility} and it has been introduced by Mantica and Molinari in \cite{Mantica2012WeylCT}.
Studying equation \eqref{Intro: weyl comp}, one can force every Weyl compatible tensor on $(M,g)$ to be proportional to the metric by imposing some algebraic conditions on the Weyl tensor (the injectivity of the $W^*$ operator defined in Section \ref{Sect 5: riem comp}).
Choosing 
		\[
		P=\ric^\p+\hess f-\di f\otimes \di f,
		\] one can prove, \emph{via} a careful tensorial analysis, that, under these assumptions,
 system \eqref{Intro: C-phi-PF} reduces to a system closely related to \eqref{Intro: phi SPFST}. This is the content of Theorem \ref{riem comp: risultato principale}.\\
 \\
\textbf{Structure of the paper.}
In Section \ref{Sect 1: setting}, we recall some basic facts in Lorentzian Geometry and we introduce the Cotton Gravity equations. In subsection \ref{subset: origine variazionale}, following \cite{Harada}, we discuss and prove the variational origin of the field equations. Note that subsection \ref{subset: origine variazionale} is independent from the rest of the paper.

In Section \ref{Sect 2: phi-curv} we introduce the $\p$-curvature tensors. In Section \ref{Sect 3: system} we present the C-$\p$-PF equations and the tensors $D^A$ and $D^B$. We then compare the C-$\p$-PFs with their analogues in General Relativity, namely the $\p$-SPFSTs.
The derivation of system \eqref{Intro: C-phi-PF} starting from the field equations is postponed to Appendix \ref{Appendice A}, together with a general discussion concerning the role that Codazzi tensors hold in this theory. In Section \ref{Sect 4: local struct}, we first carry out the aforementioned analysis of the level sets of the function $f$ and then apply it to prove Theorem \ref{Introduzione: codazzi: teo: teoremone}.

In Section \ref{Sect 5: riem comp}, we discuss the notions of Riemann and Weyl compatibility and we prove Theorem \ref{riem comp: risultato principale} about the relation between systems \eqref{Intro: C-phi-PF} and \eqref{Intro: phi SPFST}. In Appendix \ref{Appendice B}, we further characterize Riemann and Weyl compatibility giving special emphasis to the case of 4-dimensional Riemannian manifolds. On the one hand, this will allow us to provide natural conditions under which the operator $W^*$ is injective, therefore giving more depth to the results of Section \ref{Sect 5: riem comp}, and, on the other hand, we will be able to strengthen some results of Mantica and Molinari, see \cite{Mantica2012RiemannCT}.
\tableofcontents

\section{The setting}\label{Sect 1: setting}
The aim of this section is to recall the definition of Cotton gravity and the variational origin of the equations.
\subsection{Basic Lorentzian geometry}\label{Lorentzian Geometry}

A Lorentzian manifold is a smooth $n$-dimensional manifold $\hat M$ with a non-degenerate $(0,2)$-symmetric tensor $\hat g$ of signature $(-,+,\dots,+)$. Greek letters $\alpha,\beta,\dots$ will denote indexes running from $0$ to $n-1$. Given a Lorentzian $n$-manifold $(\hat M,\hat g)$ and a point $p\in\hat M$, we can find an open neighbourhood $U\subset M$ of $p$ and a local co-frame $\{\omega^\gamma\}_{\gamma=0}^{n-1}$ on $U$ such that
\[
	\hat g = -\omega^0\otimes\omega^0 + \omega^1\otimes\omega^1 + \dots + \omega^{n-1}\otimes\omega^{n-1} \, .
\]
We also write
\[
	\hat g = g_{\gamma\beta} \, \omega^\gamma\otimes\omega^\beta
\]
where $(g_{\gamma\beta}) = \mathrm{diag}(-1,1,\dots,1)$. We denote by $g^{\gamma\beta}$ the entries of the inverse matrix of $(g_{\gamma\beta})$ and by $\{e_\gamma\}_{\gamma=0}^{n-1}$ the frame dual to $\{\omega^\gamma\}_{\gamma=0}^{n-1}$. One can then prove that, on $U$, there exist $1$-forms $\{\omega_{\ \gamma}^\beta\}_{\gamma,\beta=0}^{n-1}$, called the \emph{Levi-Civita connection forms}, characterized by the equations
\[
	\begin{cases}
		\di\omega^\gamma = -\omega_{\ \beta}^\gamma \wedge \omega^\beta \qquad \text{(first structure equations)} \\
		\omega_{\gamma\beta} + \omega_{\beta\gamma} = 0
	\end{cases}
\]
where $\omega_{\gamma\beta} := g_{\eta\gamma} \, \omega_{\ \beta}^\eta$. More in general, we will use $g_{\gamma\beta}$ and $g^{\gamma\beta}$ to raise and lower indexes. Then, the Levi-Civita connection of $(\hat M,\hat g)$ is defined \emph{via} the formula
\[
	\hat{\nabla} e_\gamma = \omega_{\ \gamma}^\beta \otimes e_\beta \, .
\]
The second structure equations read
\[
	\di\omega_{\ \gamma}^\beta = -\omega_{\ \gamma}^\eta \wedge \omega_{\ \eta}^\beta + \Omega_{\ \gamma}^\beta
\]
where the $2$-forms $\{\Omega_{\ \gamma}^\beta\}$ are the curvature forms associated to $\{\omega^\gamma\}_{\gamma=0}^{n-1}$. Setting $\Omega_{\gamma\beta} := g_{\gamma\eta} \, \Omega_{\ \beta}^\eta$, we also have
\[
	\Omega_{\gamma\beta} + \Omega_{\beta\gamma} = 0 \, .
\]
The $(0,4)$-version, $\widehat{\Riem}$, of the Riemann curvature tensor has components $\hat R_{\eta\beta\gamma\delta}$ determined by the expression
\[
	\Omega_{\eta\beta} = \frac{1}{2} \hat R_{\eta\beta\gamma\delta} \, \omega^\gamma \wedge \omega^\delta \, .
\]
The $(0,4)$-version has the symmetries
\[
	\hat R_{\eta\beta\gamma\delta} = \hat R_{\gamma\delta\eta\beta} = - \hat R_{\delta\gamma\eta\beta} = \hat R_{\delta\gamma\beta\eta} \, .
\]
Defining as usual the components of the $(1,3)$-version of $\widehat{\Riem}$ as
\[
	\hat R^\eta_{\ \beta\gamma\delta} = g^{\eta\rho} \hat R_{\rho\beta\gamma\delta},
\]
we can define those of the Ricci tensor, $\hat{\ric}=\hat{R}_{\eta\beta}\omega^\eta\otimes \omega^\beta$, by
\[
	\hat R_{\eta\beta} = \hat R^\gamma_{\ \eta\gamma\beta} \, .
\]
Note that $\hat R_{\eta\beta} = \hat R_{\beta\eta}$. Given a $(0,2)$-tensor $P$, we define its $\hat g$-trace by
\[
	\tr_{\hat g} P= g^{\eta\beta} P_{\eta\beta} = P^\eta_{\ \eta} \, .
\]
The scalar curvature of $(\hat M,\hat g)$ is
\[
	\hat S = \tr_{\hat g} \hat{\Ric} = g^{\eta\beta} \hat R_{\eta\beta} \, .
\]
Given two $(0,2)$-tensors $P$ and $Q$, their Kulkarni-Nomizu product $P\owedge Q$ is defined by
\[
	(P\owedge Q)_{\eta\beta\gamma\rho} := P_{\eta\gamma} Q_{\beta\rho} + P_{\beta\rho} Q_{\eta\gamma} - P_{\eta\rho} Q_{\beta\gamma} - P_{\beta\gamma} Q_{\eta\rho} \, .
\]
With this notation in mind, the components of the Weyl tensor of $(\hat M,\hat g)$ can be written as
\[
	\hat W_{\eta\beta\gamma\rho} = \hat R_{\eta\beta\gamma\rho} - \frac{1}{n-2} (\hat A\owedge\hat g)_{\eta\beta\gamma\rho}
\]
where
\[
	\hat A_{\eta\beta} = \hat R_{\eta\beta} - \frac{1}{2(n-1)} \hat S \, g_{\eta\beta}
\]
determines the Schouten tensor $\hat{A}$.
To express the components of the covariant derivative of a tensor, we will use commas to separate old and new indexes. For example, given a $(0,2)$-tensor $P$, the notation $P_{\eta\beta,\gamma}$ will denote the components of its covariant derivative $\hat{\nabla} P$, while $P_{\eta\beta,\gamma\rho}$ will be the components of its second covariant derivative and so on.
The Cotton tensor of $(\hat M,\hat g)$ is given by
\[
	\hat C_{\eta\beta\gamma} = \hat A_{\eta\beta,\gamma} - \hat A_{\eta\gamma,\beta} \, .
\]
When $n\geq 4$, $\hat C$ can be alternatively, but equivalently, defined by
\[
	\hat C_{\eta\beta\gamma} = -\frac{n-3}{n-2} \hat W^\rho_{\ \eta\beta\gamma,\rho} \, .
\]
We define the $1$-divergence of a $p$ times covariant tensor $P$ by setting
\[
	(\div_1 P)_{\alpha_1\dots\alpha_{p-1}} = g^{\eta\beta} P_{\eta\alpha_1\dots\alpha_{p-1},\beta}
\]
so that, when $n\geq 4$,
\[
	\hat C = -\frac{n-3}{n-2} \div_1 \hat W \, .
\]
\subsection{Cotton Gravity and General Relativity}\label{subset: Cotton and GR}
Given a symmetric covariant $(0,2)$-tensor field $\hat T$, the stress-energy tensor, the Einstein field equations write in the form
\begin{equation} \label{Ein_eq}
	\hat\Ric - \frac{1}{2} \hat S\hat g = \hat T \, .
\end{equation}
As it is well-known, $\hat\Ric - \frac{1}{2}\hat S\hat g$ is divergence-free by the Schur's identity (that is, the twice contracted second Bianchi identity for the Riemann tensor), so that \eqref{Ein_eq} forces $\hat T$ to be so.
Following \cite{Harada}, we define
\begin{equation}\label{Tensore T maiuscolo}
	\hat{\mathcal T} := \frac{1}{n-2} \left( \hat T\owedge \hat g - \frac{1}{2(n-1)} (\tr_{\hat g}\hat T) \hat g\owedge \hat g \right).
\end{equation}
Then, the equations of Cotton gravity are
\begin{equation} \label{Har_eq}
	\hat C = -(n-2)\div_1\hat{\mathcal T} \, .
\end{equation}
We compute $\div_1\hat{\mathcal T}$ explicitly. By definition,
\[
	\hat{\mathcal T}_{\eta\beta\gamma\delta} = \frac{1}{n-2} \left[ \hat T_{\eta\gamma} g_{\beta\delta} + \hat T_{\beta\delta} g_{\eta\gamma} - \hat T_{\eta\delta} g_{\beta\gamma} - \hat T_{\beta\gamma} g_{\eta\delta} - \frac{1}{n-1} \hat T^\rho_{\ \rho} (g_{\eta\gamma} g_{\beta\delta} - g_{\eta\delta} g_{\beta\gamma}) \right]
\]
so that
\begin{equation} \label{div_That}
	\begin{split}
	(\div_1\hat{\mathcal T})_{\beta\gamma\delta} & = \frac{1}{n-2} g^{\eta\rho} \left[ \hat T_{\eta\gamma,\rho} g_{\beta\delta} + \hat T_{\beta\delta,\rho} g_{\eta\gamma} - \hat T_{\eta\delta,\rho} g_{\beta\gamma} - \hat T_{\beta\gamma,\rho} g_{\eta\delta} \right] \\
	& \phantom{=\;} - \frac{1}{(n-1)(n-2)} \hat T^\xi_{\ \xi,\rho} (g_{\eta\gamma} g_{\beta\delta} - g_{\eta\delta} g_{\beta\gamma}) \\
	& = \frac{1}{n-2} \left[ (\div_1\hat T)_\gamma g_{\beta\delta} - (\div_1\hat T)_\delta g_{\beta\gamma} + \hat T_{\beta\delta,\gamma} - \hat T_{\beta\gamma,\delta} \right] \\
	& \phantom{=\;} - \frac{1}{(n-1)(n-2)} (\hat T^\xi_{\ \xi,\gamma} g_{\beta\delta} - \hat T^\xi_{\ \xi,\delta} g_{\beta\gamma}) .
	\end{split}
\end{equation}
Using \eqref{div_That} we can show that \eqref{Har_eq} implies $\div_1\hat T = 0$. Indeed, substituting \eqref{div_That} into \eqref{Har_eq} we get
\[
	\hat C_{\beta\gamma\delta} = (\div_1\hat T)_\delta g_{\beta\gamma} - (\div_1\hat T)_\gamma g_{\beta\delta} + \hat T_{\beta\gamma,\delta} - \hat T_{\beta\delta,\gamma} + \frac{1}{n-1} (\hat T^\eta_{\ \eta,\gamma} g_{\beta\delta} - \hat T^\eta_{\ \eta,\delta} g_{\beta\gamma})
\]
and then contracting both sides of the resulting equality with $g^{\beta\gamma}$ we get
\[
	g^{\beta\gamma} \hat C_{\beta\gamma\delta} = (n-2) (\div_1\hat T)_\delta \, .
\]
Since $\hat C$ is totally trace-free, as it is well-known, we deduce $\div_1\hat T = 0$. In view of this, the previous equality simplifies to
\begin{equation} \label{Har_eq2}
	\hat C_{\beta\gamma\delta} = \hat T_{\beta\gamma,\delta} - \hat T_{\beta\delta,\gamma} + \frac{1}{n-1} (\hat T^\eta_{\ \eta,\gamma} g_{\beta\delta} - \hat T^\eta_{\ \eta,\delta} g_{\beta\gamma}) \, .
\end{equation}
Using equation \eqref{Har_eq2} it is also easy to show how the solutions of \eqref{Ein_eq} also solve \eqref{Har_eq}. Indeed, tracing \eqref{Ein_eq} one obtains $(n-2) \hat S = -2\tr_{\hat g}\hat T$ and therefore
\begin{align*}
	\hat\Ric & = \hat T + \frac{1}{2} \hat S \hat g = \hat T - \frac{1}{n-2} (\tr_{\hat g}\hat T) \hat g \, , \\
	\hat A & = \hat\Ric - \frac{1}{2(n-1)} \hat S\hat g = \hat T - \frac{1}{n-1} (\tr_{\hat g}\hat T) \hat g \, ,
\end{align*}
yielding
\[
	\hat C_{\beta\gamma\delta} = \hat A_{\beta\gamma,\delta} - \hat A_{\beta\delta,\gamma} = \hat T_{\beta\gamma,\delta} - \hat T_{\beta\delta,\gamma} + \frac{1}{n-1} (\hat T^\eta_{\ \eta,\gamma} g_{\beta\delta} - \hat T^\eta_{\ \eta,\delta} g_{\beta\gamma}).
\]
Summarizing, we have proved the following
\begin{proposition}\label{Prop: cotton gravity: conseguenze base delle eq. di campo}
Given an $n$-dimensional Lorentzian manifold $(\hat{M}, \hat{g}),\ n\geq 3$, and a 2-covariant, symmetric tensor $\hat{T}$ on $\hat{M}$, consider the Cotton Gravity field equations
\begin{equation} \label{Har_eq bis}
	\hat C = -(n-2)\div_1\hat{\mathcal T} \, .
\end{equation}
with $\hat{\mathcal T}$ given by \eqref{Tensore T maiuscolo}; then we have the following:
\begin{itemize}
	\item[1.] if $(\hat{M},\hat{g},\hat{T})$ solves \eqref{Har_eq bis}, then $\div_1T=0$;
	\item[2.] component-wise, \eqref{Har_eq bis} is equivalent to
	\begin{equation} \label{Har_eq2 bis}
		\hat C_{\beta\gamma\delta} = \hat T_{\beta\gamma,\delta} - \hat T_{\beta\delta,\gamma} + \frac{1}{n-1} (\hat T^\eta_{\ \eta,\gamma} g_{\beta\delta} - \hat T^\eta_{\ \eta,\delta} g_{\beta\gamma}) \, ;
	\end{equation}
	\item[3.] if $(\hat{M},\hat{g},\hat{T})$ solves Einstein field equation \eqref{Ein_eq}, then it also solves \eqref{Har_eq bis}.
\end{itemize}
\end{proposition}
Our goal is to study the solutions of \eqref{Har_eq} for a stress-energy tensor that encodes both the informations of a perfect fluid and those of a non-linear field $\hat{\varphi}$ with a self-interacting potential $U$, that is, we choose
\[
	\hat T = \hat T^{\hat{\p}} + \hat T^F
\]
for two tensors $\hat T^{\hat{\p}}$, $\hat T^F$ that we are going to introduce. Let ${\hat{\p}} : (\hat M,\hat g) \to (N,h)$ be a smooth map, where $(N,h)$ is a Riemannian manifold. Let $\alpha\neq0$ be a real coupling constant and let $U\in C^\infty(N)$. Then
\begin{equation}\label{T phi def}
	\hat T^{\hat{\p}} := \alpha {\hat{\p}}^\ast h - \left(\frac{\alpha}{2}|\di{\hat{\p}}|^2 + U({\hat{\p}})\right)\hat g
\end{equation}
where ${\hat{\p}}^\ast h$ denotes the pullback of $h$ along ${\hat{\p}}$ and $|\di{\hat{\p}}|^2 := \tr_{\hat g}{\hat{\p}}^\ast h$ is twice the energy density of the map $\hat{\p}$. To introduce $\hat T^F$ we need the following definition.
\begin{definition}\label{Def: spaziotempo statico}
 Given a Lorentzian manifold $(\hat M,\hat g)$ we say that $\hat M$ is a \emph{static space-time} if it satisfies the following two conditions:
\begin{enumerate}
	\item as a differentiable manifold, $\hat M$ splits as the product $\hat M = M\times\RR$;
	\item there exists a Riemannian metric $g$ on $M$ and a function $f\in C^\infty(M)$ such that
	\[
		\hat g = -e^{-2\hat{f}} \di t\otimes\di t + g
	\]
	where $t : \hat M \to \RR$ is the canonical projection onto the second factor and $\hat{f}:=f\circ \hat{\pi}_M$. 
\end{enumerate}
\end{definition}
We will use the notation $\hat{M}=M\times_f \RR$ to denote the structure of static space-time on $(\hat{M}, \hat{g})$.  In this context, given two smooth functions $\mu,p\in C^\infty(M)$, that will be respectively called \emph{density} and \emph{pressure} of the fluid, we define
\begin{equation}\label{T F def}
\hat T^F := (\hat{\mu}+\hat{p}) e^{-2\hat{f}} \di t\otimes\di t + \hat{p} \hat g  ,
\end{equation}
where $\hat{\mu}:=\mu\circ \hat{\pi}_M$ and $\hat{p}:=p\circ \hat{\pi}_M$.
In the present paper we study solutions of the Cotton gravity equations $\hat C = -(n-2)\div_1\hat{\mathcal T}$ on a static space-time for a stress energy tensor of the type $\hat T = \hat T^{\hat{\p}} + \hat T^F$.
\subsection{Variational origin of the equations: the case of GR}\label{subset: variazioni: GR}

As it is well-known, the stress energy tensor $\hat T^F$ of a perfect fluid space-time is not derived from a Lagrangian so that it does not have a (clear) variational origin. A derivation of $\hat T^F$ which is variational, but only in a relaxed sense, is given in \cite[Chapter 3]{Hawking}, where one varies an appropriate functional depending on $\hat{\mu}$ with respect to the flow lines of the vector field $e^{-\hat{f}}\frac{\partial}{\partial t}$. By contrast, the situation is simpler for $\hat T^{\hat{\p}}$, since the equation
\[
	\hat\Ric - \frac{1}{2}\hat S\hat g = \hat T^{\hat{\p}}
\]
is the Euler-Lagrange equation of the functional
\[
	I[\hat{g}] = \int \hat S - \alpha|\di{\hat{\p}}|^2 - 2U({\hat{\p}}) \, 
\]
with respect to compactly supported variations of the metric $\hat{g}$. Here and in the following, integration is always implicitly assumed to be with respect to the natural volume form $\di \mu_{\hat{g}}$ induced by $\hat{g}$.
Similarly, in the regime of Cotton gravity, the equation
\begin{equation} \label{Har_phi}
	\hat C = -(n-2) \div_1 \hat{\mathcal T}^{\hat{\p}}
\end{equation}
has a variational origin. This is what we are going to show in the next subsection, without assuming that $(\hat{M},\hat{g})$ be a static space-time.
We first need to reformulate \eqref{Har_phi} in a more manageable form.
\begin{proposition}\label{prop: cotton gravity: C U phi}
	Consider an $n$-dimensional Lorentzian manifold $(\hat{M,\hat{g}})$, $n\geq 3$, a smooth map $\hat{\p}:(\hat{M}, \hat{g})\to (N,h)$, where $(N,h)$ is a Riemannian manifold, a constant $\alpha\in\RR$ and a smooth function $U\in C^\infty(N)$. Define a 3-covariant tensor $\hat{C}^{U,\hat{\p}}$ on $\hat{M}$ by the expression
	\begin{align*}
		(\hat C^{U,{\hat{\p}}})_{\beta\gamma\delta} & := \hat C_{\beta\gamma\delta} - \alpha\left({\hat{\p}}^a_{\beta\delta}{\hat{\p}}^a_\gamma-{\hat{\p}}^a_{\beta\gamma}{\hat{\p}}^a_\delta\right) + \frac{\alpha}{n-1} g^{\rho\eta}{\hat{\p}}^a_\rho \left( {\hat{\p}}^a_{\eta\delta} g_{\beta\gamma} - {\hat{\p}}^a_{\eta\gamma} g_{\beta\delta} \right) \\
		& \phantom{:=\;} - \frac{1}{n-1} U^a{\hat{\p}}^a_\delta g_{\beta\gamma} + \frac{1}{n-1} U^a{\hat{\p}}^a_\gamma g_{\beta\delta}.
	\end{align*}
	Then, \eqref{Har_phi} is equivalent to
	\begin{equation} \label{Cot_Uphi}
		\hat C^{U,{\hat{\p}}} = 0 \, .
	\end{equation}
	
\end{proposition}
\begin{proof}
 By equation \eqref{Har_eq2 bis}, \eqref{Har_phi} is equivalent to
\begin{equation} \label{Cot_Tphi}
	\hat C_{\beta\gamma\delta} = - \left( \hat T^{\hat{\p}}_{\beta\delta,\gamma} - \hat T^{\hat{\p}}_{\beta\gamma,\delta} - \frac{1}{n-1} (\tr_{\hat g}\hat T^{\hat{\p}})_{\gamma} g_{\beta\delta} + \frac{1}{n-1} (\tr_{\hat g}\hat T^{\hat{\p}})_{\delta} g_{\beta\gamma} \right) .
\end{equation}
From the definition of $\hat T^{\hat{\p}}$ we get
\begin{equation}\label{Deduzione: derivata di T phi}
	\hat T^{\hat{\p}}_{\beta\delta,\gamma} = \alpha\left[ {\hat{\p}}^a_{\beta\gamma}{\hat{\p}}^a_\delta + {\hat{\p}}^a_\beta{\hat{\p}}^a_{\delta\gamma} - \left(g^{\rho\eta} {\hat{\p}}^a_{\rho\gamma}{\hat{\p}}^a_\eta + \frac{1}{\alpha}U^a{\hat{\p}}^a_\gamma\right) g_{\beta\delta} \right]
\end{equation}
and
\begin{equation}\label{Deduzione: traccia di T phi}
	\tr_{\hat g}\hat T^{\hat{\p}} = -\frac{n-2}{2} \alpha|\di{\hat{\p}}|^2 - nU({\hat{\p}})
\end{equation}
so that
\begin{equation}\label{Deduzione: der cov di traccia di T phi}
	(\tr_{\hat g}\hat T^{\hat{\p}})_{\gamma} = -\alpha(n-2)g^{\rho\eta}{\hat{\p}}^a_{\rho\gamma}{\hat{\p}}^a_\eta - nU^a{\hat{\p}}^a_\gamma \, .
\end{equation}
Using \eqref{Deduzione: derivata di T phi} and \eqref{Deduzione: der cov di traccia di T phi} in \eqref{Cot_Tphi} we immediately get \eqref{Cot_Uphi}.
\end{proof}
\begin{remark}
Note that, defining
\begin{equation} \label{A_Uphi}
	\hat A^{U,{\hat{\p}}} := \hat A - \alpha{\hat{\p}}^\ast h + \frac{\alpha}{2(n-1)} |\di{\hat{\p}}|^2 \hat g - \frac{U({\hat{\p}})}{n-1}\hat g
\end{equation}
we have that $\hat C^{U,{\hat{\p}}}$ measures the obstruction to $\hat A^{U,{\hat{\p}}}$ being a Codazzi tensor. Moreover, the conservation equation
\[
	0 = \div_1\hat T^{\hat{\p}} \, ,
\]
that is,
\[
	0 = \alpha g^{\gamma\eta}{\hat{\p}}^a_{\gamma\eta}{\hat{\p}}^a_\beta - U^a {\hat{\p}}^a_\beta
\]
is always implied by \eqref{Cot_Tphi},  as we already saw, and therefore also by \eqref{Cot_Uphi}, as it can be verified by direct computation.  Indeed, one has
\begin{equation} \label{phi_eq}
	g^{\beta\gamma} (\hat C^{U,{\hat{\p}}})_{\beta\gamma\delta} = \alpha g^{\beta\gamma} {\hat{\p}}^a_{\beta\gamma} {\hat{\p}}^a_\delta - U^a {\hat{\p}}^a_\delta \, .
\end{equation}
\end{remark}
\subsection{Variational origin of the equations: the case of Cotton Gravity}\label{subset: origine variazionale}
We are now ready to give a derivation of equation  \eqref{Har_phi}, that is,
\begin{equation*} 
	\hat C = -(n-2) \div_1 \hat{\mathcal T}^{\hat{\p}}.
\end{equation*}
 Following \cite{Harada}, we are interested in variations where the connection $\hat{\nabla}$ is varying while the metric $\hat g$ is fixed. To be more precise, recall that the difference between two connections on $T\hat M$ is a tensor, in particular, a section of $T^{(1,2)}(\hat M)$, the space of $(1,2)$-tensors
\[
	B = B^\eta_{\ \beta\gamma} e_\eta \otimes \theta^\beta \otimes \theta^\gamma.
\]
 Therefore, the space of connections on $T\hat M$ is an affine space over $T^{(1,2)}(\hat M)$. Having fixed a section $B$ of $T^{(1,2)}(\hat M)$, we define $\hat{\nabla}^t = \hat{\nabla} + tB$. In any arbitrary coordinate system $(U,\{x^\eta\}_{\eta=0}^{n-1})$ we have that the variation of the Christoffel symbols of $\hat{\nabla}$ in the direction of $B$ is
\[
	\delta\Gamma^\eta_{\ \beta\gamma} \frac{\partial}{\partial x^\eta} = \lim_{t\to0} \frac{1}{t} (\hat{\nabla}^t_{\partial x^\beta}\frac{\partial}{\partial x^\eta} - \hat{\nabla}_{\partial x^\beta}\frac{\partial}{\partial x^\eta}) = B^\eta_{\ \beta\gamma} \frac{\partial}{\partial x^\eta} \, .
\]
Given a $2$-covariant tensor $T$, define
\begin{equation} \label{S2T}
	S_2(T) := \frac{1}{2} \left[(\tr_{\hat g}T)^2 - T^{\eta\beta} T_{\eta\beta}\right] \, .
\end{equation}
To derive \eqref{Har_phi}, we consider the action functional
\[
	\mathcal S_{{\hat{\p}}}(\hat g) = \int S_2(\hat A^{\hat{\p}}) - \frac{\alpha}{2}|\tau({\hat{\p}})|^2
\]
where $\tau({\hat{\p}})$ is the tension field of the map ${\hat{\p}}$, given in components by
\[
	\tau({\hat{\p}})^a = g^{\eta\beta} \frac{\partial^2{\hat{\p}}^a}{\partial x^\eta\partial x^\beta} - g^{\eta\beta}\Gamma^\gamma_{\ \eta\beta} \frac{\partial{\hat{\p}}^a}{\partial x^\gamma} + {}^N \Gamma^a_{\ bc} g^{\eta\beta} \frac{\partial{\hat{\p}}^b}{\partial x^\eta} \frac{\partial{\hat{\p}}^c}{\partial x^\beta}
\]
where ${}^N\Gamma^a_{\ bc}$ denote the Christoffel symbols of the Levi-Civita connection of $(N,h)$ and $\hat A^{\hat{\p}}$ is the ${\hat{\p}}$-Schouten tensor, defined by
\[
	\hat A^{\hat{\p}} = \hat A - \alpha\,{\hat{\p}}^\ast h + \frac{\alpha}{2(n-1)} |\di{\hat{\p}}|^2 \hat g \, .
\]
More informations on the ${\hat{\p}}$-curvature tensors will be given in Section \ref{Sect 2: phi-curv}. Here, we recall that, having set
\[
	\hat\Ric^{\hat{\p}} := \hat\Ric - \alpha\,{\hat{\p}}^\ast h \qquad \text{and} \qquad \hat S^{\hat{\p}} := \hat S - \alpha\,|\di{\hat{\p}}|^2
\]
we have the validity of the ${\hat{\p}}$-Schur's identity, that is,
\[
	g^{\beta\gamma}\hat R^{\hat{\p}}_{\eta\beta,\gamma} = \frac{1}{2} \hat S^{\hat{\p}}_\eta - \alpha \tau({\hat{\p}})^a {\hat{\p}}^a_\eta \, .
\]
The functional $\mathcal S_{\hat{\p}}$ has been introduced in \cite{Anselli_2021} where it is shown that, in Riemannian signature and for $n=4$, the critical points of $\mathcal S_{\hat{\p}}$ with respect to variations of the metric have vanishing ${\hat{\p}}$-Bach tensor $\hat{B}^{\hat{\p}}$ (see Section \ref{Sect 2: phi-curv} for the definition of $\hat{B}^{\hat{\p}}$). Here we modify $\mathcal S_{\hat{\p}}$ in order to take into account the contribution of the potential $U$. From \eqref{A_Uphi} we have
	$\hat A^{U,{\hat{\p}}} = \hat A^{\hat{\p}} - \frac{U({\hat{\p}})}{n-1} \hat g$ .
We will study the functional
\begin{equation} \label{S_Uphi}
	\mathcal S_{U,{\hat{\p}}}(\hat g) = \int S_2(\hat A^{U,{\hat{\p}}}) - \frac{\alpha}{2} \left|\tau({\hat{\p}}) - \frac{\nabla^h U}{\alpha}\right|^2_h \, .
\end{equation}
When $n=4$, $\p$ is constant and $U\equiv 0$, Harada proved that the Cotton flat metrics on $\hat{M}$ are the critical points of $S_{U,\hat{\p}}$ with respect to compactly supported variations of $\hat{\nabla}$ that leave $\hat{g}$ fixed, see \cite{Harada}. More generally, we have the following
\begin{proposition}\label{variazioni: proposizione}
	Let $(\hat M,\hat g)$ be an $n$-dimensional Lorentzian manifold, $n\geq 3$. Then, equation \eqref{Har_phi} is satisfied by $(\hat M,\hat g)$ if and only if $(\hat M,\hat g)$ is a critical point of the action $\mathcal S_{U,{\hat{\p}}}$ defined in \eqref{S_Uphi} with respect to compactly supported variations of $\hat{\nabla}$ that leave $\hat g$ fixed.
\end{proposition}

\begin{proof}
	By Proposition \ref{prop: cotton gravity: C U phi}, we only need to prove that the critical points of the given action satisfy \eqref{Cot_Uphi}.
	As it is well-known, the Riemann curvature tensor can be expressed using the Christoffel symbols by
	\[
		\hat R^\eta_{\ \beta\gamma\rho} = \partial_\gamma \Gamma^\eta_{\ \beta\rho} - \partial_\rho \Gamma^\eta_{\ \beta\gamma} + \Gamma^\xi_{\ \beta\rho} \Gamma^\eta_{\ \gamma\xi} - \Gamma^\xi_{\ \beta\gamma} \Gamma^\eta_{\ \rho\xi}
	\]
	so that its variation is
	\[
		\delta\hat R^\eta_{\ \beta\gamma\rho} = \partial_\gamma B^\eta_{\ \beta\rho} - \partial_\rho B^\eta_{\ \beta\gamma} + B^\xi_{\ \beta\rho} \Gamma^\eta_{\ \gamma\xi} + \Gamma^\xi_{\ \beta\rho} B^\eta_{\ \gamma\xi} - B^\xi_{\ \beta\gamma} \Gamma^\eta_{\ \rho\xi} - \Gamma^\xi_{\ \beta\gamma} B^\eta_{\ \rho\xi}
	\]
	where $B^\eta_{\ \beta\gamma} = \delta\Gamma^\eta_{\ \beta\gamma}$ is the variation of the Christoffel symbols. Recalling the expression for the covariant derivative of $B$
	\[
		B^\eta_{\ \beta\rho,\gamma} = \partial_\gamma B^\eta_{\ \beta\rho} + B^\xi_{\ \beta\rho} \Gamma^\eta_{\ \xi\gamma} - B^\eta_{\ \xi\rho} \Gamma^\xi_{\ \beta\gamma} - B^\eta_{\ \beta\xi} \Gamma^\xi_{\ \rho\gamma}
	\]
	the above equation yields
	\[
		\delta\hat R^\eta_{\ \beta\gamma\rho} = B^\eta_{\ \beta\rho,\gamma} - B^\eta_{\ \beta\gamma,\rho} \, .
	\]
	Lowering an index and skew-symmetrizing we get
	\[
		\delta\hat R_{\eta\beta\gamma\rho} = \frac{1}{2} ( B_{\eta\beta\rho,\gamma} - B_{\beta\eta\rho,\gamma} - B_{\eta\beta\gamma,\rho} + B_{\beta\eta\gamma,\rho} ) .
	\]
	Therefore, the variation of $\hat\Ric$ is
	\begin{equation} \label{var_Ric}
		\delta\hat R_{\eta\beta} = \frac{1}{2} (B^{\;\,\rho}_{\eta\;\rho,\beta} - B^{\rho}_{\;\,\eta\rho,\beta} - B^{\;\,\rho}_{\eta\;\beta,\rho} + B^\rho_{\;\,\eta\beta,\rho})
	\end{equation}
	and that of $\hat S$ is
	\begin{equation} \label{var_S}
		\delta\hat S = \frac{1}{2} (B^{\eta\beta}_{\;\;\;\;\beta,\eta} - B^{\eta\beta}_{\;\;\;\;\eta,\beta} - B^{\eta\beta}_{\;\;\;\;\eta,\beta} + B^{\eta\beta}_{\;\;\;\;\beta,\eta}) = B^{\eta\beta}_{\;\;\;\;\beta,\eta} - B^{\eta\beta}_{\;\;\;\;\eta,\beta} \, .
	\end{equation}
	Next, note that, since ${\hat{\p}}^\ast h$, $|\di{\hat{\p}}|^2$ and $U({\hat{\p}})$ are independent from $\hat{\nabla}$, their variation is zero, so that $\delta\hat A^{U,{\hat{\p}}} = \delta\hat A$.
	By the definition of the tension field $\tau({\hat{\p}})$ one immediately gets
	\begin{equation}
		\delta\tau({\hat{\p}})^a = - g^{\eta\beta} B^\gamma_{\ \eta\beta} \frac{\partial{\hat{\p}}^a}{\partial x^\gamma} = - B^{\eta\beta}_{\;\;\;\;\beta} \frac{\partial{\hat{\p}}^a}{\partial x^\eta}
	\end{equation}
	and therefore
	\begin{equation}\label{Variazione: variazione di tau p - nabla U quadro}
		\begin{split}
			\delta\left|\tau({\hat{\p}})-\frac{\nabla^h U}{\alpha}\right|^2_h & = 2 \left(\delta\tau({\hat{\p}})^a - \frac{1}{\alpha}\delta U^a\right)\left(\tau({\hat{\p}})^a - \frac{1}{\alpha}U^a\right) \\
			& = - 2 B^{\eta\beta}_{\;\;\;\;\beta} \frac{\partial{\hat{\p}}^a}{\partial x^\eta} \left(\tau({\hat{\p}})^a - \frac{1}{\alpha}U^a\right) .
		\end{split}
	\end{equation}
	By the definition of $\hat A^{U,{\hat{\p}}}$ we compute
	\begin{align} \label{tr_AUp}
		\tr_{\hat g} \hat A^{U,{\hat{\p}}} = \tr_{\hat g} \hat A^{\hat{\p}} - \frac{n}{n-1} U({\hat{\p}}) = \frac{n-2}{2(n-1)} \hat S^{\hat{\p}} - \frac{n}{n-1} U({\hat{\p}}) \, .
	\end{align}
	Using \eqref{tr_AUp} we get
	\begin{align*}
		\delta S_2(\hat A^{U,{\hat{\p}}}) & = (\tr_{\hat g} \hat A^{U,{\hat{\p}}}) \delta(\tr_{\hat g} \hat A^{U,{\hat{\p}}}) - (\hat A^{U,{\hat{\p}}})^{\beta\gamma} (\delta\hat A^{U,{\hat{\p}}})_{\beta\gamma} \\
		& = (\tr_{\hat g} \hat A^{U,{\hat{\p}}}) \delta(\tr_{\hat g} \hat A) - (\hat A^{U,{\hat{\p}}})^{\beta\gamma} (\delta\hat A)_{\beta\gamma} \\
		& = \frac{n-2}{2(n-1)} (\tr_{\hat g} \hat A^{U,{\hat{\p}}}) \delta\hat S - (\hat A^{U,{\hat{\p}}})^{\beta\gamma} (\delta\hat\Ric)_{\beta\gamma} + \frac{1}{2(n-1)}(\hat A^{U,{\hat{\p}}})^{\beta\gamma} g_{\beta\gamma} \delta\hat S \\
		& = \frac{1}{2}(\tr_{\hat g} \hat A^{U,{\hat{\p}}}) \delta\hat S - (\hat A^{U,{\hat{\p}}})^{\beta\gamma} \delta\hat R_{\beta\gamma} \, .
	\end{align*}
	Using \eqref{var_Ric} and \eqref{var_S} we find
	\[
		\delta S_2(\hat A^{U,{\hat{\p}}}) = \frac{1}{2}(\tr_{\hat g} \hat A^{U,{\hat{\p}}})(B^{\eta\beta}_{\;\;\;\;\beta,\eta} - B^{\eta\beta}_{\;\;\;\;\eta,\beta}) - \frac{1}{2} (\hat A^{U,{\hat{\p}}})^{\eta\beta} (B^{\;\,\rho}_{\eta\;\rho,\beta} - B^{\rho}_{\;\,\eta\rho,\beta} - B^{\;\,\rho}_{\eta\;\beta,\rho} + B^\rho_{\;\,\eta\beta,\rho}) \, .
	\]
	Integrating and using the divergence theorem we get
	\begin{align*}
		\int \delta S_2(\hat A^{U,{\hat{\p}}}) & = \frac{1}{2} \int \left[(\tr_{\hat g} \hat A^{U,{\hat{\p}}})_{\beta} B^{\eta\beta}_{\;\;\;\;\eta} - (\tr_{\hat g} \hat A^{U,{\hat{\p}}})_{\eta} B^{\eta\beta}_{\;\;\;\;\beta}\right] \\
		& \phantom{=\;} + \frac{1}{2} \int (\hat A^{U,{\hat{\p}}})^{\;\;\beta}_{\eta\;\,,\beta} (B^{\eta\rho}_{\;\;\;\;\rho} - B^{\rho\eta}_{\;\;\;\;\rho}) - \frac{1}{2} \int (\hat A^{U,{\hat{\p}}})^{\eta\beta}_{\;\;\;\;,\rho} (B^{\;\,\rho}_{\eta\;\beta} - B^\rho_{\;\,\eta\beta}) \, .
	\end{align*}
	Using the ${\hat{\p}}$-Schur's identity and \eqref{tr_AUp} we have
	\begin{align*}
		(\hat A^{U,{\hat{\p}}})^{\;\;\beta}_{\eta\;\,,\beta} & = (\hat R^{\hat{\p}})^{\;\;\beta}_{\eta\;\,,\beta} - \frac{1}{2(n-1)} \hat S^{\hat{\p}}_\eta - \frac{1}{n-1} U^a {\hat{\p}}^a_\eta \\
		& = \frac{n-2}{2(n-1)} \hat S^{\hat{\p}}_\eta - \alpha \tau({\hat{\p}})^a {\hat{\p}}^a_\eta - \frac{1}{n-1} U^a {\hat{\p}}^a_\eta \\
		& = (\tr_{\hat g} \hat A^{U,{\hat{\p}}})_{\eta} - \alpha \left(\tau({\hat{\p}})^a - \frac{1}{\alpha}U^a\right){\hat{\p}}^a_\eta
	\end{align*}
	and substituting the latter into the above integral identity we obtain
	\begin{align*}
		\int \delta S_2(\hat A^{U,{\hat{\p}}}) & = - \frac{1}{2} \int 
		(\hat{A}^{U,{\hat{\p}}})^{\eta\beta}_{\ \ ,\rho} (B_{\eta \ \beta}^{\ \rho}-B^{ \rho}_{\ \eta\beta}) - \frac{\alpha}{2} \int \pa{\tau({\hat{\p}})^a-\frac{1}{\alpha}U^a}{\hat{\p}}^a_\eta\pa{B^{\eta\rho}_{\ \ \rho}-B^{\rho\eta}_{\ \ \rho}}.
	\end{align*}
	Renaming indexes we get
	\begin{align*}
		\int \delta S_2 (\hat{A}^{U,{\hat{\p}}}) & = \frac{1}{2}\int ((\hat{A}^{U,{\hat{\p}}})_{\beta\eta,\rho}-(\hat{A}^{U,{\hat{\p}}})_{\beta\rho,\eta})B^{\rho\eta\beta} \\
		& \phantom{=\;} + \frac{\alpha}{2} \int \pa{\tau({\hat{\p}})^a-\frac{1}{\alpha}U^a}\pa{{\hat{\p}}^a_\eta g_{\beta\rho}-{\hat{\p}}^a_\rho g_{\beta\eta}}B^{\rho\eta\beta}.
	\end{align*}
	Recalling that, by definition, $\hat{C}^{U,{\hat{\p}}}_{\beta\eta\rho}=\hat{A}^{U,{\hat{\p}}}_{\beta\eta,\rho}-\hat{A}^{U,{\hat{\p}}}_{\beta\rho,\eta}$
	we get
	\begin{align}\label{Variazione: variazione di S 2 A U phi}
		\int \delta S_2(\hat{A}^{U,{\hat{\p}}}) & = \frac{1}{2}\int 
		\sq{\hat{C}^{U,{\hat{\p}}}_{\beta\eta\rho} + \alpha \pa{\tau({\hat{\p}})^a-\frac{1}{\alpha}U^a}\pa{{\hat{\p}}^a_\eta g_{\beta\rho}-{\hat{\p}}^a_\rho g_{\eta \beta}}}B^{\rho\eta\beta}.
	\end{align}
	Therefore, computing the variation of \eqref{S_Uphi} and using \eqref{Variazione: variazione di S 2 A U phi} and \eqref{Variazione: variazione di tau p - nabla U quadro} we find
	\begin{align*}
		\delta \mathcal{S}_{U,{\hat{\p}}} & = \int \delta S_2(\hat{A}^{U,{\hat{\p}}})-\frac{\alpha}{2}\delta \pa{\abs{\tau({\hat{\p}})-\frac{1}{\alpha}\nabla U}^2_h}\\
		& = \frac{1}{2} \int \sq{\hat{C}^{U,{\hat{\p}}}_{\beta\eta\rho}+\alpha \pa{\tau({\hat{\p}})^a-\frac{1}{\alpha}U^a}\pa{{\hat{\p}}^a_\eta g_{\beta\rho}-{\hat{\p}}^a_\rho g_{\eta \beta}}}B^{\rho\eta\beta} \\
		& \phantom{=\;} + \frac{\alpha}{2} \int \pa{\tau({\hat{\p}})^a-\frac{1}{\alpha}U^a}{\hat{\p}}^a_\rho g_{\eta\beta}B^{\rho\eta\beta}\\
		& = \frac{1}{2}\int \sq{\hat{C}^{U,{\hat{\p}}}_{\beta\eta\rho}+\alpha \pa{\tau({\hat{\p}})^a-\frac{1}{\alpha}U^a}\pa{{\hat{\p}}^a_\eta g_{\beta\rho}+{\hat{\p}}^a_\rho g_{\eta \beta}}}B^{\rho\eta\beta}.
	\end{align*}
	Therefore, the compactly supported variations of $\mathcal{S}_{U,{\hat{\p}}}$ with respect to $\hat{\nabla}$ satisfy
	\begin{equation}\label{Variazione: quasi equazione critica}
		0 = \hat{C}^{U,{\hat{\p}}}_{\rho\beta\eta}+\alpha \pa{\tau({\hat{\p}})^a-\frac{1}{\alpha}U^a}\pa{{\hat{\p}}^a_\eta g_{\beta\rho}+{\hat{\p}}^a_{\beta}g_{\eta\rho}}.
	\end{equation}
	Contracting \eqref{Variazione: quasi equazione critica} with respect to $g^{\rho\beta}$ and using \eqref{phi_eq} we deduce
	\begin{align*}
		0=\alpha \pa{\tau({\hat{\p}})^a-\frac{1}{\alpha}U^a}{\hat{\p}}^a_\eta+\alpha(n+1)\pa{\tau({\hat{\p}})^a-\frac{1}{\alpha}U^a}{\hat{\p}}^a_\eta
	\end{align*}
	which implies
	\begin{align*}
		\tau({\hat{\p}})^a{\hat{\p}}^a_\eta-\frac{1}{\alpha}U^a{\hat{\p}}^a_\eta=0
	\end{align*}
	so that \eqref{Variazione: quasi equazione critica} implies \eqref{Cot_Uphi} and we are done.
\end{proof}

\section{$\varphi$-curvatures}\label{Sect 2: phi-curv}

As we shall see, it is useful to encode some informations on the non-linear field $\p$ into the curvature tensors of $\pa{M,g}$ in order to see, at the same time,  the combined action of $\p$ and that of the Riemannian metric $g$ of $M$.
With these motivations (other can be found in \cite{ACR},\cite{CMR2022},\cite{MR}) we introduce modified curvature tensors depending on the map $\p: \pa{M,g}\to \pa{N,h}$. The first step in this direction, that is, the definition of the $\p$-Ricci tensor, is due to B. List, that merged the Ricci flow with the harmonic map flow (for details and background see \cite{List2008EvolutionOA}). For some fixed coupling constant $\alpha \neq 0$ we set
\begin{align}
	\Ric^\p = \Ric-\alpha \p^*h
\end{align}
for the $\p$-Ricci tensor and the $\p$-scalar curvature will be its contraction with the metric $g$, that is,
\begin{align}
	S^\p=S-\alpha \abs{\di \p}^2.
\end{align}
In components (that is, with respect to orthonormal coframes ($U, \theta^i$),($V,\omega^a$) 
with $\p(U)\subset V$, $1\leq i,j,...\leq m=\dim M$, $1\leq a,b,...\leq d=\dim N$), we have
\begin{align}
	R^\p_{ij}=R_{ij}-\alpha \p^a_i\p^a_j
\end{align}
(note that List uses the notation $S_{ij}$ instead of $R^\p_{ij}$).
The next formula (see \cite{ACR} for a simple proof) should be called the $\p$-Schur's identity and it will be repeatedly used in the sequel
\begin{align}\label{phi curv: phi Schur}
	R^\p_{ij,i}=\frac{1}{2}S^\p_j-\alpha \p^a_{ii}\p^a_j.
\end{align} 
From \eqref{phi curv: phi Schur} we immediately infer that, for $m\geq 3$, if
\begin{align}\label{phi curv: ric phi proporz metrica}
	\Ric^\p=\Lambda g, \ \ \ \ \ \Lambda\in C^\infty(M),
\end{align}
then the function $\Lambda$ satisfies
\begin{align}
	(m-2)\nabla \Lambda=2\alpha h(\tau(\p),\di \p)^\sharp.
\end{align}
Here, $\tau(\p)$, with components $\p^a_{ll}$, is the tension field of the map $\p$, see \cite{EL}. In case $\p$ is conservative, that is, the divergence of the stress-energy tensor of $\p$ is null, equivalently,
\begin{align*}
	h(\tau(\p),\di \p)\equiv 0,
\end{align*}
then $\Lambda$ is constant. This is, in particular, the case of a harmonic map, for which $\tau(\p)\equiv 0$. Thus, if $\pa{M,g}$ is harmonic-Einstein with respect to $\p$ for some $\alpha \neq 0$, that is, the system
\begin{align}\label{phi curv: harmonic Einstein}
	\begin{cases}
		\Ric^\p=\Lambda g,\\
		\tau(\p)=0
	\end{cases}
\end{align}
holds on $M$, then $\Lambda$ is constant, or in other terms the $\p$-scalar curvature $S^\p$ is constant. This also suggests the use of the term $\p$-Einstein manifold for $\pa{M,g}$, but the previous terminology has been introduced before and we shall adhere to it.
In analogy with the classical case, the $\p$-Schouten tensor is defined by setting
\begin{align}
	A^\p = \Ric^\p-\frac{S^\p}{2(m-1)}g
\end{align}
and the $\p$-Cotton tensor as the obstruction to $A^\p$ to be Codazzi.
The $\p$-Weyl tensor is introduced with the care to formally respect the usual decomposition of the Riemann tensor, that is,
\begin{align}\label{phi curv: phi weyl}
	W^\p=\Riem-\frac{1}{m-2}A^\p\KN g
\end{align}
where $\KN$ is the Kulkarni-Nomizu product. Although $W^\p$ has the same ``algebraic'' symmetries of the Riemann curvature tensor it is not totally trace free. Indeed
\begin{align}\label{phi curv: traccia di W phi}
	W^\p_{kikj}=\alpha\p^a_i\p^a_j
\end{align}
while the remaining traces are determined by \eqref{phi curv: phi weyl} and by the algebraic symmetries of $W^\p$. Note that, in terms of the classical counterparts, we have
\begin{align}
	A^\p=A-\alpha A(\p^*h),
\end{align}
\begin{align}\label{phi curv: W phi e W}
	W^\p=W+\frac{\alpha}{m-2}A(\p^*h)\KN g
\end{align}
where
\begin{align*}
	A(\p^*h)=\p^*h-\frac{\abs{\di \p}^2}{2(m-1)}g
\end{align*}
is the ``Schouten tensor'' of the symmetric, 2-covariant tensor $\p^*h$. For the $\p$-Cotton tensor $C^\p$ we have the symmetries
\begin{align}
	C^\p_{ijk}=-C^\p_{ikj} \ \text{ so that } \ C^\p_{ikk}=0,
\end{align}
however it is not totally trace free, since
\begin{align}\label{phi curv: traccia di C phi e tensione}
	C^\p_{kki}=\alpha \p^a_{kk}\p^a_i.
\end{align}
It is easy to prove that $C^\p$ satisfies the Bianchi-type identity
\begin{align}
	C^\p_{ijk}+C^\p_{kji}+C^\p_{jki}=0.
\end{align}
Again, from its definition it is immediate to see that
\begin{align}
	C^\p_{ijk} = C_{ijk}-\alpha\sq{\p^a_{ik}\p^a_j-\p^a_{ij}\p^a_k-\frac{\p^a_l}{m-1}\pa{\p^a_{lk}\delta_{ij}-\p^a_{tj}\delta_{ik}}}.
\end{align}
The next two formulas are ``fake'' second Bianchi identities for the $\p$-Weyl tensor and the $\p$-Cotton tensor. They will not be used in the sequel, but we state them for completeness.
We have
\begin{align}
	W^\p_{ijkl,p} + W^\p_{ijpk,l} + W^\p_{ijlp,k} = \frac{1}{m-2}(C^\p_{ikl}\delta_{pj}+C^\p_{ilp}\delta_{jk}+C^\p_{ipk}\delta_{jl} - C^\p_{jlp}\delta_{ki}-C^\p_{jpk}\delta_{li}-C^\p_{jkl}\delta_{pi})\nonumber
\end{align}
and
\[
	C^\p_{ijk,l} + C^\p_{ikl,j}+C^\p_{ilj,k} = W^\p_{pilj}R^\p_{kp}+W^\p_{pijk}R^\p_{lp}+W^\p_{pikl}R^\p_{jp}.
\]
The next alternative definition of the $\p$-Cotton tensor for $m\geq 4$, points out at deep differences between the classical and the $\p$-curvatures (see Proposition 2.64 of \cite{ACR})
\begin{align}\label{phi curv: divergenza di W phi}
	-\frac{m-3}{m-2}C^\p_{jkl}=W^\p_{sjkl,s}-\alpha \pa{\p^a_{jk}\p^a_l-\p^a_{jl}\p^a_k}-\frac{\alpha}{m-2}\p^a_{ss}\pa{\p^a_k\delta_{jl}-\p^a_l\delta_{jk}}.
\end{align}
With the obvious meaning of the notation we shall set
\begin{align}
	\pa{\div_1 W^\p}_{jkl}=W^\p_{sjkl,s}.
\end{align}
This notation and its obvious extensions will be used throughout the paper.
The $\p$-Bach tensor, $B^\p$, is not defined in analogy with the classical one; indeed, for $m\geq 3$, we set
\begin{align}\label{phi curv: phi bach}
	(m-2)B^\p_{ij} & = C^\p_{ijk,k}+R^\p_{lk}W^\p_{likj}-\alpha R^\p_{lj}\p^a_l\p^a_i + \alpha\pa{\p^a_{ij}\p^a_{kk}-\p^a_{kkj}\p^a_i-\frac{\abs{\tau(\p)}^2}{m-2}\delta_{ij}}. \nonumber
\end{align}
Note that $B^\p$ is a symmetric, 2-covariant tensor; for a proof see \cite{ACR}. Contrary to the Bach tensor, $B^\p$ is not trace free, in general. Indeed we have
\begin{align}
	B^\p_{ii}=\alpha \frac{m-4}{(m-3)^2}\abs{\tau(\p)}^2.
\end{align}

\section{The system}\label{Sect 3: system}

In this section, we show how the equations of Cotton gravity, for the stress energy tensor $\hat{T}=\hat{T}^F+\hat{T}^{\hat{\p}}$ of the Introduction, reduce in the space and time components of a static space-time.
Motivated by physical reasons (see Remark \ref{Deduzione: rmk: U armonicità} below), we will split the natural condition $0=\div_1 \hat{T}$ into the two conditions $0=\div_1\hat{T}^F=\div_1\hat{T}^{\hat{\p}}$.
\subsection{Cotton-$\p$-Perfect Fluids}
In this section, we consider an $n$-dimensional static space-time $\hat{M}=M\times_f \RR$ that solves the Cotton Gravity equations \eqref{Har_eq}
 for a stress-energy tensor $\hat{T}=\hat{T}^F+\hat{T}^{\hat{\p}}$, where $\hat{T}^F$ and $\hat{T}^{\hat{\p}}$ are given by \eqref{T F def} and \eqref{T phi def}, respectively.
After a long, but routine, computation one sees that equation \eqref{Har_eq} can be reduced on the time-slice $(M,g)$ of the static space-time, giving rise to the system 
	\begin{subequations}\label{Deduzione: C-phi-PF versione 3}
	\begin{empheq}[left={\empheqlbrace}]{alignat=2}
		0 = & \; C^\p_{ijk}+f_{ijk}-f_{ikj}-f_{ik}f_j+f_{ij}f_k\notag\\
		&-\frac{1}{m}\sq{\pa{f_{llk}-2f_lf_{lk}}\delta_{ij}-\pa{f_{llj}-2f_lf_{lj}}\delta_{ik}}\notag\\
		&+\frac{1}{2m(m-1)}\pa{S^\p_k\delta_{ij}-S^\p_j\delta_{ik}}\notag\\
		&+\frac{1}{m}U^a\p^a_j\delta_{ik}-\frac{1}{m}U^a\p^a_k\delta_{ij}+\frac{1}{m}\pa{\mu_j\delta_{ik}-\mu_k\delta_{ij}},\label{Deduzione: C-phi-PF versione 3 1}\\
		0=&-\frac{m-1}{m}f_{lli}+\frac{m-2}{m}f_lf_{li}+\pa{\Delta f}f_i+\frac{1}{2m}S^\p_i-f_lR^\p_{il}\notag\\
		&+\frac{1}{m}U^a\p^a_i-\frac{m-1}{m}\mu_i,\label{Deduzione: C-phi-PF versione 3 2}\\
		0=& \; h\pa{\tau(\p)-\di\p(\nabla f)-\frac{1}{\alpha}\nabla^h U (\p),\di \p},\label{Deduzione: C-phi-PF versione 3 4} \\ 
		0=& \; \nabla p-(\mu+p)\nabla f \label{Deduzione: C-phi-PF versione 3 3}.
	\end{empheq}
\end{subequations}
Here $m=\dim M$, so that $n=m+1$. The deduction of system \eqref{Deduzione: C-phi-PF versione 3} is carried out in full details in Appendix \ref{Appendice A 1}.
\begin{remark}\label{Deduzione: rmk: U armonicità}
	In what follows, we will replace condition \eqref{Deduzione: C-phi-PF versione 3 4} with the stronger
	\begin{equation}\label{Deduzione: tau-dphi -U}
		\tau(\p)=\di\p(\nabla f)+\frac{1}{\alpha}\nabla^h U(\p).
	\end{equation}
	This is motivated by the following reasoning: the stress-energy tensor $\hat{T}^{\hat{\p}}$ is obtained from the matter action
	\begin{equation*}
		I(\hat{g},\hat{\p})=\int \pa{\abs{\di \hat{\p}}_{\hat{g}}^2+2U(\hat{\p})}
	\end{equation*} 
	by considering the critical points of its variations with respect to the metric. The critical points for variations of the field $\hat{\p}$ of this action satisfy
	\begin{equation}\label{Deduzione: U armonicità}
		\tau(\hat{\p})=\frac{1}{\alpha}\nabla^h U (\hat{\p})
	\end{equation}
	see \cite{Lemaire1977OnTE} for a proof in the Riemannian case, the Lorentzian case being identical. It is elementary to prove that \eqref{Deduzione: U armonicità} reduces to \eqref{Deduzione: tau-dphi -U} on a static space-time, so that it is natural to consider equation \eqref{Deduzione: tau-dphi -U} instead of
	\eqref{Deduzione: C-phi-PF versione 3 4}. 
	This point of view also motivates the splitting of the natural condition
	\begin{equation*}
		0=\div_1 \hat{T},
	\end{equation*}
	which, as we saw in Proposition \ref{Prop: cotton gravity: conseguenze base delle eq. di campo}, is a direct consequence of the field equations,
	into the two conditions
	\begin{equation*}
		0=\div_1 \hat{T}^F=\div_1 \hat{T}^{\hat{\p}}.
	\end{equation*}
\end{remark}
\begin{remark}\label{Deduction: remark: U armonica}
	In the Riemannian setting, a smooth map $\p$ satisfying \eqref{Deduzione: U armonicità} is called \emph{$\frac{1}{\alpha}U$-harmonic} 
	\noindent Note that, for $U$ constant, we obtain the case of harmonic maps.
	The previous definition goes back to the work of Fardoun, Ratto and Regbaoui (\cite{Fardoun1997HarmonicMW}, \cite{FRRharmonicPot}, \cite{Ratto}); it comes from a variational setting that has been vastly analyzed by Lemaire in his Ph.D. Thesis, (\cite{Lemaire1977OnTE}).
	Indeed, as in the Lorentzian case, $\frac{1}{\alpha}U$-harmonic maps  are the critical points of the functional
	\begin{align*}
		E(\p)=\frac{1}{2}\int\sq{\abs{d\p}^2+\frac{2}{\alpha}U(\p)}\,
	\end{align*}
	(see \cite{Lemaire1977OnTE} for the details).
	A map that satisfies
	\begin{align*}
		0=h\pa{\tau(\p)-\frac{1}{\alpha}\nabla^h U,\di \p}
	\end{align*}	
	is called \emph{$\frac{1}{\alpha}U$-conservative}.
\end{remark}

The following definition identifies a (static) Cotton-$\p$-perfect fluid, C-$\p$-PF for shorts.
\begin{definition}
	Given a Riemannian manifold $(M,g)$ and a smooth map $\p:(M,g)\to (N,h)$ that targets a second Riemannian manifold, and given $\alpha\in \RR, \alpha \neq 0, U\in C^\infty(N), \mu,p\in C^\infty(M)$ and a function $f\in C^{\infty}(M)$, we will say that $(M,g,f)$ is a C-$\p$-PF if it satisfies
	 \begin{subequations}\label{Deduzione: C-phi-PF}
	\begin{empheq}[left={\empheqlbrace}]{alignat=2}
			0=& \; C^\p_{ijk}+f_{ijk}-f_{ikj}-f_{ik}f_j+f_{ij}f_k\notag\\
			&-\frac{1}{m}\sq{\pa{f_{llk}-2f_lf_{lk}}\delta_{ij}-\pa{f_{llj}-2f_lf_{lj}}\delta_{ik}}\notag\\
			&+\frac{1}{2m(m-1)}\pa{S^\p_k\delta_{ij}-S^\p_j\delta_{ik}}\notag\\
			&+\frac{1}{m}U^a\p^a_j\delta_{ik}-\frac{1}{m}U^a\p^a_k\delta_{ij}+\frac{1}{m}\pa{\mu_j\delta_{ik}-\mu_k\delta_{ij}}, \label{Deduzione: C-phi-PF 1}\\
			0=&-\frac{m-1}{m}f_{lli}+\frac{m-2}{m}f_lf_{li}+\pa{\Delta f}f_i+\frac{1}{2m}S^\p_i-f_lR^\p_{il}\notag\\
			&+\frac{1}{m}U^a\p^a_i-\frac{m-1}{m}\mu_i, \label{Deduzione: C-phi-PF 2}\\
			0=& \; \tau(\p)-\di\p(\nabla f)-\frac{1}{\alpha}\nabla^h U (\p), \label{Deduzione: C-phi-PF 4}\\
			0=& \; \nabla p-(\mu+p)\nabla f \label{Deduzione: C-phi-PF 3}.
		\end{empheq}
	\end{subequations}
\end{definition}
We want to write system \eqref{Deduzione: C-phi-PF} in an equivalent form that will be useful in the next section.
From equation \eqref{Deduzione: C-phi-PF 2} we get
\begin{align*}
	\frac{m-1}{m}f_{lli}-2\frac{m-1}{m}f_lf_{li} = & \; \pa{\frac{m-2}{m}-2\frac{m-1}{m}}f_lf_{li}+\pa{\Delta f}f_i+\frac{1}{2m}S^\p_i\\
	& -f_lR^\p_{li}+\frac{1}{m}U^a\p^a_i-\frac{m-1}{m}\mu_i.
\end{align*}
Simplifying and dividing by $m-1$ we deduce
\begin{align*}
	\frac{1}{m}\pa{f_{lli}-2f_lf_{li}}=&-\frac{1}{m-1}\sq{f_lf_{li}-\pa{\Delta f}f_i}+\frac{1}{2m(m-1)}S^\p_i\\
	&-\frac{1}{m-1}f_lR^\p_{li}+\frac{1}{m(m-1)}U^a\p^a_i-\frac{1}{m}\mu_i.
\end{align*}
Using the latter into \eqref{Deduzione: C-phi-PF 1} we obtain
\begin{align*}
	0=&\;C^\p_{ijk}+f_{ijk}-f_{ikj}-f_{ik}f_j+f_{ij}f_k\\
	&+\frac{1}{m-1}\pa{f_lf_{lk}-\pa{\Delta f}f_k}\delta_{ij}-\frac{1}{2m(m-1)}S^\p_k\delta_{ij}+\frac{1}{m-1}f_lR^\p_{lk}\delta_{ij}\\
	&-\frac{1}{m(m-1)}U^a\p^a_k\delta_{ij}+\frac{1}{m}\mu_k\delta_{ij}\\
		&-\frac{1}{m-1}\pa{f_lf_{lj}-\pa{\Delta f}f_j}\delta_{ik}+\frac{1}{2m(m-1)}S^\p_j\delta_{ik}-\frac{1}{m-1}f_lR^\p_{lj}\delta_{ik}\\
	&+\frac{1}{m(m-1)}U^a\p^a_j\delta_{ik}-\frac{1}{m}\mu_j\delta_{ik}\\
	&+\frac{1}{2m(m-1)}\pa{S^\p_k\delta_{ij}-S^\p_j\delta_{ik}}+\frac{1}{m}U^a\pa{\p^a_j\delta_{ik}-\p^a_k\delta_{ij}}\\
	&+\frac{1}{m}\pa{\mu_j\delta_{ik}-\mu_k\delta_{ij}}.
\end{align*}
Simplifying, the latter becomes
\begin{align}\label{Deduzione: C-phi 1 riscritta}
	0 = & \; C^\p_{ijk}+f_{ijk}-f_{ikj}+f_{ij}f_k-f_{ik}f_j+\frac{1}{m-1}\sq{f_lf_{lk}-\pa{\Delta f}f_k}\delta_{ij}\\ \nonumber
	&-\frac{1}{m-1}\sq{f_lf_{lj}-(\Delta f)f_j}\delta_{ik}+\frac{1}{m-1}f_l\pa{R^\p_{lk}\delta_{ij}-R^\p_{lj}\delta_{ik}}\\
	&-\frac{1}{m-1}U^a\p^a_k\delta_{ij}+\frac{1}{m-1}U^a\p^a_j\delta_{ik}. \nonumber
\end{align}
Using the Ricci commutation relations and the definition of $W^\p$ we get
\begin{align}\label{Deduzione: ricci rel + W phi}
	f_{ijk}-f_{ikj} & = f_lR_{lijk}\\ \nonumber
	& = f_lW^\p_{lijk}+\frac{1}{m-2}\sq{f_jR^\p_{ik}-f_kR^\p_{ij}+f_lR^\p_{lj}\delta_{ik}-f_lR^\p_{lk}\delta_{ij}}\\ \nonumber
	& \phantom{=\;} -\frac{S^\p}{(m-1)(m-2)}\pa{f_j\delta_{ik}-f_k\delta_{ij}}.
\end{align}
Plugging \eqref{Deduzione: ricci rel + W phi} into \eqref{Deduzione: C-phi 1 riscritta} we deduce
\begin{align*}
	0=&\; C^\p_{ijk}+f_lW^\p_{lijk}+\frac{1}{m-2}\sq{f_jR^\p_{ik}-f_kR^\p_{ij}+f_lR^\p_{lj}\delta_{ik}-f_lR^\p_{lk}\delta_{ij}}\\
	&-\frac{S^\p}{(m-1)(m-2)}\pa{f_j\delta_{ik}-f_k\delta_{ij}}+\frac{1}{m-1}f_l\pa{R^\p_{lk}\delta_{ij}-R^\p_{lj}\delta_{ik}}\\ 
	&+f_{ij}f_k-f_{ik}f_j-\frac{1}{m-1}\sq{f_lf_{lj}-(\Delta f)f_j}\delta_{ik}+\frac{1}{m-1}\sq{f_lf_{lk}-(\Delta f)f_k}\delta_{ij}\\
	&-\frac{1}{m-1}U^a\p^a_k\delta_{ij}+\frac{1}{m-1}U^a\p^a_j\delta_{ik}.
\end{align*}
Simplifying
\begin{equation}\label{Deduzione: spezzamento: quasi versione con D A}
	\begin{split}
		0 = & \; C^\p_{ijk}+f_lW^\p_{lijk} + \frac{1}{m-2} (f_jR^\p_{ik}-f_kR^\p_{ij}) \\
		& + \frac{1}{(m-1)(m-2)} \left[f_l (R^\p_{lj}\delta_{ik}-R^\p_{lk}\delta_{ij})-S^\p \pa{f_j\delta_{ik}-f_k\delta_{ij}}\right] \\
		& + f_{ij}f_k-f_{ik}f_j+\frac{1}{m-1}f_l\pa{f_{lk}\delta_{ij}-f_{lj}\delta_{ik}} \\
		& - \frac{\Delta f}{m-1}\pa{f_k\delta_{ij}-f_j\delta_{ik}} -\frac{1}{m-1}U^a\pa{\p^a_k\delta_{ij}-\p^a_j\delta_{ik}}.
	\end{split}
\end{equation}
We define
\begin{equation}\label{Deduzione D A definizione}
	D^A_{ijk} := \frac{1}{m-2}\bigg[f_kR^\p_{ij}-f_jR^\p_{ik}+\frac{f_t}{m-1}\pa{R^\p_{tk}\delta_{ij}-R^\p_{tj}\delta_{ik}} - \frac{S^\p}{m-1}\pa{f_k\delta_{ij}-f_j\delta_{ik}}\bigg]
\end{equation}
and
\begin{equation}\label{Deduzione: D B definizione}
	D^B_{ijk} := \frac{1}{m-2}\bigg[f_jf_{ik}-f_kf_{ij}+\frac{f_t}{m-1}\pa{f_{tj}\delta_{ik}-f_{tk}\delta_{ij}} - \frac{\Delta f}{m-1}\pa{f_j\delta_{ik}-f_k\delta_{ij}}\bigg]
\end{equation}
so that \eqref{Deduzione: spezzamento: quasi versione con D A} becomes
\begin{equation}\label{Deduzione: C phi prima: versione con le D}
	0 = C^\p_{ijk}+f_lW^\p_{lijk}-D^A_{ijk}-(m-2)D^B_{ijk}-\frac{1}{m-1}U^a\pa{\p^a_k\delta_{ij}-\p^a_j\delta_{ik}}.
\end{equation}
System \eqref{Deduzione: C-phi-PF} is therefore equivalent to
\begin{subequations}\label{Deduzione: C-phi-PF: versione con le D}
\begin{empheq}[left={\empheqlbrace}]{alignat=2}
		0=& \; C^\p_{ijk}+f_lW^\p_{lijk}-D^A_{ijk}-(m-2)D^B_{ijk}-\frac{1}{m-1}U^a\pa{\p^a_k\delta_{ij}-\p^a_j\delta_{ik}}, \label{Deduzione: C-phi-PF: versione con le D 1}\\
		0=&\; -\frac{m-1}{m}f_{lli}+\frac{m-2}{m}f_lf_{li}+\pa{\Delta f}f_i+\frac{1}{2m}S^\p_i - f_lR^\p_{li}+\frac{1}{m}U^a\p^a_i-\frac{m-1}{m}\mu_i, \label{Deduzione: C-phi-PF: versione con le D 2}\\
		0=& \; \tau(\p)-\di\p(\nabla f)-\frac{1}{\alpha}\nabla^h U (\p), \label{Deduzione: C-phi-PF: versione con le D 4}\\
		0=& \; \nabla p-(\mu+p)\nabla f \label{Deduzione: C-phi-PF: versione con le D 3}.
\end{empheq}
\end{subequations}
Equation \eqref{Deduzione: C-phi-PF: versione con le D 1} has two important features, that are not seen in \eqref{Deduzione: C-phi-PF 1}:
on the one hand, the tensors $D^A$ and $D^B$ appear explicitly in it; as we will see in Section \ref{Sect 4: local struct}, these tensors are fundamental in the study of the geometry of the regular level sets of $f$; on the other hand, \eqref{Deduzione: C-phi-PF: versione con le D} shares some striking formal similarities with the first integrability condition of  a $\p$-\emph{static perfect fluid space-time}.  
%
%
%
%
%
%
%
%
%
\subsection{$\p$-static perfect fluid space-times}
Following \cite{phiSPFST}, consider, on a static space-time $M\times_{f} \RR$, the solutions of the Einstein field equations
\begin{align}\label{C phi e phi PF: equazioni di einstein}
	\hat{\ric}-\frac{1}{2}\hat{S}\hat{g}=\hat{T}^F+\hat{T}^{\hat \p}
\end{align}
with source given by the combination of a perfect fluid and a non-linear field $\hat{\p}$,
subject to the further condition
\begin{align}\label{C phi e phi PF: equazioni di conservazione}
	\begin{cases}
		\div \hat{T}^F=0,\\[0.2cm]
		\tau(\hat{\p})= \dfrac{1}{\alpha} \nabla^h U .
	\end{cases}
\end{align}
Splitting  \eqref{C phi e phi PF: equazioni di einstein} and \eqref{C phi e phi PF: equazioni di conservazione} on the space and time components of $(\hat{M},\hat{g})$ we obtain a solution of the system
\begin{subequations}\label{Cphi e phiPF: sistema phiSPFST con la f}
	\begin{empheq}[left={\empheqlbrace}]{alignat=2}
		&\ric^\p+\hess f-\di f\otimes \di f-\frac{1}{m-1}\pa{\frac{S^\p}{2}-p+U(\p)}g=0, \label{Cphi e phiPF: sistema phiSPFST con la f 1}\\
		&\Delta_f f=-\frac{1}{m-1}\sq{mp-mU(\p)+\frac{m-2}{2}S^\p}, \label{Cphi e phiPF: sistema phiSPFST con la f 2}\\
		&\tau(\p)-\di \p(\nabla f)=\frac{1}{\alpha} \nabla^h U , \label{Cphi e phiPF: sistema phiSPFST con la f 3}\\
		&\mu+U(\p)=\frac{1}{2}S^\p, \label{Cphi e phiPF: sistema phiSPFST con la f 4}\\[0.2cm]
		&\nabla p-(\mu+p)\nabla f=0. \label{Cphi e phiPF: sistema phiSPFST con la f 5}
	\end{empheq}
\end{subequations}
Performing the substitution $u=e^{-f}$, system \eqref{Cphi e phiPF: sistema phiSPFST con la f} becomes
\begin{align}\label{Cphi e phiPF: sistema phiSPFST con la u}
	\begin{cases}
		\hess u-u\left\{\ric^\p-\dfrac{1}{m-1}\pa{\dfrac{S^\p}{2}-p+U(\p)}g\right\}=0,\\[0.2cm]
		\Delta u=\dfrac{u}{m-1}\sq{mp-mU(\p)+\dfrac{m-2}{2}S^\p},\\[0.4cm]
		u\tau(\p)=-\di \p(\nabla u)+\dfrac{u}{\alpha} \nabla^h U,\\[0.2cm]
		\mu+U(\p)=\dfrac{1}{2}S^\p,\\[0.2cm]
		(\mu+p)\nabla u=-u\nabla p.
	\end{cases}
\end{align}
From \eqref{Cphi e phiPF: sistema phiSPFST con la u} one can see that, when $\p$ is constant and $U\equiv 0$, the system reduces to that of a classical static perfect fluid.
It will be useful to reformulate system \eqref{Cphi e phiPF: sistema phiSPFST con la f} in the following way: solving \eqref{Cphi e phiPF: sistema phiSPFST con la f} for $p$ we find
\begin{align*}
p=-\frac{m-1}{m}\Delta_f f+U(\p)-\frac{m-2}{2m}S^\p;
\end{align*}
using the latter in \eqref{Cphi e phiPF: sistema phiSPFST con la f 1} we get
\begin{align*}
	\ric^\p+\hess f-\di f\otimes \di f=\frac{1}{m}\pa{S^\p+\Delta_f f}g.
\end{align*}
Solving \eqref{Cphi e phiPF: sistema phiSPFST con la f 4} for $U(\p)$ we find
\begin{align*}
	U(\p)=\frac{1}{2}S^\p-\mu
\end{align*} 
and inserting the latter in \eqref{Cphi e phiPF: sistema phiSPFST con la f 2} we get
\begin{align*}
	\Delta_f f=\frac{1}{m-1}\pa{S^\p-m(\mu+p)}.
\end{align*}
Calling 
\begin{align*}
	\lambda:=\frac{1}{m}\pa{S^\p+\Delta_f f}
\end{align*}
we deduce that \eqref{Cphi e phiPF: sistema phiSPFST con la f} is equivalent to
\begin{subequations}\label{Cphi e phiPF: sistema phiSPFST riscritto}
	\begin{empheq}[left={\empheqlbrace}]{alignat= 2}
		&\ric^\p+\hess f-\di f\otimes \di f=\lambda g, \label{Cphi e phiPF: sistema phiSPFST riscritto 1}\\
		&\Delta_f f=\frac{1}{m-1}\sq{S^\p-m(\mu+p)}, \label{Cphi e phiPF: sistema phiSPFST riscritto 2}\\
		&\tau(\p)=\di \p(\nabla f)+\frac{1}{\alpha} \nabla^h U, \label{Cphi e phiPF: sistema phiSPFST riscritto 3}\\
		&\frac{1}{2}S^\p=U(\p)+\mu, \label{Cphi e phiPF: sistema phiSPFST riscritto 4}\\[0.1cm]
		&(\mu+p)\nabla f-\nabla p=0. \label{Cphi e phiPF: sistema phiSPFST riscritto 5}
	\end{empheq}
\end{subequations}

\begin{definition}
	Let $(M,g)$ be an $m$-dimensional Riemannian manifold and let $\p: (M,g)\to (N,h)$ be a smooth map that targets another Riemannian manifold. Given $\alpha \in \RR, \alpha \neq 0, U\in C^
	{\infty}(N), \mu,p,\lambda\in C^\infty(M)$ and $f\in C^{\infty}(M)$, we will say that $(M,g)$ is a $\p$-static perfect fluid space-time ($\p$-SPFST for short) if $f$ is a solution of \eqref{Cphi e phiPF: sistema phiSPFST riscritto}, for $\lambda$ as above.
\end{definition}

\begin{remark}
Recall from Section \ref{Sect 1: setting} that every solution of Einstein's field equations for a stress energy tensor $T$ is also a solution of the Cotton gravity equations, for the same stress-energy tensor. Therefore, we have that every $\p$-SPFST is also a C-$\p$-PF.
\end{remark}

The next result, which corresponds to the case $\eta=1$ of Proposition 4.1 of \cite{phiSPFST}, gives some integrability conditions for $\p$-SPFSTs. 
\begin{proposition}\label{KO: 1 st and 2 nd integrability}
	Let $(M,g)$ be a manifold of dimension $m\geq 3$. Let $\p:(M,g)\to (N,h)$, $U:(N,h)\to \RR$, $\lambda:(M,g)\to \RR$ be smooth maps, $\alpha\in \RR\setminus \{0\}$ and let $f\in C^{\infty}(M)$; assume that equations \eqref{Cphi e phiPF: sistema phiSPFST riscritto 1} and \eqref{Cphi e phiPF: sistema phiSPFST riscritto 3} hold on $M$. We then have
	\begin{align}\label{1st int cond}
		(m-1)D^A_{ijk}=C^\p_{ijk}+f_lW^\p_{lijk}+\frac{U^a}{m-1}\pa{\p^a_j\delta_{ik}-\p^a_k\delta_{ij}}
	\end{align}
	and
	\begin{align}\label{2nd int cond}
		(m-2)B^\p_{ij}=&\;(m-1)\sq{D^\p_{ijk,k}-\frac{\alpha}{m-2}\p^a_{ss}\p^a_if_j}+\frac{m-3}{m-2}f_kC^\p_{jik}- W^\p_{lijk}f_lf_k\\ \nonumber
		&+\frac{U^a}{m-1}\sq{(m-2)\p^a_{ij}-\frac{1}{m-2}\p^a_{ss}\delta_{ij}}\\ \nonumber
		&+\frac{U^{ab}}{m-1}\sq{\p^a_k\p^b_k\delta_{ij}-m\p^a_i\p^b_j}\\ \nonumber
		&+ f_jU^a\p^a_i.
	\end{align}
\end{proposition}
\begin{remark}
It is easy to prove how, assuming that \eqref{Cphi e phiPF: sistema phiSPFST riscritto 1} holds, one has $D^A\equiv D^B$. Therefore, under this assumption, \eqref{Deduzione: C-phi-PF: versione con le D 1}  becomes \eqref{1st int cond}. Equation \eqref{2nd int cond} is important in order to obtain rigidity results for $\p$-SPFSTs, see \cite[Theorem 4.17]{phiSPFST}. A generalization of \eqref{2nd int cond}, with applications to C-$\p$-PFs, is given by Proposition \ref{teoremone: prop: divergenza di D bar e co} below.
\end{remark}

\section{Local structure and a rigidity result}\label{Sect 4: local struct}
In this section, we prove a rigidity result for a C-$\p$-PF (in fact, for a slightly more general system). The latter is inspired by an analogous result for $\p$-SPFST, namely \cite[Theorem 4.17]{phiSPFST} and therefore it gives an answer to Question \ref{Question: domanda 1} of the Introduction.
In Section \ref{subsect: local struct} we study how the vanishing of $D^A$ and $D^B$, as defined in \eqref{Deduzione D A definizione} and \eqref{Deduzione: D B definizione}, affects the geometry of the regular level sets of $f$.
Then we find conditions that imply the vanishing of $D^A$ and $D^B$.
 \subsection{The local structure of a C-$\p$-PF}\label{subsect: local struct}
In the next Theorem, we show that the vanishing of $D^A$ and $D^B$ entails a local warped product splitting of the metric $g$ together with a characterization of the geometry of the level regular level sets of the defining function $f$ of a C-$\p$-PF. Note that, of all the  equations of a C-$\p$-PF, we will only need \eqref{Deduzione: C-phi-PF: versione con le D 1} and\eqref{Deduzione: C-phi-PF: versione con le D 4}. 
 \begin{theorem}\label{local struct: teorema warp product}
 	Let $(M,g)$ be a smooth, $m$-dimensional Riemannian manifold and let $\p:(M,g)\to (N,h)$ be a smooth map that targets another Riemannian manifold $(N,h)$. For a constant $0\neq \alpha\in \RR$ and  some functions $U\in C^\infty(N), f\in C^\infty(M)$, assume that the following system holds:
 	\begin{subequations}\label{local struct: 1 cond di int}
 		\begin{empheq}[left={\empheqlbrace}]{alignat=2}
 			&0=C^\p_{ijk}+f_lW^\p_{lijk}-\frac{1}{m-1}U^a\p^a_k\delta_{ij}+\frac{1}{m-1}U^a\p^a_j\delta_{ik}-D^A_{ijk}-(m-2)D^B_{ijk}, \label{local struct: 1 cond di int 1}\\
 			&\tau(\p)=\di \p(\nabla f)+\frac{1}{\alpha} \nabla^h U. \label{local struct: 1 cond di int 2}
 		\end{empheq}
 	\end{subequations}
 	Assume that $D^A\equiv D^B\equiv 0$ and that $\p$ is $\frac{1}{\alpha}U$-harmonic. Let $c\in \RR$ be a regular value of $f$ and let $\Sigma=\Sigma_c$ be the corresponding level set. Then, for all $ p\in \Sigma$, there exists in $M$ an open neighbourhood $A$ of $p$ such that $g_{|_A}$ is a warped product $(I\times_\rho(\Sigma \cap A), dr^2+\rho^2(r)g_{\Sigma})$ where $I$ is an open interval, $r$ is the signed distance function from $\Sigma$ and $\rho:I\to \RR$ is the warping factor. Moreover $U(\p)$ is locally constant on $\Sigma$ and $(\Sigma, g_{\Sigma})$ satisfies
 	\begin{align}\label{local struct: teo: harmonic einstein}
 		\begin{cases}
 			\Ricc^{\p_{|_\Sigma}} = \dfrac{S^{\p_{|_\Sigma}}}{m-1}g_{\Sigma}, \\[0.2cm]
 			h\pa{\tau(\p_{|_{\Sigma}}),\di \p_{|_\Sigma}}=0,
 		\end{cases}
 	\end{align}
 	where $\Ricc^{\p_{|_\Sigma}}$ and $S^{\p_{|_\Sigma}}$ denote the $\p$-Ricci curvature and the $\p$-scalar curvature of $(\Sigma,g_\Sigma)$.
 \end{theorem}
 
 \begin{remark}
 	Theorem \ref{local struct: teorema warp product} is similar to Theorem 4.14 of \cite{phiSPFST}. There, $(M,g)$ satisfied
 	\begin{align}\label{local struct: phi einstein type}
 		\begin{cases}
 			\ric^\p + \hess f - \eta \, \di f\otimes \di f=\lambda g\\[0.1cm]
 			\tau(\p) = \di \p(\nabla f) + \dfrac{1}{\alpha} \nabla^h U .
 		\end{cases}
 	\end{align}
 	As we saw in Section \ref{Sect 3: system},  such a manifold satisfies \eqref{local struct: 1 cond di int} when $\eta=1$, but this is just a necessary condition, so that the present setting is more general.
 	The proof of Theorem \ref{local struct: teorema warp product} follows the same lines of that of \cite[Theorem 4.14]{phiSPFST} but, due to the differences between equations \eqref{local struct: 1 cond di int} and \eqref{local struct: phi einstein type}, some modifications are needed. We will split the proof into several propositions.
 \end{remark}
 From here on, $c$ will be a regular value of $f$, $\Sigma$ will be the corresponding level set and $p\in \Sigma$.
 At such a $p$, $\{e_i\}_{i=1}^m$ will be an orthonormal frame such that
 \begin{align*}
 	e_m=\frac{\nabla f}{\abs{\nabla f}}
 \end{align*}
 and upper case letters $A,B,C,...$ will denote indexes ranging from $1$ to $m-1$. The dual co-frame of $\{e_i\}_{i=1}^m$ will be denoted by $\{\theta^i\}_{i=1}^m$. For such a co-frame we have
 \begin{align*}
 	\theta^m_{\ A}(e_B)=\Pi_{AB}=-\frac{f_{AB}}{\abs{\nabla f}},
 \end{align*}
 see Proposition 8.1 of \cite{CatMastMontRig}. Here, $\Pi$ is the second fundamental form of $\Sigma$ and $\theta^m_{\ A}(e_B)$ is of course one of the components of the connection form $\theta^m_{\ A}$ associated to the given co-frame. The trace of $\Pi$, the mean curvature, is denoted by $H$. Recall that the immersion $i:\Sigma\hookrightarrow M$ is said to be \emph{totally umbilical} at $p\in \Sigma$ if it holds $\Pi=\frac{H}{m-1}g_{\Sigma}$ at $p$. 
 \begin{proposition}\label{local struct: prop: D(P) è 0}
 	Let $(M,g)$ be a Riemannian manifold and let $P$ be a $2$-covariant symmetric tensor field on $M$. For $f\in C^\infty (M)$, define
 	\begin{align}\label{local struct: tensore D di P}
 		D(P)_{ijk}:=\frac{1}{m-2}\bigg[&f_kP_{ij}-f_jP_{ik}+\frac{f_l}{m-1}\pa{P_{lk}\delta_{ij}-P_{lj}\delta_{ik}}-\frac{P_{ll}}{m-2}\pa{f_k\delta_{ij}-f_j\delta_{ik}}\bigg].
 	\end{align}
 	Then $D(P)\equiv 0$ if and only if we have the following:
 	\begin{itemize}
 		\item[i)] at any regular point $p\in M$ of $f$, $\nabla f$ is an eigenvector of $P$;
 		\item[ii)] at any regular point $p\in M$ of $f$ and for any orthonormal frame $\{e_i\}_{i=1}^m$ at $p$ such that $e_m=\frac{\nabla f}{\abs{\nabla f}}$ we have
 		\begin{align*}
 			P_{AB}=\frac{P_{CC}}{m-1}\delta_{AB}.
 		\end{align*}
 	\end{itemize}
 \end{proposition}
 \begin{remark}
 	Proposition \ref{local struct: prop: D(P) è 0} implies that, when $D(P)\equiv 0$, $P$ has at most two eigenvalues at any regular point $p$ of $f$, with eigenspaces $\left<\nabla f\right>$ and $\left<\nabla f\right>^{\perp}$.
 \end{remark}
 \begin{proof}[Proof of Proposition \ref{local struct: prop: D(P) è 0}]
For the time of this proof, we will set $D=D(P)$ for simplicity.
 The condition that $\nabla f$ be an eigenvector of $P$ at $p$ is equivalent to
 \begin{align}\label{local struct: D(P)=0: eigenvector}
 	f_jf_lP_{lk}=f_kf_lP_{lj}.
 \end{align}	
 	Indeed, in the dual co-frame $\{\theta^i\}_{i=1}^m$ of $\{e_i\}_{i=1}^m$, \eqref{local struct: D(P)=0: eigenvector} is equivalent to
 	\begin{align*}
 		\abs{\nabla f}f_jP_{mk}=\abs{\nabla f}f_kP_{mj};
 	\end{align*}
 	therefore, since $f_A=0, \text{for all } A\in \{1,...,m-1\}$ and $f_m=\abs{\nabla f}\neq 0$, we get
 	\begin{align*}
 		0=f_A P_{mm}=f_mP_{mA}=\abs{\nabla f}P_{mA}
 	\end{align*}
 	so that $P_{mA}=0, \text{for all } A\in \{1,...,m-1\}$.
 Now, contracting \eqref{local struct: tensore D di P} with $f_i$ we get
 \begin{align*}
 	f_iD_{ijk}=&\;\frac{1}{m-2}\bigg[f_lf_kP_{lj}-f_lf_jP_{lk}+\frac{1}{m-1}f_l(P_{lk}f_j-P_{lj}f_k) - \frac{P_{ll}}{m-1}(f_jf_k-f_kf_j)\bigg]\\
 	=&\;\frac{1}{m-1}(f_lf_kP_{lj}-f_lf_jP_{lk})
 \end{align*}	
 	implying that \eqref{local struct: D(P)=0: eigenvector} is equivalent to
 	\begin{align}\label{local struct: D(P)=0: D m j k}
 		D_{mjk}=0.
 	\end{align}
 	Having that \eqref{local struct: D(P)=0: eigenvector} holds, we also get
 	\[
 		D_{ABC} = \frac{1}{m-2}\bigg[f_CP_{AB}-f_BP_{AC}+\frac{\abs{\nabla f}}{m-1}(P_{mC}\delta_{AB}-P_{mB}\delta_{AC})-\frac{P_{ll}}{m-1}(f_C\delta_{AB}-f_B\delta_{AC})\bigg] = 0
 	\]
 	since $f_A=f_B=f_C=0$. We are only left to prove that, assuming \eqref{local struct: D(P)=0: eigenvector}, the condition
 	\begin{align*}
 		D_{ABm}=0
 	\end{align*}
 	is equivalent to
 	\begin{align*}
 		P_{AB}=\frac{P_{CC}}{m-1}\delta_{AB}.
 	\end{align*}
 	From \eqref{local struct: tensore D di P}, using $f_A=f_B=0, f_m=\abs{\nabla f}$ and $P_{Am}=0, \text{for all } A,$ we deduce
 	\begin{align*}
 		D_{ABm}=&\;\frac{1}{m-2}\bigg[f_mP_{AB}-f_BP_{Am}+\frac{\abs{\nabla f}}{m-1}\pa{P_{mm}\delta_{AB}-P_{mB}\delta_{Am}}-\frac{P_{ll}}{m-1}\pa{f_m\delta_{AB}-f_B\delta_{Am}}\bigg]\\
 		=&\;\frac{\abs{\nabla f}}{m-2}\bigg[P_{AB}+\frac{P_{mm}}{m-1}\delta_{AB}-\frac{1}{m-1}\pa{P_{CC}+P_{mm}}\delta_{AB}\bigg]\\
 		=&\;\frac{\abs{\nabla f}}{m-2}\bigg[P_{AB}-\frac{P_{CC}}{m-1}\delta_{AB}\bigg]
 	\end{align*}
 	and we are done.
 \end{proof}
 \begin{corollary}\label{local struct: cor: D B =0}
 	Let $(M,g)$ be a Riemannian manifold and let $f\in C^\infty(M)$. Assume $D^B\equiv 0$. Then:
 	\begin{itemize}
 		\item[i)] for every regular point $p$ of $f$, $\nabla f$ is an eigenvector of $\hess f$;
 		\item[ii)] for any regular level set $\Sigma$ of $f$, $i:\Sigma \hookrightarrow M$ is totally umbilical.
  	\end{itemize}
 \end{corollary}
 \begin{proof}
 	Since $D^B=D(-\hess f)$, the conclusion follows from Proposition \ref{local struct: prop: D(P) è 0} and $f_{AB}=-\abs{\nabla f}\Pi_{AB}$.
 \end{proof}

 \begin{proposition}\label{local struct: prop: warped product}
 	Let $(M,g)$ be a Riemannian manifold and let $f \in C^{\infty}(M)$. Assume that, at any regular point $p$ of $f$,
 	\begin{align}
 		D^B  = 0, \qquad R^{\p}_{Am} = 0 \quad \text{ for all } A, \qquad \di \p(\nabla f) = 0.
 	\end{align}
 	For a regular value $c \in \mathbb{R}$ of $f$, let $\Sigma = f^{-1}(c)$ be the corresponding level set of $f$. Then, for all $ p \in \Sigma$, $\exists A$, $p \in A \subset M$ open in $M$ such that $g_{|_A}$ is a warped product
 	\begin{align}
 		\left(I \times_{\rho} ({\Sigma \cap A}), \di r^2 + \rho^2(r) g_{\Sigma}\right)
 	\end{align}
 	where $I$ is an open interval, $r$ is the signed distance function from $\Sigma \cap A$ and $\rho: I \to \mathbb{R}$ is the warping factor.
 \end{proposition}
 
 \begin{proof}
 	First, we prove that the mean curvature $H$ of $\Sigma$ is locally constant on it. Since $D^B = 0$ we have, by Corollary \ref{local struct: cor: D B =0},
 	\begin{align*}
 		\Pi_{AB} = \frac{H}{m-1} \delta_{AB}.
 	\end{align*}
 	Using Codazzi equations and $R^{\p}_{Am} = 0$ we get
 	\begin{align*}
 		H_B &= \Pi_{AA,B} = \Pi_{AB,A} - R_{mAAB} \\
 		&= \frac{H_B}{m-1} + R^{\p}_{mB} +\alpha \p^a_m\p^a_B \\
 		&= \frac{H_B}{m-1} + \alpha \p^a_m\p^a_B 
 	\end{align*}
 	and therefore
 	\begin{align*}
 		\frac{m-2}{m-1}H_B = \alpha \p^a_m\p^a_B.
 	\end{align*}
 	From $\di \p(\nabla f) = 0$ we get
 	\begin{align*}
 		\p^a_m = 0
 	\end{align*}
 	so that $H_B = 0$ and $H$ is constant on $\Sigma$.
%
	
	Since $D^B = 0$, we have, see Corollary \ref{local struct: cor: D B =0},
 	\begin{equation*}
 		\frac{1}{2}\abs{\nabla f}^2_A = f_lf_{Al} = \abs{\nabla f}f_{Am} = 0
 	\end{equation*}
 	so that $\abs{\nabla f}$ is locally constant on each regular level set of $f$. As it is well-known, this implies that, for all $ p\in \Sigma $, there exists  an open neighbourhood $A$ of $p$ in $M$  such that $f_{\mid A}$ only depends on $r$. See the proof of the equivalence of item $ii)$ and item $iii)$ of \cite[Lemma 4.4]{phiSPFST} for more details. By restricting $A$, we can assume that $A$ is made only of regular points of $f$, and $A$ is a tubular neighbourhood of $A\cap \Sigma$. We can choose $r$ so that
	\begin{align*}
		\nabla r = \frac{\nabla f}{|\nabla f|}.
	\end{align*}
 	Expressing $g$ in Fermi coordinates $(x^1,\ldots,x^{m-1},r)$ we get, on $A$,
 	\begin{equation*}
 		g = \di r \otimes \di r + g_{AB}(\underline{x},r) \, \di x^A \otimes \di x^B
 	\end{equation*}
 	and
 	\begin{equation*}
 		\hess f = f''\, \di r \otimes \di r + \frac{f'}{2}\, \partial_r g_{AB}(\underline{x},r) \, \di x^A \otimes \di x^B
 	\end{equation*}
 	where $\underline{x} = (x^1,\ldots,x^{m-1})$. Therefore
 	\begin{equation}\label{local struct: warped: hessiano di f}
 		f_{AB} = \frac{f'}{2}\, \partial_r g_{AB}.
 	\end{equation}
 	Since $i: \Sigma_{\overline c} \hookrightarrow M$ is totally umbilical, for every $\overline{c}$ sufficiently close to $c$, we get
 	\begin{equation}\label{local struct: warped: hessiano di f 2}
 		 \frac{f_{AB}}{f'} = -\Pi_{AB} = -\frac{H}{m-1}\, g_{AB}
 	\end{equation}
 	Comparing \eqref{local struct: warped: hessiano di f} and \eqref{local struct: warped: hessiano di f 2} we have
 	\begin{equation*}
 		-\frac{H}{m-1}\, g_{AB} = \frac{1}{2}\, \partial_r g_{AB}
 	\end{equation*} 	
 	and integrating the above expression
 	\begin{equation*}
 		g_{AB}(\underline{x},r) = e^{-2(m-1)^{-1}\int^r_0 H(s)\,ds}\, g_{AB}(\underline{x},0);
 	\end{equation*}
    since $g_{AB}(\underline{x},0)\, \di x^A \otimes \di x^B = g_\Sigma$,
 	we are done.
\end{proof}

We now want to describe the geometry of the regular level sets of $f$.

\begin{proposition}\label{local struct: prop: fibre harm einstein}
	Let $(M,g)$ satisfy system \eqref{local struct: 1 cond di int} and
	assume that
	\begin{align*}
		D^A \equiv D^B \equiv 0
	\end{align*}
	and
	\begin{align*}
		\tau(\p) = \frac{\nabla^h U}{\alpha}
	\end{align*}
	and let $\Sigma$ be a regular level set of $f$. Then,
	\begin{subequations}\label{local struct: fibre: harm einstein}
 		\begin{empheq}[left={\empheqlbrace}]{alignat=2}
		&	\ric^{\p_{|_{\Sigma}}} = \frac{S^{\p_{|_{\Sigma}}}}{m-1} g_{\Sigma}, \label{local struct: fibre: harm einstein 1} \\
		&	h(\tau(\p_{|_{\Sigma}}), \di \p_{|_{\Sigma}}) = 0. \label{local struct: fibre: harm einstein 2}
		\end{empheq}
	\end{subequations}
\end{proposition}

To prove Proposition \ref{local struct: prop: fibre harm einstein},we will need two lemmata.

\begin{lemma}\label{local struct: lemma: W phi mamb}
	Under the assumption of Proposition \ref{local struct: prop: fibre harm einstein}, let $p\in M$ be a regular point for $f$, and let $\{\theta^i\}_{i=1}^m$ be an orthonormal coframe at $p$ such that $\theta^m=\frac{\di f}{\abs{\nabla f}}$. In this coframe, we have
	\begin{align}\label{local struct: fibre: W phi 0}
		W^\p_{mAmB}=0, \ \ \text{ for all } A,B\in \{1,...,m-1\}.
	\end{align}
\end{lemma}

\begin{proof}
	From \eqref{local struct: 1 cond di int 1} and $D^A \equiv D^B \equiv 0$ we get
	\begin{equation}\label{local struct: fibre: sistema bis}
		0 = C^{\p}_{AmB} + \abs{\nabla f} W^{\p}_{mAmB} - \frac{1}{m-1}U^{a}\p^{a}_B\delta_{Am} + \frac{1}{m-1}U^{a}\p^{a}_m\delta_{AB}.
	\end{equation}
	Since $D^A = 0$, Proposition \ref{local struct: prop: D(P) è 0} implies that $R^{\p}_{Am} = 0$ for all $A\in\{1,...,m-1\}$ and therefore
	\begin{align*}
		0 = \di R^{\p}_{Am} &= R^{\p}_{Am,k}\theta^k + R^{\p}_{km} \theta^k_{\ A} + R^{\p}_{Ak} \theta^k_{\ m} \\
		&= R^{\p}_{Am,k} \theta^k + R^{\p}_{Bm} \theta^B_{\ A} + R^{\p}_{mm} \theta^m_{\ A} + R^{\p}_{AB} \theta^B_{\ m} + R^{\p}_{Am} \theta^m_m \\
		&= R^{\p}_{Am,k} \theta^k + R^{\p}_{mm} \theta^m_{\ A} + R^{\p}_{AB} \theta^B_{\ m}.
	\end{align*}
	Since $D^A = 0$, Proposition \ref{local struct: prop: D(P) è 0} gives
	\begin{equation}\label{local struct: fibra: fibra einstein 1 eq}
		\quad R^{\p}_{AB} = \frac{S^{\p} - R^{\p}_{mm}}{m-1} \delta_{AB}
	\end{equation}
	so that
	\begin{align*}
		R^{\p}_{Am,k} \theta^k &= -R^{\p}_{mm} \theta^m_{\ A} - \frac{1}{m-1}(S^{\p} - R^{\p}_{mm})\theta^A_{\ m} \\
		&= \frac{1}{m-1}(S^{\p} - mR^{\p}_{mm}) \theta^m_{\ A}.
	\end{align*}
	Recalling that $\theta^m_{\ A} (e_B) = \Pi_{AB} = -\frac{f_{AB}}{|\nabla f|}$, we get 
	\begin{equation}\label{local struct: R phi A m ,B}
		\quad R^{\p}_{Am,B} = -\frac{1}{m-1}(S^{\p} - mR^{\p}_{mm}) \frac{f_{AB}}{|\nabla f|} \, .
	\end{equation}
	By the definition of the $\p$-Cotton tensor, using \eqref{local struct: fibra: fibra einstein 1 eq} and \eqref{local struct: R phi A m ,B} we obtain
	\begin{align*}
		C^{\p}_{ABm} =& \; R^{\p}_{AB,m} - R^{\p}_{Am,B} - \frac{1}{2(m-1)}\left(S^{\p}_m \delta_{AB} - S^{\p}_B \delta_{Am}\right) \\
		=& \; \frac{1}{m-1}\left(S^{\p} - R^{\p}_{mm}\right)_m \delta_{AB} + \frac{1}{m-1}\left(S^{\p} - mR^{\p}_{mm}\right)\frac{f_{AB}}{|\nabla f|} - \frac{S^{\p}_m}{2(m-1)} \delta_{AB} \\
		= & \; \frac{1}{2(m-1)} S^{\p}_m \delta_{AB} - \frac{1}{m-1}R^{\p}_{mm,m} \delta_{AB} + \frac{1}{m-1}\left(S^{\p} - mR^{\p}_{mm}\right)\frac{f_{AB}}{|\nabla f|}.
	\end{align*}
	By the $\p$-Schur identity
	\begin{align*}
		\frac{1}{2}S^{\p}_m =& \; \alpha \p^{a}_{ll} \p^{a}_m + R^{\p}_{lm,l} \\
		=& \; \alpha \p^{a}_{ll} \p^{a}_m + R^{\p}_{Am,A} + R^{\p}_{mm,m}
	\end{align*}
	and, since $\p^{a}_{ll} = \frac{1}{\alpha}U^a$, we deduce
	\begin{align*}
		C^{\p}_{ABm} &= \frac{1}{m-1}U^a \p^{a}_m \delta_{AB} + \frac{1}{m-1}R^{\p}_{Cm,C} \delta_{AB} + \frac{1}{m-1}\left(S^{\p} - mR^{\p}_{mm}\right)\frac{f_{AB}}{|\nabla f|}.
	\end{align*}
	Taking the trace of \eqref{local struct: R phi A m ,B} we have
	\begin{align*}
		R^{\p}_{Cm,C} = -\frac{1}{m-1}\left(S^{\p} - mR^{\p}_{mm}\right)\frac{f_{BB}}{|\nabla f|},
	\end{align*}
	and therefore
	\begin{align*}
		C^{\p}_{ABm} & = \frac{1}{m-1}U^a \p^{a}_m \delta_{AB} + \frac{1}{m-1}\left(S^{\p} - mR^{\p}_{mm}\right)\left(\frac{f_{AB}}{|\nabla f|} - \frac{f_{CC}}{(m-1)|\nabla f} \delta_{AB}\right).
	\end{align*}
	Since $i: \Sigma \hookrightarrow M$ is totally umbilical it holds $f_{AB} = \frac{f_{CC}}{m-1} \delta_{AB}$ so that
	\begin{align*}
		C^{\p}_{ABm} = \frac{U^a \p^a_m}{m-1} \delta_{AB};
	\end{align*}
	from \eqref{local struct: fibre: sistema bis} we get \eqref{local struct: fibre: W phi 0} as we wanted to prove.
\end{proof}

\begin{lemma}\label{lemma: local struct: ric phi sigma}
	Under the assumptions of Proposition \ref{local struct: prop: fibre harm einstein}, for every regular level set $\Sigma$ of $f$ it holds
	\begin{align}\label{local struct: fibre: harm einstein 1 bis}
			\ric^{\p_{|_{\Sigma}}} = \frac{S^{\p_{|_{\Sigma}}}}{m-1} g_{\Sigma}
	\end{align}
	and
	\begin{equation}\label{local struct: fibre: S phi sigma}
		S^{\varphi_{|_{\Sigma}}} = S^{\varphi}- \frac{2}{\abs{\nabla f}^2}\Ric^{\varphi}(\nabla f,\nabla f) + (m-2)(m-1) H^2.
	\end{equation}
\end{lemma}

\begin{proof}
	Since $D^A = 0$, Proposition \ref{local struct: prop: D(P) è 0} gives
	\begin{equation}\label{local struct: fibra: fibra einstein 1 eq bis}
		R^{\p}_{AB} = \frac{S^{\p} - R^{\p}_{mm}}{m-1} \delta_{AB}.
	\end{equation}
	Using \eqref{local struct: fibre: W phi 0} and the definition of $W^{\varphi}$ we deduce
	\begin{align*}
		R_{mAmB} = & \; W^{\varphi}_{mAmB} + \frac{1}{m-2}(R^{\varphi}_{mm}\delta_{AB} + R^{\varphi}_{AB}\delta_{mm}- R^{\varphi}_{mB}\delta_{Am} -R^{\varphi}_{Am}\delta_{mB})\notag\\
		&-\frac{S^{\varphi}}{(m-1)(m-2)}(\delta_{mm}\delta_{AB} - \delta_{mB}\delta_{mA})\\
		= & \; \frac{1}{m-2} (R^{\varphi}_{mm}\delta_{AB} + R^{\varphi}_{AB} - \frac{S^{\varphi}}{m-1}\delta_{AB}),
	\end{align*}
	so that \eqref{local struct: fibra: fibra einstein 1 eq bis} implies
	\begin{equation}\label{local struct: fibre: Riemann}
		R_{mAmB} = \frac{1}{m-1} R^{\varphi}_{mm}\delta_{AB}.
	\end{equation}
	From the Gauss equation and the fact that $i:\Sigma \hookrightarrow M$ is totally umbilical we deduce
	\begin{equation*}
		^\Sigma R_{AC} = R_{AC} -  R_{AmCm} + (m-2) H^2\delta_{AC}
	\end{equation*}
	so that, by the definition of $\ric^{\varphi_{|\Sigma}},$ we obtain
	\begin{equation*}
		R^{\varphi_{|_{\Sigma}}}_{AC} =\  ^\Sigma R_{AC} - \alpha \p^a_A \p^a_C = R^{\varphi}_{AC} - R_{AmCm} + (m-2) H^2 \delta_{AC}.
	\end{equation*}
	From \eqref{local struct: fibra: fibra einstein 1 eq bis} and \eqref{local struct: fibre: Riemann} we have
	\begin{align*}
		R^{\varphi_{|_{\Sigma}}}_{AC} &= \frac{S^{\varphi}-R^\p_{mm}}{m-1}\delta_{AC} - \frac{1}{m-1} R^{\varphi}_{mm}\delta_{AC} + (m-2) H^2\delta_{AC}\\
		&= \bigg[ \frac{1}{m-1} (S^{\varphi} - 2 R^{\varphi}_{mm}) + (m-2) H^2\bigg]\delta_{AC}.
	\end{align*}
	Tracing the above equation we get \eqref{local struct: fibre: S phi sigma} and then \eqref{local struct: fibre: harm einstein 1 bis}.
\end{proof}

\begin{proof}[Proof of Proposition \ref{local struct: prop: fibre harm einstein}]
 	In view of Lemma \ref{local struct: lemma: W phi mamb} and Lemma \ref{lemma: local struct: ric phi sigma}, we only need to prove \eqref{local struct: fibre: harm einstein 2}.
 	 We first need to prove that $S^{\varphi_{|_{\Sigma}}}$ is constant.
 	Since $D^A \equiv 0$, by its definition, \eqref{Deduzione D A definizione}, we get
 	\begin{equation*}
 		  0 = R^{\varphi}_{ij} f_k - R^{\varphi}_{ik} f_j + \frac{1}{m-1}f_l (R^{\varphi}_{lk} \delta_{ij}- R^{\varphi}_{lj}\delta_{ik}) - \frac{S^{\varphi}}{m-1}(f_{k}\delta_{ij}-f_j\delta_{ik}).
 	\end{equation*}
 	Taking the divergence of the latter with respect to $i$ and using the $\varphi$-Schur identity we obtain
 	\begin{align*}
 		0 = & \; R^{\varphi}_{ij,i} f_k + R^{\varphi}_{ij} f_{ik} - R^{\varphi}_{ik,i} f_j - R^{\varphi}_{ik} f_{ij}  \\
 		&+ \frac{1}{m-1}f_{l} (R^{\varphi}_{lk,j} - R^{\varphi}_{lj,k}) + \frac{1}{m-1} (f_{lj}R^{\varphi}_{lk}  - f_{lk}R^{\varphi}_{lj}) \ \\
 		&- \frac{1}{m-1} (S^{\varphi}_{j}f_{k} - S^{\varphi}_kf_{j}) - \frac{1}{m-1} S^{\varphi}(f_{jk}-f_{kj})\\
 		= & \; \frac{1}{2} \big( S^{\varphi}_{j} f_k - S^{\varphi}_{k}f_j \big)- \alpha\p^{a}_{ll} \big( \p^{a}_j f_k - \p^{a}_k f_j \big) + \frac{m-2}{m-1}( R^{\varphi}_{ij}f_{ik}-R^\p_{ik}f_{ij})\\
 		&+ \frac{1}{m-1} f_l(R^{\varphi}_{lk,j}- R^{\varphi}_{lj,k}) - \frac{1}{m-1} (S^{\varphi}_{j}f_k - S^{\varphi}_kf_j).
 	\end{align*}
 	Using the definition of $C^{\varphi}$ and
 		$\tau (\p) = \frac{\nabla^h U}{\alpha} $
 	 we get
	\begin{align*}
		0 = & \; \Big[ \frac{1}{2} + \frac{1}{2(m-1)^2} - \frac{1}{m-1} \Big] \Big(S^{\varphi}_{j} f_k - S^{\varphi}_{k}f_j\Big) + \frac{m-2}{m-1}\Big(R^{\varphi}_{lj}f_{lk}-R^\p_{lk}f_{lj} \Big)\\
		&+ \frac{1}{m-1} f_l C^{\p}_{lkj} - U^a\pa{\p^a_jf_k-\p^a_kf_j}.
	\end{align*}
	From \eqref{local struct: 1 cond di int 1} and $D^A\equiv D^B\equiv 0$, we have
	\begin{equation*}
		f_l C^{\p}_{lkj} = \frac{1}{m-1}U^a\p^a_jf_k - \frac{1}{m-1}U^a\p^a_kf_j,
	\end{equation*}
	so that
	\begin{align*}
		0 = & \; \frac{(m-2)^2}{2(m-1)^2} \Big(S^{\varphi}_{j} f_k - S^{\varphi}_{k}f_j\Big) + \frac{m-2}{m-1}\Big(R^{\varphi}_{lj}f_{lk}-R^\p_{lk}f_{lj} \Big) + \frac{1-(m-1)^2}{(m-1)^2}U^a\pa{\p^a_j f_k-\p^a_k f_j}.
	\end{align*}
	From Proposition \ref{local struct: prop: D(P) è 0}, and since $D^A\equiv D^B\equiv 0$,	$\ric^{\varphi}$ and $\hess f$ have the same eigenspaces so that they commute, i.e.
		$R^{\varphi}_{lj}f_{lk} - R^{\varphi}_{lk}f_{lj} = 0.$
	Therefore,
	\begin{equation*}
		0 = \Big[ \frac{(m-2)^2}{2}S^{\varphi}_{j} + (1 - (m-1)^2)U^{a}\p^a_{j}\Big]f_{k} - \Big[ \frac{(m-2)^2}{2}S^{\varphi}_{k} + (1 - (m-1)^2)U^a\p^a_k \Big]f_j.
	\end{equation*}	
	Choosing $j=A$ and $k=m$ we get
	\begin{align*}
		0 = & \; \left[ \frac{(m-2)^2}{2}S^{\varphi}_{A} + (1 - (m-1)^2)U^{a}\p^a_{A}\right]f_{m} - \left[ \frac{(m-2)^2}{2}S^{\varphi}_{m} + (1 - (m-1)^2)U^a\p^a_m \right]f_A\\
		= & \; \abs{\nabla f}\left[ \frac{(m-2)^2}{2}S^{\varphi}_{A} + (1 - (m-1)^2)U^{a}\p^a_{A}\right], 
	\end{align*}
	so that
	\begin{equation}\label{local struct: fibre: S phi e U parte 1}
		 \quad 0 = \frac{(m-2)^2}{2}S^{\varphi}_{A} + (1 - (m-1)^2)U^{a}\p^a_{A}. 
	\end{equation}
	Differentiate \eqref{local struct: fibre: S phi sigma} to get
	\begin{align*}
		S^{\varphi_{|_{\Sigma}}}_{A} = & \; S^{\varphi}_{A} - 2R^{\varphi}_{mm,A} + 2(m-2)(m-1) H H_A\\
		= & \; S^{\varphi}_{A} - 2R^{\varphi}_{mm,A},
	\end{align*}
	where the last equality follows because $H$ is locally constant on $\Sigma$.
	Now,
	\begin{align*}
		R^{\varphi}_{mm,A} = & \; R_{mm,A} - 2 \alpha\p^a_m\p^a_{Am}\\
		= & \; R_{mm,A}
	\end{align*}
	since $\p^a_m=0$. Since $g$ is a warped product metric, $R_{mm}$ is a function of $r$ alone so that it is locally constant on $\Sigma$, see Chapter 1.7 of \cite{AMR}. Therefore $R_{mm,A}=0$, and $S^{\varphi_{|_{\Sigma}}}_{A} = S^{\varphi}_{A}$ so that \eqref{local struct: fibre: S phi e U parte 1} becomes
 	\begin{equation}\label{local struct: fibre: S phi e U parte 2}
 		0 = \frac{(m-2)^2}{2}S^{\varphi_{|_{\Sigma}}}_{A} + (1 - (m-1)^2)U^{a}\p^a_{A}.
 	\end{equation}
 	By a direct computation, we have 
 	\begin{align*}
 		\tau^a(\varphi_{|_{\Sigma}}) = & \; \tau^a(\varphi) - \p^a_{mm} + (m-1)\p^a_mH\\
 		= & \; \tau^a(\varphi) - \p_{mm}^{a}.
 	\end{align*}
	We want to prove that $\p^a_{mm} = 0$. Since $ \tau(\varphi) = \frac{\nabla^h U}{\alpha}$, \eqref{local struct: 1 cond di int 2} gives $\p_l^{a} f_l = 0$. Differentiating, since $f_{Am} = 0 = \p_m^{a}$,
 	\begin{align*}
 		0 = & \; \p_{lm}^{a}f_l  + \p_l^{a}f_{lm}\\
 		= & \; \p_{mm}^{a} |\nabla f| + \p^a_m f_{mm} + \p^a_A f_{Am}\\
 		= & \; \p_{mm}^{a} |\nabla f|
 	\end{align*}
 	so that
 	\begin{align}\label{local strcut: fibre: phi sigma è U armonica}
 	 \tau^a(\varphi_{|_{\Sigma}}) =  \tau(\varphi) = \frac{\nabla^h U}{\alpha}.
 	\end{align}
 	Taking the divergence of \eqref{local struct: fibre: harm einstein 1 bis}
 	and using the $\varphi$-Schur  identity we get
 	\begin{equation*}
 		R^{\varphi_{|_{\Sigma}}}_{AC,C} = \frac{1}{2} S^{\varphi_{|_{\Sigma}}}_{A} - \alpha \tau^a(\varphi_{|_{\Sigma}})\p_A^{a} = \frac{S^{\varphi_{|_{\Sigma}}}_{A}}{m-1}
 	\end{equation*}
 	so that
 	\begin{equation}\label{local struct: fibre: S phi e U parte 3}
 		\frac{m-2}{2(m-1)}S^{\varphi_{|_{\Sigma}}}_{A} =  \alpha\tau^a(\varphi_{|_{\Sigma}}) \p_{A}^{a} = U^{a}\p_A^a
 	\end{equation}
 	Combining with \eqref{local struct: fibre: S phi e U parte 2} we get
 	\begin{equation*}
 		\frac{(m-2)^2}{2} S^{\varphi_{|_{\Sigma}}}_{A} + (1-(m-1)^2 )\frac{m-2}{2(m-1)} S^{\varphi_{|_{\Sigma}}}_{A}=0;
 	\end{equation*}
 	since
 	\begin{equation*}
 		\frac{(m-2)^2}{2} + (1-(m-1)^2)\frac{m-2}{2(m-1)}=-\frac{(m-2)^2}{2(m-1)} <0
 	\end{equation*}
 	we obtain
 		$S^{\varphi_{|_{\Sigma}}}_{A}=0.$
 	From \eqref{local struct: fibre: S phi e U parte 2} we also have
 		$U^{a}\p_A^{a} = 0$
 	and from \eqref{local strcut: fibre: phi sigma è U armonica}
 	we deduce \eqref{local struct: fibre: harm einstein 2}.
 \end{proof}
 	\begin{proof}[Proof of Theorem \ref{local struct: teorema warp product}]
 		This follows from Propositions \ref{local struct: prop: D(P) è 0} and \ref{local struct: prop: fibre harm einstein} combined.
 	\end{proof}

\subsection{Proof of the main Theorem}
We finally prove Theorem \ref{Introduzione: codazzi: teo: teoremone} of the Introduction, that we restate for the convenience of the reader.
\begin{theorem}\label{codazzi: teo: teoremone}
	Let $(M,g)$ be a complete $m$-dimensional Riemannian manifold, $m\geq 3$, and let $\p:(M,g)\to (N,h)$ be a smooth map of Riemannian manifolds. Let $\alpha \in \RR,\alpha > 0, f\in C^\infty(M), U\in C^\infty(N)$. Assume that $(M,g)$ satisfies
	\begin{align*}
		\begin{cases}
			0 = C^\p_{ijk}+f_lW^\p_{lijk}-\dfrac{1}{m-1}U^a\pa{\p^a_k\delta_{ij}-\p^a_j\delta_{ik}}-D^A_{ijk}-(m-2)D^B_{ijk},\\[0.2cm]
			\tau(\p)-\di \p(\nabla f)=\dfrac{1}{\alpha} \nabla^h U.
		\end{cases}
	\end{align*}
	Let $S^2(M)$ be the space of $2$-covariant, symmetric tensors on $M$ and define a linear map $F:S^2(M)\to C^\infty(M)$ by setting, for $\beta\in S^2(M), \beta=\beta_{lk}\theta^l\otimes \theta^k$ locally,
	\begin{align}\label{codazzi: teoremone: funzionale F}
		F(\beta):=\pa{\frac{m}{m-1}U^a\p^a_lf_k+D^A_{lik}f_i-W^\p_{likj}f_if_j}\beta_{lk}.
	\end{align}
	Assume that
	\begin{enumerate}[1.]
		\item $f$ is proper;\label{Codazzi: main: f propria}
		\item $B^\p(\nabla f,\nabla f)=0$;\label{Codazzi: main: bach piatto}
		\item $\p$ is $\frac{1}{\alpha} U$-harmonic;\label{Codazzi: main: U armonica}
		\item for all $p\in M$ regular for $f$, we have that $\nabla f$ is an eigenvector of $\ric^\p$ at $p$;\label{Codazzi: main: autovett}
		\item $\ric^\p+\hess f\in \ker (F)$.\label{Codazzi: main: funzionale}
	\end{enumerate}
	Then, for each regular level set $\Sigma$ of $f$ and for every $p\in \Sigma$, there exists $A\subset M$ open such that $p\in A$ and $g_{|_A}$ is a warped product metric.  Moreover, $(\Sigma,g_{\Sigma})$ satisfies
	\begin{align*}
		\begin{cases}
			\ric^{\p_{|_{\Sigma}}}=\dfrac{S^{\p_{|_{\Sigma}}}}{m-1}g_\Sigma,\\[0.2cm]
			\tau(\p_{|_{\Sigma}})=0.
		\end{cases}
	\end{align*}
\end{theorem}
We divide the proof of Theorem \ref{codazzi: teo: teoremone} in several propositions.
First, we set some notations.
For an arbitrary smooth function $\lambda$ on $M$, set
\begin{align}\label{codazzi 1: tensore zeta}
	Z:=\ric^\p+\hess f-\di f\otimes \di f-\lambda g
\end{align}
and
\begin{align}\label{codazzi 1: D barrato}
	\overline{D}_{ijk}=Z_{ij,k}-Z_{ik,j}-\frac{1}{m-1}\sq{\pa{Z_{ll,k}-Z_{lk,l}}\delta_{ij}-\pa{Z_{ll,j}-Z_{lj,l}}\delta_{ik}}.
\end{align}
Clearly, we have that $\overline{D}\equiv 0$ if and only if $Z$ is a Codazzi tensor. Moreover, an easy computation gives the following
\begin{lemma}\label{codazzi 1:lemma: D bar e prima cond}
	Let $Z$ and $\overline{D}$ be defined as in \eqref{codazzi 1: tensore zeta} and \eqref{codazzi 1: D barrato}. Then we have
	\begin{equation}\label{codazzi 1: prima cond int con D bar}
		D^A_{ijk} +(m-2)D^B_{ijk}+\overline{D}_{ijk} = C^\p_{ijk}+f_lW^\p_{lijk}-\frac{\alpha}{m-1}\pa{\p^a_{ll}-\p^a_lf_l}\pa{\p^a_k\delta_{ij}-\p^a_j\delta_{ik}}.
	\end{equation}
\end{lemma}
In particular, we have that \eqref{Deduzione: C-phi-PF: versione con le D 1} holds if $Z$ is Codazzi and \eqref{Deduzione: C-phi-PF: versione con le D 3} holds.
See Appendix \ref{Appendice: tensori di Codazzi} for a proof of Lemma \ref{codazzi 1:lemma: D bar e prima cond}. We compute the divergence of \eqref{codazzi 1: prima cond int con D bar}.
\begin{proposition}\label{teoremone: prop: divergenza di D bar e co}
	Assume that
	\begin{equation}\label{teoremone: prop: seconda del sistema}
		\tau(\p)-\di \p(\nabla f)=\frac{\nabla^h U}{\alpha}.
	\end{equation}
	Then we have
	\begin{align}\label{teoremone: prop: divergenza, formula}
		D^A_{ijk,k}&+(m-2)D^B_{ijk,k}+\overline{D}_{ijk,k} \nonumber\\[0.2cm]
		= & \; (m-2)B^\p_{ij}-(R^\p_{lk}+f_{lk})W^\p_{likj}+\alpha (R^\p_{lj}+f_{lj})\p^a_l\p^a_i \nonumber\\
		&-\frac{m-2}{m-1}U^a\p^a_{ij}+\frac{m}{m-1}U^{ab}\p^a_i\p^b_j+\frac{1}{(m-1)(m-2)}U^a\p^a_{ll}\delta_{ij}\nonumber\\
		&+\frac{m-3}{m-2}f_lC^\p_{jli}+\frac{\alpha}{m-2}f_j\p^a_i\p^a_{ll}-\frac{1}{m-1}U^{ab}\p^a_k\p^b_k \delta_{ij}. 
	\end{align}
\end{proposition}

\begin{proof}
	Use \eqref{teoremone: prop: seconda del sistema} into \eqref{codazzi 1: prima cond int con D bar} and take the divergence to get
	\begin{align}\label{teoremone: prop: conto 1}
		D^A_{ijk,k}+(m-2)D^B_{ijk,k} +\overline{D}_{ijk,k} = & \; C^\p_{ijk,k}+f_{lk}W^\p_{lijk}+f_lW^\p_{lijk,k} \\[0.1cm]
		& -\frac{1}{m-1}U^{ab}\p^b_k\pa{\p^a_k\delta_{ij}-\p^a_j\delta_{ik}} \\ \nonumber
		& -\frac{1}{m-1}U^a\pa{\p^a_{kk}\delta_{ij}-\p^a_{ij}}. \nonumber
	\end{align}
	We use equations \eqref{phi curv: divergenza di W phi} and \eqref{phi curv: phi bach}
	into \eqref{teoremone: prop: conto 1} to infer
	\begin{align}\label{teoremone: prop: conto 2}
		D^A_{ijk,k} +(m-2)D^B_{ijk,k}+\overline{D}_{ijk,k} = & \; (m-2)B^\p_{ij}-R^\p_{lk}W^\p_{likj}+\alpha R^\p_{lj}\p^a_l\p^a_i\\ \nonumber
		&-\alpha \pa{\p^a_{ij}\p^a_{ll}-\p^a_i\p^a_{llj}-\frac{1}{m-2}\abs{\tau(\p)}\delta_{ij}}-f_{lk}W^\p_{likj}\\ \nonumber
		&+\frac{m-3}{m-2}f_lC^\p_{jli}+\alpha f_l\pa{\p^a_{ij}\p^a_l-\p^a_{jl}\p^a_i}\\ \nonumber
		&+\frac{\alpha}{m-2}f_l\p^a_{pp}\pa{\p^a_i\delta_{lj}-\p^a_l\delta_{ij}}\\ \nonumber
		&-\frac{1}{m-1}U^{ab}\p^b_k\pa{\p^a_k\delta_{ij}-\p^a_j\delta_{ik}}-\frac{1}{m-1}U^a\pa{\p^a_{kk}\delta_{ij}-\p^a_{ij}}.
	\end{align}
	Taking the covariant derivative of \eqref{teoremone: prop: seconda del sistema} we get
	\begin{align}\label{teoremone: prop: der cov seconda sistema}
		\p^a_{llj}-\p^a_{lj}f_l-\p^a_lf_{lj}=\frac{1}{\alpha}U^{ab}\p^b_j.
	\end{align}
	Using \eqref{teoremone: prop: der cov seconda sistema} and \eqref{teoremone: prop: seconda del sistema} into \eqref{teoremone: prop: conto 2} we obtain
	\begin{align*}
		D^A_{ijk,k}&+(m-2)D^B_{ijk,k}+\overline{D}_{ijk,k}\\
		=&\;(m-2)B^\p_{ij}-(R^\p_{lk}+f_{lk})W^\p_{likj}+\alpha R^\p_{lj}\p^a_l\p^a_i\\
		&-U^a\p^a_{ij}+U^{ab}\p^a_i\p^b_j+\alpha f_{lj}\p^a_l\p^a_i+\frac{\alpha}{m-2}\abs{\tau(\p)}^2\delta_{ij}\\
		&+\frac{m-3}{m-2}f_lC^\p_{jli}+\frac{\alpha}{m-2}\p^a_{ll}\p^a_if_j -\frac{\alpha}{m-2}f_l\p^a_l\p^a_{pp}\delta_{ij}\\
		&-\frac{1}{m-1}U^{ab}\p^b_k\pa{\p^a_k\delta_{ij}-\p^a_j\delta_{ik}}-\frac{1}{m-1}U^a\pa{\p^a_{kk}\delta_{ij}-\p^a_{ij}}\\
		=&\;(m-2)B^\p_{ij}-(R^\p_{lk}+f_{lk})W^\p_{likj}+\alpha (R^\p_{lj}+f_{lj})\p^a_l\p^a_i\\
		&-\frac{m-2}{m-1}U^a\p^a_{ij}+\frac{m}{m-1}U^{ab}\p^a_i\p^b_j+\frac{1}{m-2}U^a\p^a_{ll}\delta_{ij}\\
		&+\frac{m-3}{m-2}f_lC^\p_{jli}+\frac{\alpha}{m-2}\p^a_{ll}\p^a_if_j-\frac{1}{m-1}U^{ab}\p^a_k\p^b_k\delta_{ij}-\frac{1}{m-1}U^a\p^a_{ll}\delta_{ij}.
		\end{align*}
	After some simplifications, the latter becomes \eqref{teoremone: prop: divergenza, formula}.
\end{proof}

Let
\begin{align}\label{teoremone: campo Y}
	Y_k:=f_if_j\pa{D^A_{ijk}+(m-2)D^B_{ijk}+\overline{D}_{ijk}}-\frac{1}{m-1}U^a\pa{\p^a_jf_jf_k-\abs{\nabla f}^2\p^a_k};
\end{align}
then we have
\begin{proposition}\label{teoremone: prop: divergenza di Y}
	With the assumptions of Proposition \ref{teoremone: prop: divergenza di D bar e co} and with $Y$ defined in \eqref{teoremone: campo Y}, we have
	\begin{align}\label{teoremone: divergenza di Y}
		\div Y = & \; \frac{m-2}{2}\pa{\abs{D^A}^2+(m-2)\abs{D^B}^2}+D^A_{ijk}\pa{f_{ik}f_j+R^{\p}_{ik}f_j}\\ \nonumber
		&+\overline{D}_{ijk}f_jf_{ik}+(m-2)B^\p_{ij}f_if_j-\pa{R^\p_{lk}+f_{lk}}W^\p_{likj}f_if_j\\ \nonumber
		&+\alpha (R^\p_{lj}+f_{lj})\p^a_l\p^a_if_if_j-U^a\p^a_{ij}f_if_j+U^{ab}\p^a_i\p^b_jf_if_j\\ \nonumber
		&+\frac{\abs{\nabla f}^2}{m-2}U^a\p^a_{jj}+\frac{\alpha}{m-2}\abs{\nabla f}^2\p^a_{jj}\p^a_if_i\\ \nonumber
		&+\frac{1}{m-1}U^a\p^a_j f_k f_{jk}-\frac{1}{m-1}U^a\p^a_jf_j\Delta f. 
	\end{align}
\end{proposition}
\begin{proof}
Taking the divergence of $Y$
\begin{align*}
	\div Y = & \; f_{ik}f_j\pa{D^A_{ijk}+(m-2)D^B_{ijk}+\overline{D}_{ijk}}+f_if_{jk}\pa{D^A_{ijk}+(m-2)D^B_{ijk}+\overline{D}_{ijk}}\\
	&+f_if_j\pa{D^A_{ijk,k}+(m-2)D^B_{ijk,k}+\overline{D}_{ijk,k}}-\frac{1}{m-1}U^{ab}\p^b_k\pa{\p^a_jf_jf_k-\abs{\nabla f}^2\p^a_k}\\
	&-\frac{1}{m-1}U^a\pa{\p^a_{jk}f_jf_k+\p^a_jf_{jk}f_k+\p^a_j f_j\Delta f-\abs{\nabla f}^2\p^a_{kk}-2\p^a_kf_{jk}f_j};
\end{align*}
using the symmetries of $D^A, D^B$ and $\overline{D}$ we get
\begin{align}\label{teoremone: prop 2: conto 1}
	\div Y = & \; f_{ik}f_j\pa{D^A_{ijk}+(m-2)D^B_{ijk}+\overline{D}_{ijk}}+f_if_j\pa{D^A_{ijk,k}+(m-2)D^B_{ijk,k}+\overline{D}_{ijk,k}}\\ \nonumber
	&-\frac{1}{m-1}U^{ab}\p^b_k\pa{\p^a_jf_jf_k-\abs{\nabla f}^2\p^a_k}\\ \nonumber
	&-\frac{1}{m-1}U^a\pa{\p^a_{jk}f_jf_k-\p^a_jf_{jk}f_k+\p^a_jf_j\Delta f-\abs{\nabla f}^2\p^a_{kk}}.
\end{align}
Using the definition of $D^B$, its symmetries and the fact that it is totally trace-free we get
\begin{align}\label{teoremone: prop 2: norma D B}
	f_{ik}f_j D^B_{ijk}=\frac{1}{2}\pa{f_{ik}f_j-f_{ij}f_k}D^B_{ijk}=\frac{m-2}{2}\abs{D^B}^2.
\end{align}
In the same fashion we have
\begin{align*}
	R^\p_{ij}f_kD^A_{ijk}=\frac{m-2}{2}\abs{D^A}^2,
\end{align*}
so that
\begin{align}\label{teoremone: prop 2: norma di D A}
	f_{ik}f_jD^A_{ijk}&=\pa{f_{ik}f_j+R^\p_{ik}f_j}D^A_{ijk}-R^\p_{ik}f_j D^A_{ijk}\\ \nonumber
	&=\pa{f_{ik}f_j+R^\p_{ik}f_j}D^A_{ijk}+R^\p_{ik}f_j D^A_{ikj}\\ \nonumber
	&=\frac{m-2}{2}\abs{D^A}^2+\pa{f_{ik}f_j+R^\p_{ik}f_j}D^A_{ijk}.
\end{align}
Using \eqref{teoremone: prop 2: norma di D A}, \eqref{teoremone: prop 2: norma D B} and \eqref{teoremone: prop: divergenza, formula} in \eqref{teoremone: prop 2: conto 1} we get
\begin{align*}
	\div Y = & \; \frac{m-2}{2}\abs{D^A}^2+\frac{(m-2)^2}{2}\abs{D^B}^2+D^A_{ijk}\pa{f_{ik}f_j+R^\p_{ik}f_j}\\
	&+f_{ik}f_j\overline{D}_{ijk}+f_if_j\sq{(m-2)B^\p_{ij}-\pa{R^\p_{lk}+f_{lk}}W^\p_{likj}}\\
	&+f_if_j\sq{\alpha\pa{R^\p_{lj}+f_{lj}}\p^a_l\p^a_i-\frac{m-2}{m-1}U^a\p^a_{ij}+\frac{m}{m-1}U^{ab}\p^a_i\p^b_j}\\
	&+f_if_j\sq{\frac{1}{(m-1)(m-2)}U^a\p^a_{ll}\delta_{ij}+\frac{m-3}{m-2}f_lC^\p_{jli}+\frac{\alpha}{m-2}\p^a_{ll}\p^a_if_j}\\
	&-\frac{\abs{\nabla f}^2}{m-1}U^{ab}\p^a_k\p^b_k-\frac{1}{m-1}U^{ab}\p^b_k\pa{\p^a_jf_jf_k-\abs{\nabla f}^2\p^a_k}\\
	&-\frac{1}{m-1}U^a\pa{\p^a_{jk}f_jf_k-\p^a_jf_{jk}f_k+\p^a_jf_j\Delta f-\abs{\nabla f}^2\p^a_{kk}}.
\end{align*}
Rearranging terms we get \eqref{teoremone: divergenza di Y}.
\end{proof}

\begin{proposition}\label{teoremone: proposizione 3}
With the assumptions of Proposition \ref{teoremone: prop: divergenza di D bar e co}, assume further that $\overline{D}\equiv 0$ and that $\p$ is $\frac{1}{\alpha}U$-harmonic, that is,
\begin{align}\label{teoremone: prop 3: U armonicità}
	\tau(\p)=\frac{\nabla^h U}{\alpha}.
\end{align}
Then we have
\begin{align}\label{teoremone: prop 3: divergenza di Y}
	\div Y = & \; (m-2)B^\p_{ij}f_if_j+\frac{\abs{\nabla f}^2}{\alpha(m-2)}\abs{\nabla^h U}^2+\frac{m-2}{2}\pa{\abs{D^A}^2+(m-2)\abs{D^B}^2}\\ \nonumber
	&-\frac{m}{m-1}U^a\p^a_lf_kR^\p_{kl}+\left\{\frac{m}{m-1}U^a\p^a_lf_k+D^A_{lik}f_i-W^\p_{likj}f_if_j\right\}\pa{R^\p_{lk}+f_{lk}}.
\end{align}
\end{proposition}

\begin{proof}
From \eqref{teoremone: prop 3: U armonicità} and \eqref{teoremone: prop: seconda del sistema} we get
\begin{align*}
	\di \p(\nabla f)=0,
\end{align*}
that is, in components
\begin{align}\label{teoremone: prop 3: d phi grad f}
\p^a_if_i=0.
\end{align}
Therefore, from \eqref{teoremone: divergenza di Y} and $\overline{D}\equiv 0$,
\begin{align*}
	\div Y = & \; \frac{m-2}{2}\pa{\abs{D^A}^2+(m-2)\abs{D^B}^2}+D^A_{ijk}\pa{f_{ik}f_j+R^\p_{ik}f_j}\\
	&+(m-2)B^\p_{ij}f_if_j-\pa{R^\p_{lk}+f_{lk}}W^\p_{likj}-U^a\p^a_{ij}f_if_j\\
	&+\frac{\abs{\nabla f}^2}{m-2}U^a\p^a_{ll}+\frac{1}{m-1}U^a\p^a_jf_kf_{jk}.\\
\end{align*}
Taking the covariant derivative of \eqref{teoremone: prop 3: d phi grad f} we obtain
\begin{align*}
	\p^a_{ij}f_j+\p^a_jf_{ij}=0
\end{align*}
so that
\begin{align*}
	\frac{1}{m-1}U^a\p^a_jf_{jk}f_k-U^a\p^a_{ij}f_if_j = & \; \frac{1}{m-1}U^a\p^a_jf_{jk}f_k+U^a\p^a_jf_{ij}f_i\\
	= & \; \frac{m}{m-1}U^a\p^a_j f_{jk}f_k;
\end{align*}
using also \eqref{teoremone: prop 3: U armonicità} we get \eqref{teoremone: prop 3: divergenza di Y}.
\end{proof}

\begin{proof}[Proof of Theorem \ref{codazzi: teo: teoremone}]
	By assumptions \ref{Codazzi: main: bach piatto} and \ref{Codazzi: main: funzionale} we have
 that \eqref{teoremone: prop 3: divergenza di Y} becomes
	\begin{align*}
		\div Y = & \; \frac{m-2}{2}\pa{\abs{D^A}^2+(m-2)\abs{D^B}^2}-\frac{m}{m-1}U^a\p^a_jR^\p_{jk}f_k + \frac{\abs{\nabla f}^2}{\alpha(m-2)}\abs{\nabla^h U}^2.
	\end{align*}
	Since $\nabla f$ is an eigenvector of $\ric^\p$ we get, at any regular point $p\in M$ of $f$,
	\begin{align*}
		U^a\p^a_jR^\p_{jk}f_k=\Lambda U^a\p^a_j f_j, \ \ \text{ for some } \Lambda\in \RR.
	\end{align*}
	Since $\p$ is $\frac{U}{\alpha}$-harmonic, we get
		$\p^a_if_i=0$
	and therefore
		$U^a\p^a_jR^\p_{jk}f_k=0$,	so that
	\begin{align}\label{teoremone: dim: div Y finale}
		\div Y=\frac{m-2}{2}\pa{\abs{D^A}^2+(m-2)\abs{D^B}^2}+\frac{\abs{\nabla f}^2}{\alpha(m-2)}\abs{\nabla^h U}^2.
	\end{align}
	Note that, by the definition \eqref{teoremone: campo Y} of $Y$, we get
	\begin{align}\label{teoremone: dim: Y e grad f ortogonali}
		Y_kf_k=0.
	\end{align}
	Let $\delta,\eta\in \RR$, $\delta<\eta$ be two regular values of $f$.
	Set
	\begin{align*}
		\Omega_{\delta,\eta}=\left\{x\in M: \delta\leq f(x)\leq \eta\right\};
	\end{align*}
	then $\Omega_{\delta,\eta}$ is compact since $f$ is proper. Integrating \eqref{teoremone: dim: div Y finale} on $\Omega_{\delta,\eta}$, using the divergence theorem and \eqref{teoremone: dim: Y e grad f ortogonali} we obtain
	\begin{align*}
		0=\int_{\partial \Omega_{\delta,\eta}}Y_k\nu_k=\int_{\Omega_{\delta,\eta}}\frac{m-2}{2}\pa{\abs{D^A}^2+(m-2)\abs{D^B}^2}+\frac{\abs{\nabla f}^2}{\alpha(m-2)}\abs{\nabla^h U}^2
	\end{align*}
	where $\nu=-\frac{\nabla f}{\abs{\nabla f}}$ is the inward unit normal to $\partial \Omega_{\delta,\eta}$. Letting $\delta\to -\infty, \eta\to +\infty$, since $\alpha>0$, we get $D^A=D^B=0$ on $M$ and $\nabla^h U=0$ at any regular point. Since $\p$ is $\frac{U}{\alpha}$-harmonic, we deduce that $\p$ is harmonic at every point and since $D^A\equiv D^B\equiv 0$ we can use Theorem \ref{local struct: teorema warp product} to conclude.
\end{proof}
\begin{remark}\label{Remark U=0}
	When $U\equiv 0$, some of the assumptions of Theorem \ref{codazzi: teo: teoremone} are unnecessary. Indeed, when $U\equiv 0$, Proposition \ref{teoremone: prop: divergenza di Y} implies
	\begin{align*}
		\div Y = & \; \frac{m-2}{2}\pa{\abs{D^A}^2+(m-2)\abs{D^B}^2}+\overline{D}_{ijk}f_j f_{ik}\\
		&+(m-2)B^\p_{ij}f_if_j+\frac{\alpha}{m-2}\abs{\nabla f}^2\abs{\tau(\p)}^2\\
		&+\pa{R^\p_{ik}+f_{ik}}\sq{f_j D^A_{ijk}-f_lf_j W^\p_{ilkj}+\alpha \p^a_i\p^a_lf_l f_k}
	\end{align*}
for the same vector field $Y$.
If, instead of studying the map $F$ given by \eqref{codazzi: teoremone: funzionale F}, one asks that $\ric^\p+\hess f\in \ker G$, where $G$ is the linear map $G: S^2(M)\to C^\infty(M)$ given by
\begin{align*}
	G(\beta):=\sq{f_j D^A_{ijk}-f_lf_j W^\p_{ilkj}+\alpha \p^a_i\p^a_lf_l f_k}\beta_{ik},
\end{align*}
then the same conclusion of Theorem \ref{codazzi: teo: teoremone} is obtained following the same procedure, without the need of assumptions \ref{Codazzi: main: U armonica} and \ref{Codazzi: main: autovett}.
\end{remark}

 \section{Riemann compatibility}\label{Sect 5: riem comp}

The following definition has been given by Mantica and Molinari \cite{Mantica2012RiemannCT}, \cite{Mantica2012WeylCT}:
\begin{definition}
	Let $(M,g)$ be a (pseudo-)Riemannian manifold and $P$ be a symmetric 2-covariant tensor. We say that $P$ is \emph{Riemann compatible} if it holds
	\begin{align}\label{riem comp: def di riem comp}
		R^l_{\ iks}P_{jl}+R^l_{\ isj}P_{kl}+R^l_{\ ijk}P_{sl}=0.
	\end{align}
\end{definition}
Contracting \eqref{riem comp: def di riem comp} with the metric and using the orthogonal decomposition of $\Riem$ one easily sees that \eqref{riem comp: def di riem comp} is equivalent to the system
\begin{align}\label{riem comp: weyl + commut}
	\begin{cases}
		W^l_{\ iks}P_{jl}+W^l_{\ isj}P_{kl}+W^l_{\ ijk}P_{sl}=0,\\
		P_{jl}R^l_{\ i}=P_{il}R^l_{\ j},
	\end{cases}
\end{align}
in particular $P$ is Riemann compatible if and only if it is Weyl compatible and it commutes with the Ricci tensor of $(M,g)$.
Every Codazzi tensor is Riemann compatible as the next proposition shows.
\begin{proposition}[Proposition 2.2 of \cite{Mantica2012RiemannCT}]
	Let $(M,g)$ be a pseudo-Riemannian manifold and $P$ be a 2-covariant symmetric tensor on $M$.
	Set
	\begin{align*}
		C(P)_{ijk}:=P_{ij,k}-P_{ik,j}.
	\end{align*}
	Then we have
	\begin{equation}\label{riem comp: riem comp e codazzi dev}
		C(P)_{ijk,l}+C(P)_{ikl,j}+C(P)_{ilj,k} = R^s_{\ ikl}P_{js}+R^s_{\ ilj}P_{ks}+R^s_{\ ijk}P_{ls}.
	\end{equation}
\end{proposition}
\begin{proof}
	By the definition of $C(P)$, the left hand side of \eqref{riem comp: riem comp e codazzi dev} is
	\begin{align*}
		P_{ij,kl}-P_{ik,jl}+P_{ik,lj}-P_{il,kj}+P_{il,jk}-P_{ij,lk} =: (\ast) .
	\end{align*}
Rearranging terms  and using the Ricci commutation relations, we get
	\begin{align*}
		(\ast) & = P_{ij,kl}-P_{ij,lk}+P_{ik,lj}-P_{ik,jl}+P_{il,jk}-P_{il,kj} \\
		& = R^s_{\ ikl}P_{js}+R^s_{\ ilj}P_{ks}+R^s_{\ ijk}P_{ls} + R^s_{\ jkl}P_{is}+R^s_{\ klj}P_{is}+R^s_{\ ljk}P_{is}.
	\end{align*}
	Using the first Bianchi identity the result follows.
\end{proof}
\begin{corollary}
	Every Codazzi tensor is Riemann compatible.
\end{corollary}
\begin{remark}\label{riem comp: riem comp e termini proporz alla metrica}
	Note that being Riemann compatible is a more general condition than being Codazzi: for example, assuming that $P$ is Riemann compatible, then the tensor
	\begin{align*}
		P-\lambda g,
	\end{align*}
	where $\lambda$ is an arbitrary smooth function on $M$, is still Riemann compatible, as a consequence of the first Bianchi identity.
	However, if $P$ is a Codazzi tensor, $P-\lambda g$ will not be so in general.
\end{remark}
A simple computations reveals that every $n$-dimensional Lorentzian manifold $(\hat{M},\hat{g})$ that solves Harada's field equations for a stress-energy tensor $\hat{T}$ admits a Codazzi tensor, namely
\begin{align*}
	\hat{Z}:=\hat{\ric}-\hat{T}-\frac{\hat{S}-2\text{tr}_{\hat{g}}\hat{T}}{2(n-1)}\hat{g}.
\end{align*}
Moreover, every $m$-dimensional C-$\p$-PF admits a Codazzi tensor, that is,
\begin{align*}
	\mathscr{C}:=\ric^\p+\hess f-\di f\otimes \di f-\frac{S^\p+2\Delta_f f+2U(\p)+2\mu}{2m}g.
\end{align*}
See Appendix \ref{Appendice: tensori di Codazzi} for more details.
Therefore, from Remark \ref{riem comp: riem comp e termini proporz alla metrica}, both $\mathscr{C}$ and $\ric^\p+\hess f-\di f\otimes \di f$ are Riemann compatible.
\begin{proposition}\label{riem comp: prop: C phi PF e riem comp}
	Let $(M,g)$ be an $m$-dimensional Riemannian manifold. Assume that there exists $f\in C^\infty(M)$ that solves the system of C-$\p$-PF. Then
	\begin{align*}
		\ric^\p+\hess f-\di f\otimes \di f
	\end{align*}
	is Riemann compatible.
\end{proposition}
In the case of a $\p$-SPFST, the tensor
\begin{align*}
	\ric^\p+\hess f-\di f\otimes \di f
\end{align*}
is proportional to the metric so that it is clearly Riemann compatible.
This suggested to us that, in some way, one might try to use Riemann compatibility (and Weyl compatibility) to study how far solutions of the Cotton gravity equations are from  solutions to the Einstein field equations, therefore giving a partial answer to Question \ref{question: domanda 2} of the Introduction.
In this direction, in the next subsection, inspired by works of Gover, Nurowski and Nagy, \cite{GoverA.Rod2006OtcE},\cite{GoverA.Rod2007FCC}, we show how the existence of a non-trivial Weyl-compatible tensor constrains the algebraic structure of the Weyl tensor. This will give us natural conditions on the Weyl tensor under which every solution of the C-$\p$-PF equations gives rise to a solution of a system that is strictly related to the $\p$-SPFST system. In Appendix \ref{Appendice B}, we further deepen this discussion by specializing to the case of dimension 4.
\subsection{Weyl compatibility}
Let us introduce some notations. Let $(M,g)$ be an $m$-dimensional Riemannian manifold and assume that it is orientable; let $\epsilon$ be its volume form. In an orthonormal co-frame, its components are given by the so called \emph{Levi-Civita symbol}. Define an $m$-times covariant tensor
\begin{align*}
	W^*_{i_1i_2...i_{m-2}jk}:=\epsilon_{i_1i_2...i_{m-2}pq}W^{pq}_{\ \ jk}.
\end{align*}
Let $S^2_0(M)$ be the space of traceless symmetric 2-covariant tensors on $M$, while $T^{(0,m-2)}M$ will be the space of $(m-2)$-covariant tensors on $M$. Then $W^*$ induces a linear map
\begin{align*}
	W^*_{|_{S^2_0}}:S^2_0(M)\to T^{(0,m-2)}M
\end{align*}
by sending $P\in S^2_0(M)$ of components $P_{ij}$ to the tensor $W^*_{|_{S^2_0}}(P)$ of components
\begin{align*}
	W^*_{|_{S^2_0}}(P)_{i_2...i_{m-2}j}=W^*_{li_2..i_{m-2}kj}P_{lk}.
\end{align*}
When $m=4$, as it is well-known, $W^*$ is an algebraic Weyl tensors and $W^*_{|_{S^2_0}}$ is an endomorphism of $S^2_0(M)$. It is clear how, in a similar way, the Weyl tensor also gives rise to a linear map $W_{|_{S^2_0}}$ which, in any dimension $m\geq 4$, is actually an endomorphism of $S^2_0(M)$.
The following simple lemma is probably well-known among experts.
\begin{lemma}\label{riem comp: W star e weyl comp}
	Let $(M,g)$ be an $m$-dimensional orientable Riemannian manifold and let $P$ belong to $S^2_0(M)$. Then $P$ lies in the kernel of $W^*_{|_{S^2_0}}$ if and only if $P$ is Weyl compatible.
\end{lemma}
\begin{proof}
	Let $p\in M$ be fixed and let $\{\theta^i\}_{i=1}^m$ be a positively oriented orthonormal co-frame at $p$. Then, at $p$, we have
	\begin{align*}
		\epsilon=\theta^1\wedge \theta^2\wedge...\wedge \theta^m.
	\end{align*}
	First, note that, by the obvious symmetries of $\epsilon$ and by the definition of $W^*$ we have
	\begin{align*}
		W^*_{i_1i_2...i_{m-2}jq}=0
	\end{align*}
	whenever there are two coinciding indices among $i_1,i_2,...,i_{m-2}$.
	Assume now that $i_1,i_2,...,i_m$ are distinct indices. Suppressing the Einstein summation convention for the rest of this proof, we have
	\begin{align*}
		\sum_{l,q}P_{lq}&W^*_{li_2...i_{m-2}jq}\\
		&=\sum_q\pa{P_{i_1q}W^*_{i_1i_2...i_{m-2}jq}+P_{i_{m-1}q}W^*_{i_{m-1}i_2...i_{m-2}jq}+P_{i_mq}W^*_{i_mi_2...i_{m-2}jq}}\\
		&=\pm\sum_q\pa{P_{i_1q}W_{i_{m-2}i_m jq}+P_{i_{m-1}q}W_{i_{m}i_{1}jq}+P_{i_mq}W_{i_1i_{m-1}jq}}
	\end{align*} 
	where the last equality follows from the definition of $W^*$ and the sign $\pm$ coincides with the determinant of $\theta^{i_1}\wedge \theta^{i_2}\wedge...\wedge \theta^{i_m}$.
	Using the symmetries of the Weyl tensor we get
	\begin{align*}
		\sum_{l,q}P_{lq}&W^*_{li_2...i_{m-2}jq}\\
		&=\mp \sum_q\pa{P_{i_1q}W_{qji_{m-2}i_m}+P_{i_{m-1}q}W_{qji_{m}i_{1}}+P_{i_mq}W_{qji_1i_{m-1}}}
	\end{align*}
	which gives the desired result.
\end{proof}
From Lemma \ref{riem comp: W star e weyl comp} we deduce that if $W^*_{|_{S^2_0}}$ is injective, then there exists no non-zero trace-free Weyl compatible tensor.

In the next Theorem, we apply the above discussion  to the setting of C-$\p$-PF, showing how, under suitable conditions on the Weyl tensor, every solution of the C-$\p$-PF system  solves a strict relative of the $\p$-SPFST system.
\begin{theorem}\label{riem comp: risultato principale}
	Let $(M,g)$ be an orientable Riemannian manifold; assume that there exists $f\in C^{\infty}(M)$ that solves the C-$\p$-PF system. Let  $W^*_{|_{S^2_0}}$ be injective. Then, there exist two constants $\Lambda_1,\Lambda_2\in \RR$ such that $f$ satisfies
	\begin{subequations}\label{Cphi e phiPF: integrazione: conclusione bis}
		\begin{empheq}[left={\empheqlbrace}]{alignat=2}
			&\ric^\p+\hess f-\di f\otimes \di f=\lambda g,\label{Cphi e phiPF: integrazione: conclusione bis 1}\\
			&\Delta_f f=\frac{1}{m-1}\sq{S^\p-m(\mu+p)}+\Lambda_1e^f, \label{Cphi e phiPF: integrazione: conclusione bis 2}\\
			&\tau(\p)=\di \p(\nabla f)+\frac{\nabla^h U}{\alpha},\label{Cphi e phiPF: integrazione: conclusione bis 3}\\
			&\frac{1}{2}S^\p=U(\p)+\mu+\Lambda_2,\label{Cphi e phiPF: integrazione: conclusione bis 4}\\[0.2cm]
			&(\mu+p)\nabla f-\nabla p=0, \label{Cphi e phiPF: integrazione: conclusione bis 5}
		\end{empheq}
	\end{subequations}
	where $\lambda=\frac{1}{m}\pa{S^\p+\Delta_f f}$.
\end{theorem}

\begin{remark}
	Assume that \eqref{Cphi e phiPF: integrazione: conclusione bis} holds and consider the substitution
	\begin{align}\label{Cphi e phiPF: integrazione: sostituzione}
		\begin{cases}
			\tilde{\mu} =\mu+\Lambda_2, \\[0.2cm]
			\tilde{p} = p - \dfrac{m-1}{m}\Lambda_1e^f-\Lambda_2.
		\end{cases}
	\end{align}
	It is easy to see that equations \eqref{Cphi e phiPF: integrazione: conclusione bis 2}, \eqref{Cphi e phiPF: integrazione: conclusione bis 4} and \eqref{Cphi e phiPF: integrazione: conclusione bis 5} reduce to equations \eqref{Cphi e phiPF: sistema phiSPFST riscritto 2}, \eqref{Cphi e phiPF: sistema phiSPFST riscritto 4} and \eqref{Cphi e phiPF: sistema phiSPFST riscritto 5} respectively, with $\tilde{\mu}$ and $\tilde{p}$ in place of $\mu$ and $p$. This implies that every solution of \eqref{Cphi e phiPF: integrazione: conclusione bis} is a $\p$-SPFST, for a different choice of pressure and density, unless $\Lambda_1=\Lambda_2=0$.
\end{remark}
\begin{proof}[Proof of Theorem \ref{riem comp: risultato principale}]
	From Proposition \ref{riem comp: prop: C phi PF e riem comp}, we have that
	\begin{align*}
		P:=\ric^\p+\hess f-\di f\otimes \di f
	\end{align*}
	is a Riemann compatible tensor, so that it is also Weyl compatible. From Lemma \ref{riem comp: W star e weyl comp} we deduce that
	\begin{align*}
		W^*_{|_{S^2_0}}(\mathring P)=0
	\end{align*}
	where $\mathring P$ denotes the trace-less part of $P$. Therefore, since $W^*_{|_{S^2_0}}$ is injective, there exists $\lambda\in C^{\infty}(M)$ such that
	\begin{align*}
		\ric^\p+\hess f-\di f\otimes \di f=\lambda g.
	\end{align*}
	From the equation above and the equations of C-$\p$-PF
	we deduce that it holds
	\begin{subequations}\label{Cphi e phiPF: intgerazione: sistemone da integrare}
		\begin{empheq}[left={\empheqlbrace}]{alignat=2}
			&R^\p_{ij}+ f_{ij}-f_if_j=\lambda \delta_{ij}, \label{Cphi e phiPF: intgerazione: sistemone da integrare 1}\\[0.2cm]
			&0=-\frac{m-1}{m}f_{lli}+\frac{m-2}{m}f_lf_{li}+\Delta f f_i+\frac{1}{2m}S^\p_i -f_lR^\p_{li}+\frac{1}{m}U^a\p^a_i-\frac{m-1}{m}\mu_i, \label{Cphi e phiPF: intgerazione: sistemone da integrare 2}\\
			&\tau(\p)=\di \p(\nabla f)+\frac{1}{\alpha} \nabla^h U(\p), \label{Cphi e phiPF: intgerazione: sistemone da integrare 3}\\[0.2cm]
			&(\mu+p)\nabla p-\nabla f=0. \label{Cphi e phiPF: intgerazione: sistemone da integrare 4}
		\end{empheq}
	\end{subequations}
	We are going to prove that \eqref{Cphi e phiPF: intgerazione: sistemone da integrare} implies \eqref{Cphi e phiPF: integrazione: conclusione bis}. 
	We will use the following equation, that will be proved in Lemma \ref{lemma: Cphi e phiPF: lemma alla Hamilton} below:
	\begin{align}\label{Cphi e phiPF: integrazione: identità pre-Hamilton}
		\frac{1}{2}S^\p_i-(m-1)\lambda_i=f_lR^\p_{li}+f_lf_{li}-\pa{\Delta f}f_i-U^a\p^a_i.
	\end{align}
	From \eqref{Cphi e phiPF: integrazione: identità pre-Hamilton} we get \eqref{Cphi e phiPF: integrazione: conclusione bis 4} in the following way: first, rearrange \eqref{Cphi e phiPF: integrazione: identità pre-Hamilton} to deduce
	\begin{align*}
		f_lR^\p_{li}-(\Delta f) f_i=\frac{1}{2}S^\p_i-(m-1)\lambda_i-f_lf_{li}+U^a\p^a_i.
	\end{align*}
	Insert the latter into \eqref{Cphi e phiPF: intgerazione: sistemone da integrare 2}:
	\begin{align*}
		0=&-\frac{m-1}{m}f_{lli}+\frac{m-2}{m}f_lf_{li}-\frac{1}{2}S^\p_i+(m-1)\lambda_i+f_lf_{li}\\
		&-U^a\p^a_i+\frac{1}{2m}S^\p_i+\frac{1}{m}U^a\p^a_i-\frac{m-1}{m}\mu_i\\
		=&-\frac{m-1}{m}\pa{\Delta_f f}_i-\frac{m-1}{2m}S^\p_i+(m-1)\lambda_i-\frac{m-1}{m}U^a\p^a_i -\frac{m-1}{m}\mu_i,
	\end{align*}
	that is,
	\begin{align*}
		0=\pa{\Delta_f f}_i+\frac{1}{2}S^\p_i-m\lambda_i+U^a\p^a_i+\mu_i.
	\end{align*}
	Integrating the latter, we get
	\begin{align*}
		0=\Delta_f f+\frac{1}{2}S^\p-m\lambda+U^\p+\mu+\Lambda_2
	\end{align*}
	for some $\Lambda_2\in \RR$. Use \eqref{Cphi e phiPF: integrazione: lambda uguale} to obtain
	\begin{align*}
		-\Lambda_2 = & \; \Delta_f f -\frac{1}{2}S^\p-\Delta_f f+U(\p)+\mu\\
		= & -\frac{1}{2}S^\p+U(\p)+\mu
	\end{align*}
	and therefore \eqref{Cphi e phiPF: integrazione: conclusione bis 4}. 
	
	Take the covariant derivative of \eqref{Cphi e phiPF: integrazione: conclusione bis 4} and rearrange
	\begin{align*}
		U^a\p^a_i=\frac{1}{2}S^\p_i-\mu_i;
	\end{align*}
	use the latter in \eqref{Cphi e phiPF: intgerazione: sistemone da integrare 2}
	\begin{align*}
		0 = & -\frac{m-1}{m}f_{lli}+\frac{m-2}{m}f_lf_{li}+(\Delta f)f_i + \frac{1}{2m}S^\p_i-f_lR^\p_{li}+\frac{1}{m}U^a\p^a_i-\frac{m-1}{m}\mu_i\\
		= & -\frac{m-1}{m}f_{lli}+\frac{m-2}{m}f_lf_{li}+(\Delta f)f_i+\frac{1}{m}S^\p_i-R^\p_{il}f_l-\mu_i.
	\end{align*}
	From \eqref{Cphi e phiPF: intgerazione: sistemone da integrare 1} we obtain
	\begin{align*}
		0 = & -\frac{m-1}{m}f_{lli}+\frac{m-2}{m}f_lf_{li}+(\Delta f)f_i+\frac{1}{m}S^\p_i-\mu_i +f_lf_{li}-\abs{\nabla f}^2f_i-\lambda f_i\\
		= & -\frac{m-1}{m}f_{lli}+2\frac{m-1}{m}f_lf_{li}+\pa{\Delta_f f}f_i+\frac{1}{m}S^\p_i - \mu_i-\lambda f_i\\
		= & -\frac{m-1}{m}\pa{\Delta_f f}_i+\pa{\Delta_f f}f_i+\frac{1}{m}S^\p_i-\mu_i-\lambda f_i.
	\end{align*}
	Use \eqref{Cphi e phiPF: integrazione: lambda uguale} to get
	\begin{align*}
		0 = &-\frac{m-1}{m}\pa{\Delta_f f}_i+\pa{\Delta_f f}f_i+\frac{1}{m}S^\p_i-\mu_i -\frac{1}{m}S^\p f_i-\frac{1}{m}\pa{\Delta_f f} f_i,
	\end{align*}
	that is,
	\begin{align*}
		0=-\frac{m-1}{m}\pa{\Delta_f f}_i+\frac{m-1}{m}\pa{\Delta_f f}f_i+\frac{1}{m}S^\p_i-\frac{1}{m}S^\p f_i-\mu_i.
	\end{align*}
	Using \eqref{Cphi e phiPF: intgerazione: sistemone da integrare 4}, the latter can be re-written
	\begin{align*}
		0=&-\frac{m-1}{m}\pa{\Delta_f f}_i+\frac{m-1}{m}\pa{\Delta_f f}f_i+\frac{1}{m}S^\p_i-\frac{1}{m}S^\p f_i-\mu_i-p_i+\pa{\mu+p}f_i.
	\end{align*}
	Multiply by $\frac{m}{m-1}$ to deduce
	\begin{align*}
		0=&-\pa{\Delta_f f}_i+\pa{\Delta_f f}f_i+\frac{1}{m-1}S^\p_i-\frac{1}{m-1}S^\p f_i-\frac{m}{m-1}\pa{\mu_i+p_i}+\frac{m}{m-1}\pa{\mu+p}f_i
	\end{align*}
	which is equivalent to
	\begin{align*}
		0=-\left\{e^{-f}\sq{\Delta_f f-\frac{1}{m-1}S^\p+\frac{m}{m-1}\pa{\mu+p}}\right\}_i.
	\end{align*}
	Integrating the latter we obtain \eqref{Cphi e phiPF: integrazione: conclusione bis 2}.
	\end{proof}
	We still owe the proof of the following
	\begin{lemma}\label{lemma: Cphi e phiPF: lemma alla Hamilton}
		Let $\p:(M,g)\to (N,h)$ be a smooth map between Riemannian manifolds and consider $\alpha\in \RR,U\in C^\infty(N), f\in C^\infty(M)$ such that \eqref{Cphi e phiPF: integrazione: conclusione bis 1} and \eqref{Cphi e phiPF: integrazione: conclusione bis 3} hold. Then the following identity holds on $M$:
			\begin{align}\label{Cphi e phiPF: integrazione: identità pre-Hamilton bis}
			\frac{1}{2}S^\p_i-(m-1)\lambda_i=f_lR^\p_{li}+f_lf_{li}-\pa{\Delta f}f_i-U^a\p^a_i.
		\end{align}
	\end{lemma}
	\begin{proof}
			From \eqref{Cphi e phiPF: integrazione: conclusione bis 1}, we clearly  have
		\begin{align}\label{Cphi e phiPF: integrazione: lambda uguale}
			m\lambda=S^\p+\Delta_f f;
		\end{align}
		taking the covariant derivative of \eqref{Cphi e phiPF: integrazione: lambda uguale} we deduce
		\begin{align}\label{Cphi e phiPF: integrazione: derivata di lambda uguale}
			S^\p_i-m\lambda_i=-f_{lli}+2f_lf_{li}.
		\end{align}
		Computing the divergence of \eqref{Cphi e phiPF: integrazione: conclusione bis 1} 
		\begin{align*}
			R^\p_{ij,j}+f_{ijj}-f_jf_{ij}-f_if_{jj}=\lambda_i
		\end{align*}
		and using the $\p$-Schur identity and the Ricci commutation relations we obtain
		\begin{align*}
			\frac{1}{2}S^\p_i-\alpha \p^a_{ll}\p^a_i+f_{lli}+f_lR_{li}-f_lf_{li}-f_i f_{ll}=\lambda_i;
		\end{align*}
		from the definition of $\ric^\p$ and \eqref{Cphi e phiPF: integrazione: conclusione bis 3} we deduce
		\begin{align*}
			-\alpha \p^a_{ll}\p^a_i+f_lR_{li}=&-\alpha \p^a_{ll}\p^a_i+f_lR^\p_{li}+\alpha \p^a_lf_l\p^a_i\\
			=&-U^a\p^a_i+f_lR^\p_{li}
		\end{align*}
		and therefore
		\begin{align*}
			\frac{1}{2}S^\p_i-U^a\p^a_i+f_lR^\p_{li}+f_{lli}-f_lf_{li}-f_if_{ll}=\lambda_i.
		\end{align*}
		Rearranging we get
		\begin{align}\label{Cphi e phi PF: integrazione: quasi identità pre Hamilton}
			\frac{1}{2}S^\p_i-\lambda_i=-f_lR^\p_{li}-f_{lli}+f_lf_{li}+f_i\Delta f+U^a\p^a_i.
		\end{align}
		Subtracting \eqref{Cphi e phi PF: integrazione: quasi identità pre Hamilton} from \eqref{Cphi e phiPF: integrazione: derivata di lambda uguale} we obtain \eqref{Cphi e phiPF: integrazione: identità pre-Hamilton bis}.\\
	\end{proof}

\appendix

\section{More on the C-$\p$-PF system}\label{Appendice A}

\subsection{Deriving the system}\label{Appendice A 1}
Starting from the Cotton Gravity equation with source given by a perfect fluid and a non-linear field $\hat{\p}$, we deduce system \eqref{Deduzione: C-phi-PF versione 3}.
\begin{theorem}
	Let $\hat M =M\times_f \RR$ be the static space-time with metric $\hat g$ given by
	\[
	\hat g = -e^{-2\hat{f}} \di t^2 + g
	\]
	with $f\in C^\infty(M), \hat f :=  \hat\pi_M\circ f$, $(M,g)$ an $m$-dimensional Riemannian manifold. For $\mu,p \in C^\infty(M)$, $\varphi :(M,g) \to (N,h)$ smooth, define
	\[
	\hat\mu = \hat\pi_M\circ \mu, \qquad \hat p = \hat\pi_M\circ p, \qquad \hat\varphi =  \hat\pi_M\circ \varphi.
	\]
	Define the stress-energy tensor $\hat T$ by
	\[
	\hat T = \hat T^{\hat{\varphi}} + \hat T^F
	\]
	with
	\begin{align*}
		\hat T^{\hat{\varphi}}  = \alpha \, \hat\varphi^\ast h - \left[ U(\hat\varphi) + \frac{\alpha}{2} |\di\hat{\varphi}|^2 \right] \hat{g},\ \ \ \ \hat T^F  =(\hat\mu +\hat p)e^{-2\hat f} \di t^2 + \hat p \hat g \, .
	\end{align*}
	for some $\alpha\in\RR\setminus\{0\}$, and assume that both $\hat{T}^{\hat{\p}}$ and $\hat{T}^F$ are divergence-free. In this setting, Cotton field equations yield
	\begin{subequations}\label{Appendice: C-phi-PF versione 3}
		\begin{empheq}[left={\empheqlbrace}]{alignat=2}
			0 = & \; C^\p_{ijk}+f_{ijk}-f_{ikj}-f_{ik}f_j+f_{ij}f_k\notag\\
			&-\frac{1}{m}\sq{\pa{f_{llk}-2f_lf_{lk}}\delta_{ij}-\pa{f_{llj}-2f_lf_{lj}}\delta_{ik}}\notag\\
			&+\frac{1}{2m(m-1)}\pa{S^\p_k\delta_{ij}-S^\p_j\delta_{ik}}\notag\\
			&+\frac{1}{m}U^a\p^a_j\delta_{ik}-\frac{1}{m}U^a\p^a_k\delta_{ij}+\frac{1}{m}\pa{\mu_j\delta_{ik}-\mu_k\delta_{ij}},\label{Appendice: C-phi-PF versione 3 1}\\
			0=&-\frac{m-1}{m}f_{lli}+\frac{m-2}{m}f_lf_{li}+\pa{\Delta f}f_i+\frac{1}{2m}S^\p_i-f_lR^\p_{il}\notag\\
			&+\frac{1}{m}U^a\p^a_i-\frac{m-1}{m}\mu_i,\label{Appendice: C-phi-PF versione 3 2}\\
			0 = & \; h\pa{\tau(\p)-\di\p(\nabla f)-\frac{1}{\alpha}\nabla^h U (\p),\di \p},\label{Appendice: C-phi-PF versione 3 4} \\ 
			0 = & \; \nabla p-(\mu+p)\nabla f \label{Appendice: C-phi-PF versione 3 3}.
		\end{empheq}
	\end{subequations}
\end{theorem}
\begin{proof}
	Let $\{\theta^i\}_{i=1}^m$ be an orthonormal coframe on $(M,g)$. Let $\{\omega^\alpha\}_{\alpha=0}^m$ be an orthonormal coframe on $(\hat M,\hat g)$ such that $\omega^i = \hat\pi^\ast_M \theta^i$ for $1\leq i\leq m$ and $\omega^0 = e^{-\hat f}\di t$. As in Section \ref{Sect 1: setting}, $\{\theta^i_{\ j}\}_{i,j=1}^m$ and $\{\omega^\alpha_{\ \beta}\}^m_{\alpha,\beta=0}$ will be the Levi-Civita connection forms for the coframes $\{\theta^i\}_{i=1}^m$ and $\{\omega^\alpha\}_{\alpha=0}^m$, respectively. Then, we have
	\begin{equation} \label{omij}
		\left\{
		\begin{array}{r@{\;}c@{\;}l}
			\omega^i_{\ j} & = & \hat\pi^\ast_M \theta^i_{\ j}, \\[0.2cm]
			\omega^0_{\ i} & = & \omega^i_{\ 0} = - f_i \omega^0 \, .
		\end{array}
		\right.
	\end{equation}
	The components of the Ricci tensor of $(M,\hat g)$ with respect to $\{\omega^\alpha\}$ are
	\begin{subequations} \label{hatRij}
		\begin{empheq}[left={\empheqlbrace}]{alignat=2}
			\hat R_{ij} & =  R_{ij} + f_{ij} - f_i f_j, \label{hatRij 1} \\
			\hat R_{00} & =  -\Delta f + |\nabla f|^2, \label{hatRij 2} \\
			\hat R_{0i} & =  0, \label{hatRij 3}
		\end{empheq}
	\end{subequations}
	where $\Delta f$ and $\abs{\nabla f}^2$ are computed with respect to the metric $g$. Note that the latter implies
	\begin{align}\label{hat S}
		\hat{S}=S+2\Delta_f f.
	\end{align}
	By the definition of covariant derivative and using \eqref{omij} and \eqref{hatRij} we get
	\begin{align*}
		\hat R_{ij,\gamma} \omega^\gamma & = \di\hat R_{ij} - \hat R_{\gamma j}\omega^\gamma_{\ i} - \hat R_{i\gamma}\omega^\gamma_{\ i} \\
		& = \di\hat R_{ij} - \hat R_{kj} \theta^k_{\ i} - \hat R_{ik} \theta^k_{\ j} \\
		& = (R_{ij,k} + f_{ijk} - f_{ik} f_j - f_i f_{jk})\theta^k \, ,
	\end{align*}
	that is,
	\begin{equation}\label{Appendice: hat R i j k}
		\left\{
		\begin{array}{r@{\;}c@{\;}l}
			\hat R_{ij,k} & = & R_{ij,k} + f_{ijk} - f_{ik} f_j - f_i f_{jk}, \\ [0.2cm]
			\hat R_{ij,0} & = & 0 \, .
		\end{array}
		\right.
	\end{equation}
	In the same fashion, we have
	\begin{align*}
		\hat R_{0i,\gamma} \omega^\gamma & = \di\hat R_{0i} - \hat R_{\gamma i}\omega^\gamma_{\ 0} - \hat R_{0\gamma}\omega^\gamma_{\ i} \\
		& = - \hat R_{ki} \omega^k_{\ 0} - \hat R_{00} \omega^0_{\ i} \\
		& = (f_k \hat R_{ki} + f_i \hat R_{00}) \omega^0 \\
		& = (f_k R_{ki} + f_k f_{ki} - |\nabla f|^2 f_i - f_i\Delta f + f_i |\nabla f|^2) \omega^0 \, ,
	\end{align*}
	that is,
	\begin{equation}\label{Appendice: hat R 0 i j}
		\left\{
		\begin{array}{r@{\;}c@{\;}l}
			\hat R_{0i,j} & = & 0, \\[0.2cm]
			\hat R_{0i,0} & = & R_{ij} f_j + f_{ij} f_j - (\Delta f) f_i \, .
		\end{array}
		\right.
	\end{equation}
	Lastly, from
	\[
	\hat R_{00,\gamma} \omega^\gamma = \di\hat R_{00} - 2\hat R_{0\gamma} \theta^\gamma_{\ 0} = \di\hat R_{00} = -\di\Delta f + \di|\nabla f|^2
	\]
	we have
	\begin{equation}\label{Appendice: hat R 0 0 i}
		\left\{
		\begin{array}{r@{\;}c@{\;}l}
			\hat R_{00,i} & = & - f_{jji} + 2 f_{ij} f_j, \\[0.2cm]
			\hat R_{00,0} & = & 0 \, .
		\end{array}
		\right.
	\end{equation}
	Next, since
	\[
	\hat S_\alpha = g^{\gamma\beta} \hat R_{\gamma\beta,\alpha} = \hat R_{tt,\alpha} - \hat R_{00,\alpha}
	\]
	we have
	\begin{equation}\label{Appendice: hat S i}
		\left\{
		\begin{array}{r@{\;}c@{\;}l}
			\hat S_i & = & S_i + 2 f_{jji} - 4 f_{ij} f_j, \\[0.2cm]
			\hat S_0 & = & 0 \, .
		\end{array}
		\right.
	\end{equation}
	We are ready to compute how the Cotton tensor of $(\hat M,\hat g)$ splits on the spatial and the time components. Recalling that $n = \dim\hat M = m+1$, $\hat C$ reads, in components,
	\[
	\hat C_{\alpha\beta\gamma} = \hat R_{\alpha\beta,\gamma} - \hat R_{\alpha\gamma,\beta} - \frac{1}{2m} (\hat S_\gamma g_{\alpha\beta} - \hat S_\beta g_{\alpha\gamma}) \, .
	\]
	Therefore, from \eqref{Appendice: hat R i j k} and \eqref{Appendice: hat S i} we get
	\begin{align*}
		\hat C_{ijk} & = R_{ij,k} - R_{ik,j} + f_{ijk} - f_{ikj} - f_{ik} f_j + f_{ij} f_k \\
		& \phantom{=\;} - \frac{1}{2m} (S_k g_{ij} - S_j g_{ik}) - \frac{1}{m} [(\Delta f - |\nabla f|^2)_k g_{ij} - (\Delta f - |\nabla f|^2)_j g_{ik}] \\
		& = C_{ijk} + f_{ijk} - f_{ikj} - f_{ik} f_j + f_{ij} f_k \\
		& \phantom{=\;} + \frac{1}{2m(m-1)} (S_k g_{ij} - S_j g_{ik}) - \frac{1}{m} [(\Delta f - |\nabla f|^2)_k g_{ij} - (\Delta f - |\nabla f|^2)_j g_{ik}].
	\end{align*}
	A direct computation yields $\hat C_{ij0} = \hat C_{0ij} = 0$. Lastly, we have, from \eqref{Appendice: hat R 0 0 i}, \eqref{Appendice: hat R 0 i j} and \eqref{Appendice: hat S i},
	\begin{align*}
		\hat C_{00i} & = \hat R_{00,i} - \hat R_{0i,0} - \frac{1}{2m} (\hat S_i g_{00} - \hat S_0 g_{0i}) \\
		& = -(\Delta f - |\nabla f|^2)_i - R_{ij} f_j - f_{ij} f_j + (\Delta f) f_i + \frac{1}{2m} S_i + \frac{1}{m} (\Delta f - |\nabla f|^2)_i \\
		& = - R_{ij} f_j - \frac{m-1}{m} (\Delta f)_i + \frac{m-2}{m} f_{ij} f_j + (\Delta f) f_i + \frac{1}{2m} S_i \, .
	\end{align*}
	Putting together the above computations, we get
	\begin{align}\label{Appendice: hat C = 0 spezzata}
		\left\{
			\begin{array}{r@{\;}c@{\;}l}
				\hat C_{ijk} & = & C_{ijk}+f_{ijk}-f_{ikj}-f_{ik}f_j+f_{ij}f_k \\ [0.2cm]
				& & -\dfrac{1}{m}\sq{\pa{f_{llk}-2f_lf_{lk}}\delta_{ij}-\pa{f_{llj}-2f_lf_{lj}}\delta_{ik}} + \dfrac{1}{2m(m-1)}\pa{S_k\delta_{ij}-S_j\delta_{ik}},\\[0.2cm]
				\hat{C}_{00i} & = & -\dfrac{m-1}{m}f_{lli}+\dfrac{m-2}{m}f_lf_{li}+\pa{\Delta f}f_i+\dfrac{1}{2m}S_i-f_lR_{il},\\[0.2cm]
				\hat{C}_{ij0} & = & 0, \\[0.2cm]
				\hat{C}_{0ij} & = & 0.
			\end{array}
		\right.
	\end{align}
	We now study the field equations. Keeping in mind that $n = \dim\hat M = m+1$, we recall from Section \ref{Sect 1: setting} that the field equations of Cotton gravity are
	\[
	\hat C = -(m-1) \div_1\hat{\mathcal T} \, 
	\]
	with $\hat{\mathcal T}$ given by \eqref{Tensore T maiuscolo}.
	As we already discussed in Section \ref{Sect 1: setting}, the field equations imply that $\hat{T}$ is divergence-free. From our assumptions, we have $\hat\varphi^a_0 = 0$, $\hat\mu_0 = \hat p_0 = 0$.
	From equation \eqref{div_That} and since $\div_1 \hat{T}=0$ we have
	\begin{equation}\label{Appendice: Divergenza di T grande}
		(\div_1\hat{\mathcal{T}})_{\alpha\beta\gamma} = \frac{1}{m-1}(\hat{T}_{\alpha\gamma,\beta}-\hat{T}_{\alpha\beta,\gamma})-\frac{1}{m(m-1)}\left[(\tr_{\hat{g}}\hat{T})_\beta g_{\alpha\gamma}-(\tr_{\hat{g}}\hat{T})_\gamma g_{\alpha\beta}\right].
	\end{equation}
	Since $\hat{T}=\hat{T}^{\hat{\p}} +\hat{T}^F$ we have
	\begin{align*}
		\div_1 \hat{\mathcal{T}}=\div_1 \hat{\mathcal{T}}^{\hat{\p}}+\div_1 \hat{\mathcal{T}}^F
	\end{align*}
	with an obvious meaning of notations. We compute the last two terms separately.
	From \eqref{Appendice: Divergenza di T grande} and using \eqref{Deduzione: derivata di T phi} and \eqref{Deduzione: traccia di T phi} we get
	\begin{align*}
		(m-1)(\div_1 \hat{\mathcal{T}}^{\hat{\p}})_{\alpha\beta\gamma} = & \; \hat{T}^{\hat{\p}}_{\alpha\gamma,\beta}-\hat{T}^{\hat{\p}}_{\alpha\beta,\gamma}-\frac{1}{m}\left[(\tr_{\hat{g}} \hat{T}^{\hat{\p}})_\beta g_{\alpha\gamma}-(\tr_{\hat{g}} \hat{T}^{\hat{\p}})_\gamma g_{\alpha\beta}\right] \\
		= & \; \alpha \bigg\{\hat{\p}^a_{\alpha\beta}\hat{\p}^a_\gamma-\hat{\p}^a_{\alpha\gamma}\hat{\p}^a_\beta-\frac{1}{\alpha}U^a\hat{\p}^a_\beta g_{\alpha\gamma}+\frac{1}{\alpha}U^a\hat{\p}^a_\gamma g_{\alpha\beta}\\
		&-g^{\eta\rho}\hat{\p}^a_{\eta\beta}\hat{\p}^a_\rho g_{\alpha\gamma}+g^{\eta\rho}\hat{\p}^a_{\eta\gamma}\hat{\p}^a_\rho g_{\alpha\beta}\\
		&+\frac{1}{m}\sq{(m-1)g^{\eta\rho}\hat{\p}^a_{\eta\beta}\hat{\p}^a_\rho g_{\alpha\gamma}+\frac{m+1}{\alpha}U^a\hat{\p}^a_\beta g_{\alpha \gamma} } \\
		&-\frac{1}{m}\sq{(m-1)g^{\eta\rho}\hat{\p}^a_{\eta\gamma}\hat{\p}^a_\rho g_{\alpha\beta}+\frac{m+1}{\alpha}U^a\hat{\p}^a_\gamma g_{\alpha\beta}} \bigg\}\\
		= & \; \alpha\bigg\{\hat{\p}^a_{\alpha\beta}\hat{\p}^a_\gamma-\hat{\p}^a_{\alpha\gamma}\hat{\p}^a_\beta-\pa{1-\frac{m+1}{m}}\frac{1}{\alpha}U^a\hat{\p}^a_\beta g_{\alpha\gamma}\\
		&+\pa{1-\frac{m+1}{m}}\frac{1}{\alpha}U^a\hat{\p}^a_\gamma g_{\alpha\beta} -\pa{1-\frac{m-1}{m}}g^{\eta\rho}\hat{\p}^a_{\eta\beta}\hat{\p}^a_\rho g_{\alpha\gamma}\\
		&+\pa{1-\frac{m-1}{m}}g^{\eta\rho}\hat{\p}^a_{\eta\gamma}\hat{\p}^a_\rho g_{\alpha\beta}\bigg\}\\
		= & \; \alpha \bigg\{\hat{\p}^a_{\alpha\beta}\hat{\p}^a_\gamma-\hat{\p}^a_{\alpha\gamma}\hat{\p}^a_\beta+\frac{1}{\alpha m}U^a\hat{\p}^a_\beta g_{\alpha\gamma} -\frac{1}{\alpha m}U^a\hat{\p}^a_\gamma g_{\alpha\beta} \\
		&-\frac{1}{m}g^{\eta\rho}\hat{\p}^a_{\eta\beta}\hat{\p}^a_\rho g_{\alpha\gamma}+\frac{1}{m}g^{\eta\rho}\hat{\p}^a_{\eta\gamma}\hat{\p}^a_\rho g_{\alpha\beta}\bigg\}.
	\end{align*}
	Set
	\begin{align*}
		\Phi_{\alpha\beta\gamma}:=&-(m-1)(\div_1 \hat{\mathcal{T}}^{\hat{\p}})_{\alpha\beta\gamma}\\
		=&-\alpha \bigg\{\hat{\p}^a_{\alpha\beta}\hat{\p}^a_\gamma-\hat{\p}^a_{\alpha\gamma}\hat{\p}^a_\beta+\frac{1}{\alpha m}(U^a\hat{\p}^a_\beta g_{\alpha\gamma}-U^a\hat{\p}^a_\gamma g_{\alpha\beta})-\frac{1}{m}g^{\eta\rho}(\hat{\p}^a_{\eta\beta}\hat{\p}^a_\rho g_{\alpha\gamma}-\hat{\p}^a_{\eta\gamma}\hat{\p}^a_\rho g_{\alpha\beta}) \bigg\}.
	\end{align*}
	Computing in the usual frame, the following holds (see pages 25 and 26 of \cite{phiSPFST}):
	\begin{align}\label{Appendice: tension di phi spezzamento}
		\begin{cases}
			\hat{\p}^a_{00}=\p^a_tf_t,\\
			\hat{\p}^a_{0i}=0,\\
			\hat{\p}^a_{ij}=\p^a_{ij}.
		\end{cases}
	\end{align}
	Using \eqref{Appendice: tension di phi spezzamento} and $\hat{\p}^a_0=0$ the following are easily computed
	\begin{equation}\label{Appendice: phi grande spezzamento}
		\left\{
		\begin{array}{r@{\;}c@{\;}l}
			\Phi_{ijk} & = & -\alpha \p^a_{ij}\p^a_k+\alpha\p^a_{ik}\p^a_j-\frac{1}{m}U^a\p^a_j\delta_{ik}+\frac{1}{m}U^a\p^a_k\delta_{ij}+\frac{\alpha}{m}\p^a_{lj}\p^a_l\delta_{ik}-\frac{\alpha}{m}\p^a_{lk}\p^a_l\delta_{ij}, \\[0.2cm]
			\Phi_{00i} & = & -\alpha \p^a_i\p^a_lf_l-\frac{1}{m}U^a\p^a_i+\frac{\alpha}{m}\p^a_{li}\p^a_l, \\[0.2cm]
			\Phi_{ij0} & = & 0, \\[0.2cm]
			\Phi_{0ij} & = & 0.
		\end{array}
		\right.
	\end{equation}
	We now study the perfect fluid stress-energy tensor.
	As we already discussed, $\hat \mu_0=\hat p_0=0$ and
	\begin{align}\label{Appendice: forme di connessione}
		\begin{cases}
			\omega^i_{\ j}=-\omega^j_{\ i}=\theta^i_{\ j},\\
			\omega^0_{\ i}=\omega^i_{\ 0}=-f_i\omega^0.
		\end{cases}
	\end{align}
	We compute the components of the covariant derivative of $\hat{T}^F$.
	From the definition of $\hat{T}^F$ we have
	\begin{align}\label{Appendice: T F}
		\begin{cases}
			\hat{T}^F_{0\alpha}=\hat{\mu}\delta_{0\alpha},\\[0.2cm]
			\hat{T}^F_{i\alpha}=\hat{p}\delta_{i\alpha}
		\end{cases}
	\end{align}
	and
	\begin{align}\label{Appendice: traccia T F}
		\tr_{\hat{g}} \hat{T}^F=m\hat{p}-\hat{\mu}.
	\end{align}
	Therefore \eqref{Appendice: T F} and \eqref{Appendice: forme di connessione} imply
	\begin{align*}
		\hat{T}^F_{ij,\alpha}\omega^\alpha & = \di \hat{T}^F_{ij}-\hat{T}^F_{\gamma j}\omega^\gamma_{\ j}-\hat{T}^F_{\gamma i}\omega^\gamma_{\ j} \\
		& = \delta_{ij}\hat{p}_\alpha\omega^\alpha-\hat{p}\omega^{j}_{\ i}-\hat{p}\omega^i_{\ j}\\
		& = \delta_{ij}\hat{p}_\alpha\omega^\alpha,
	\end{align*}
	that is,
	\begin{align}\label{Appendice: T F i j alpha}
		\hat{T}^F_{ij,\alpha}=\hat{p}_\alpha\delta_{ij}.
	\end{align}
	Using \eqref{Appendice: T F} we have
	\begin{align*}
		\hat{T}^F_{00,\alpha} \omega^\alpha & = \di \hat{T}^F_{00}-\hat{T}^F_{\gamma 0}\omega^\gamma 
		_{\ 0}-\hat{T}^F_{0\gamma}\omega^\gamma_{\ 0} \\
		& = \hat{\mu}_\alpha \omega^\alpha,
	\end{align*}
	that is,
	\begin{align}\label{Appendice T F 0 0 alpha}
		\hat{T}^F_{00,\alpha}=\hat{\mu}_\alpha.
	\end{align}
	Using \eqref{Appendice: T F} and \eqref{Appendice: forme di connessione} we get
	\begin{align*}
		\hat{T}^F_{0i,\alpha}\omega^\alpha & = \di \hat{T}^F_{0i}-\hat{T}^F_{\gamma i}\omega^\gamma_{\ 0}-\hat{T}^F_{0\gamma}\omega^\gamma_{\ i}\\
		& = 0-\hat{p}\omega^i_{\ 0}-\hat{\mu} \omega^0_{\ i}\\
		& = -(\hat{\mu}+\hat{p})\omega^i_{\ 0}=(\hat{\mu}+\hat{p})f_i\omega^0
	\end{align*}
	so that
	\begin{align}\label{Appendice: T 0 i alpha}
		\hat{T}^F_{0i,j}=0, \ \ \hat{T}_{0i,0}=(\hat{\mu}+\hat{p})f_i.
	\end{align}
	From \eqref{Appendice: Divergenza di T grande}, \eqref{Appendice: traccia T F} and \eqref{Appendice: T F i j alpha} we deduce
	\begin{align*}
		(\diver_1 \hat{\mathcal{T}}^F)_{ijk} = & \; \frac{1}{m-1}(\hat{T}^F_{ik,j}-\hat{T}^F_{ij,k})-\frac{1}{m(m-1)}\left[(\tr_{\hat{g}}\hat{T}^F)_j\delta_{ik}-(\tr_{\hat{g}}\hat{T}^F)_k\delta_{ij}\right] \\
		= & \; \frac{1}{m-1}\pa{p_j\delta_{ik}-p_k\delta_{ij}}-\frac{1}{m(m-1)}\sq{\pa{m p_j-\mu_j}\delta_{ik}-\pa{m p_k-\mu_k}\delta_{ij}}
	\end{align*}
	that is,
	\begin{align}\label{Appendice: div T F i j k}
		(\div_1\hat{\mathcal{T}}^F)_{ijk}=\frac{1}{m(m-1)}\sq{\mu_j\delta_{ik}-\mu_k\delta_{ij}}.
	\end{align}
	Using \eqref{Appendice: T F i j alpha} and \eqref{Appendice: T 0 i alpha}
	\[
		(\div_1 \hat{\mathcal{T}}^F)_{ij0} = \frac{1}{m-1}(\hat{T}^F_{i0,j}-\hat{T}^F_{ij,0})-\frac{1}{m(m-1)}[(\tr_{\hat{g}} \hat{T}^F)_j g_{i0}-(\tr_{\hat{g}}\hat{T}^F)_0g_{ij}] = 0
	\]
	and using \eqref{Appendice: T 0 i alpha}
	\[
		(\div_1\hat{\mathcal{T}}^F)_{0ij} = \frac{1}{m-1}(\hat{T}^F_{0j,i}-\hat{T}^F_{0i,j}) - \frac{1}{m(m-1)}\left[(\tr_{\hat{g}} \hat{T}^F)_i g_{0j}-(\tr_{\hat{g}}\hat{T}^F)_jg_{0i}\right] = 0.
	\]
	Using \eqref{Appendice: traccia T F}, \eqref{Appendice T F 0 0 alpha} and \eqref{Appendice: T 0 i alpha} we obtain
	\begin{align*}
		(\div_1 \hat{\mathcal{T}}^F)_{0i0} & = \frac{1}{m-1}(\hat{T}^F_{00,i}-\hat{T}^F_{0i,0})-\frac{1}{m(m-1)}\left[(\tr_{\hat{g}}\hat{T}^F)_i g_{00}-(\tr_{\hat{g}}\hat{T}^F)_0 g_{0i}\right] \\
		& = \frac{1}{m-1}\sq{\mu_i-(\mu+p)f_i}+\frac{1}{m(m-1)}\pa{mp_i-\mu_i}\\
		& = \frac{1}{m}\mu_i-\frac{\mu+p}{m-1}f_i+\frac{1}{m-1}p_i.
	\end{align*}
	Putting together these informations we get
	\begin{align}\label{Appendice: div T F spezzamento}
		\begin{cases}
			(\div_1 \hat{\mathcal{T}}^F)_{ijk} = \dfrac{1}{m(m-1)}\pa{\mu_j\delta_{ik}-\mu_k\delta_{ij}},\\[0.3cm]
			(\div_1 \hat{\mathcal{T}}^F)_{0i0}=\dfrac{1}{m}\mu_i+\dfrac{1}{m-1}\sq{p_i-(\mu+p)f_i},\\[0.3cm]
			(\div_1 \hat{\mathcal{T}}^F)_{0ij}=0,\\[0.2cm]
			(\div_1 \hat{\mathcal{T}}^F)_{ij0}=0.
		\end{cases}
	\end{align}
	Combining \eqref{Appendice: hat C = 0 spezzata}, \eqref{Appendice: phi grande spezzamento} and \eqref{Appendice: div T F spezzamento} the field equations become
	\begin{align*}
		0 = & \; C_{ijk}+f_{ijk}-f_{ikj}-f_{ik}f_j+f_{ij}f_k\\
		&-\frac{1}{m}\sq{\pa{f_{llk}-2f_lf_{lk}}\delta_{ij}-\pa{f_{llj}-2f_lf_{lj}}\delta_{ik}}\\
		&+\frac{1}{2m(m-1)}\pa{S_k\delta_{ij}-S_j\delta_{ik}}\\
		&+\alpha \p^a_{ij}\p^a_k-\alpha \p^a_{ik}\p^a_j+\frac{1}{m}U^a\p^a_j\delta_{ik}-\frac{1}{m}U^a\p^a_k\delta_{ij}\\
		&-\frac{\alpha}{m}\p^a_{lj}\p^a_l\delta_{ik}+\frac{\alpha}{m}\p^a_{lk}\p^a_l\delta_{ij}+\frac{1}{m}\pa{\mu_j\delta_{ik}-\mu_k\delta_{ij}}
	\end{align*}
	and
	\begin{align*}
		0=&-\frac{m-1}{m}f_{lli}+\frac{m-2}{m}f_lf_{li}+\pa{\Delta f}f_i+\frac{1}{2m}S_i-f_lR_{il}\\
		&+\alpha \p^a_i\p^a_lf_l+\frac{1}{m}U^a\p^a_i-\frac{\alpha}{m}\p^a_{li}\p^a_l\\
		&-\frac{m-1}{m}\mu_i+p_i-(\mu+p)f_i.
	\end{align*}
	Using the definitions of the $\p$-curvatures the above equations can be re-written
	\begin{align}\label{Appendice: spezzamento 1}
		\left\{
		\begin{array}{r@{\;}c@{\;}l}
			0 & = & C^\p_{ijk}+f_{ijk}-f_{ikj}-f_{ik}f_j+f_{ij}f_k \\ [0.2cm]
			& & - \dfrac{1}{m}\sq{\pa{f_{llk}-2f_lf_{lk}}\delta_{ij}-\pa{f_{llj}-2f_lf_{lj}}\delta_{ik}} + \dfrac{1}{2m(m-1)}\pa{S^\p_k\delta_{ij}-S^\p_j\delta_{ik}} \\ [0.2cm]
			& & + \dfrac{1}{m}U^a\p^a_j\delta_{ik}-\dfrac{1}{m}U^a\p^a_k\delta_{ij}+\dfrac{1}{m}\pa{\mu_j\delta_{ik}-\mu_k\delta_{ij}},\\ [0.2cm]
			0 & = & -\dfrac{m-1}{m}f_{lli}+\dfrac{m-2}{m}f_lf_{li}+\pa{\Delta f}f_i+\dfrac{1}{2m}S^\p_i-f_lR^\p_{il}\\ [0.2cm]
			& & + \dfrac{1}{m}U^a\p^a_i-\dfrac{m-1}{m}\mu_i+p_i-(\mu+p)f_i.
		\end{array}
		\right.
	\end{align}
	As a last step, we need to prove equations \eqref{Appendice: C-phi-PF versione 3 4} and \eqref{Appendice: C-phi-PF versione 3 3}; since by assumption we have $\div_1 \hat{T}^F=\div_1\hat{T}^{\hat{\p}}=0$, it will be enough to prove the following:
	\begin{equation} \label{Appendice: inutile}
		\begin{cases}
			\div_1 \hat{T}^F=(\hat{\mu}+\hat{p}) \hat{\nabla} \hat{f} - \hat{\nabla} \hat{p},\\
			\div_1\hat{T}^{\hat{\p}}= \alpha h\left(\tau(\hat{\varphi}) - \frac{1}{\alpha}\pa{\nabla^h U}(\hat{\varphi}),\di\hat{\varphi}\right) \, .
		\end{cases}
	\end{equation}
	First, we compute the divergence of $\hat{T}^F$. This is
	\begin{align*}
		g^{\alpha\beta}\hat{T}^F_{\gamma\alpha,\beta}=\hat{T}^F_{0i,i}-\hat{T}^F_{\gamma0,0}.
	\end{align*}
	Using \eqref{Appendice: T 0 i alpha}, \eqref{Appendice T F 0 0 alpha} and \eqref{Appendice: div T F i j k} we get
	\begin{align*}
		(\div_1\hat{T}^F)_0=\hat{T}^F_{0i,i}-\hat{T}^F_{00,0}=0-\hat \mu_0=0
	\end{align*}
	and
	\begin{align*}
		(\div_1\hat{T}^F)_j=\hat{T}^F_{ji,i}-\hat{T}^F_{j0,0}=p_i-(\mu+p)f_i=0
	\end{align*}
	so that
	\begin{align*}
		\div_1\hat{T}^F=\hat{\nabla}\hat p-(\hat \mu+\hat p)\hat{\nabla} \hat f.
	\end{align*}
	To compute $\div_1\hat{T}^{\hat{\p}}$, use equation \eqref{Deduzione: derivata di T phi} 
	\begin{align*}
		\div_1 \hat{T}^{\hat{\p}}_\gamma & = g^{\alpha\beta}\hat{T}^{\hat{\p}}_{\gamma\alpha,\beta}\\
		& = \alpha g^{\alpha\beta}\hat \p^a_{\gamma\beta}\hat \p^a_\gamma+\alpha g^{\alpha\beta}\hat \p^a_{\alpha\beta}\hat \p^a_\gamma - \alpha\pa{g^{\eta\rho}\hat \p^a_{\eta\beta}\hat \p^a_\rho+\frac{1}{\alpha}U^a\hat \p^a_\beta}g_{\gamma\alpha}g^{\alpha\beta}\\
		& = \alpha \hat \p^a_\gamma \tau(\hat{\p})^a-U^a\hat \p^a_\gamma.
	\end{align*}
	Since $\div_1\hat{T}^{\hat{\p}}=\div_1\hat{T}^F=0$, this gives \eqref{Appendice: inutile}.
	Recalling that
	\begin{align*}
		\tau(\hat{\p})=\tau(\p)-\di\p(\nabla f),
	\end{align*}
	we  get that \eqref{Appendice: C-phi-PF versione 3} is implied by   \eqref{Appendice: spezzamento 1} and \eqref{Appendice: inutile}.
\end{proof}

\subsection{On Codazzi tensors}\label{Appendice: tensori di Codazzi}
It has been first pointed out by Mantica and Molinari, \cite{ManticaCottonG}, 
how, on an $n$-dimensional Lorentzian manifold $(\hat{M},\hat{g})$, Cotton field equations \eqref{Har_eq}
be equivalent to the requirement that the tensor
\begin{align}\label{codazzi: tensore zeta hat}
	\hat{Z}_{\alpha \beta}:=\hat{R}_{\alpha \beta}-\hat{T}_{\alpha \beta}-\frac{\hat{S}-2\tr_{\hat{g}}\hat{T}}{2(n-1)}g_{\alpha \beta}
\end{align}
be Codazzi.
Indeed, by \eqref{codazzi: tensore zeta hat}
\begin{align*}
	\hat{Z}_{\alpha\beta,\gamma}-\hat{Z}_{\alpha \gamma,\beta} = & \; \hat{R}_{\alpha \beta,\gamma}-\hat{R}_{\alpha\gamma,\beta} -\hat{T}_{\alpha \beta,\gamma}+\hat{T}_{\alpha\gamma,\beta}\\
	& -\frac{\hat{S}_\gamma}{2(n-1)}g_{\alpha \beta}+\frac{\hat{S}_\beta}{2(n-1)}g_{\alpha \gamma}+\frac{(\tr_{\hat{g}}\hat{T})_\gamma}{n-1}g_{\alpha \beta}-\frac{(\tr_{\hat{g}}\hat{T})_\beta}{n-1}g_{\alpha \gamma},	
\end{align*}
that is,
\begin{align}\label{codazzi: codazzi di hat zeta}
	\hat{Z}_{\alpha\beta,\gamma}-\hat{Z}_{\alpha \gamma,\beta}=\hat{C}_{\alpha\beta\gamma}-\hat{T}_{\alpha\beta,\gamma}+\hat{T}_{\alpha\gamma,\beta}+\frac{1}{n-1}\sq{(\tr_{\hat{g}}\hat{T})_\gamma g_{\alpha\beta}-(\tr_{\hat{g}}\hat{T})_\beta g_{\alpha\gamma}}.
\end{align}

By equation \eqref{Har_eq2 bis} we see that the right hand side of \eqref{codazzi: codazzi di hat zeta} vanishes if and only if \eqref{Har_eq} holds, as we wanted to show. Similarly, we are going to see that every C-$\p$-PF admits a Codazzi tensor.

Define a 2-covariant symmetric tensor on $M$
\begin{align}\label{codazzi: tensore zeta}
	Z:=\ric^\p+\hess f-\di f\otimes \di f-\lambda g,
\end{align}
for some $\lambda \in C^\infty(M)$. In the following we will make different choices of $\lambda$.
If 
\begin{align*}
	\lambda=\frac{S^\p+2\Delta_f f+2U(\p)+2\mu}{2m},
\end{align*}
then $Z$ is the projection of $\hat{Z}$ on $(M,g)$. Indeed, since $\hat{T}=\hat{T}^{\hat{\p}}+\hat{T}^F$, its trace is
\begin{align*}
	\tr_{\hat{g}}\hat{T}&=\alpha \abs{\di \hat\p}^2-\alpha\frac{n}{2}\abs{\di \hat\p}^2-nU(\hat\p)-(\hat\mu+\hat p)+n\hat p\\
	&=\alpha\frac{2-n}{2}\abs{\di \hat\p}^2-nU(\hat \p)-\hat \mu+(n-1)\hat p,
\end{align*}
so that \eqref{codazzi: tensore zeta hat} becomes
\begin{align*}
	\hat{Z}_{\alpha\beta} = & \; \hat{R}_{\alpha \beta}-\alpha \hat \p^*h+\frac{\alpha}{2}\abs{\di \hat\p}^2\hat{g}+U(\hat \p)\hat{g}-(\hat\mu+\hat p)e^{-2\hat f}t_\alpha t_\beta\\
	&-\hat pg_{\alpha \beta}-\frac{\hat{S}-\alpha (2-n)\abs{\di \hat \p}^2+2nU(\hat \p)+2\hat\mu-2(n-1)\hat p}{2(n-1)}g_{\alpha\beta}
\end{align*}
and, simplifying,
\begin{align*}
	\hat{Z}_{\alpha \beta}=\hat{R}^{\hat \p}_{\alpha \beta}-(\hat \mu+\hat p)e^{-2\hat f} t_\alpha t_\beta-\frac{\hat{S}^{\hat\p}+2U(\hat\p)+2\hat\mu}{2(n-1)}g_{\alpha \beta}.
\end{align*}
From \eqref{hatRij 1} and \eqref{hat S} and since $n=m+1$ we have
\begin{align*}
	\hat{Z}_{ij}=R^\p_{ij}+f_{ij}-f_if_j-\frac{S^\p+2\Delta_f f+2U(\p)+2\mu}{2m}\delta_{ij}
\end{align*}
as we claimed.
A straightforward computation reveals that \eqref{Deduzione: C-phi-PF versione 3 1} is equivalent to the requirement that $Z$ be Codazzi for the choice $\lambda=\frac{1}{2m} (S^\p+2\Delta_f f+2U(\p)+2\mu).$
Recall that, using \eqref{Deduzione: C-phi-PF versione 3 2} into \eqref{Deduzione: C-phi-PF versione 3 1}, and elaborating a bit, the latter becomes
\begin{equation}\label{codazzi: prima cond int}
	0 = C^\p_{ijk}+f_lW^\p_{lijk}-\frac{1}{m-1}U^a\p^a_k\delta_{ij}+\frac{1}{m-1}U^a\p^a_j\delta_{ik} -D^A_{ijk}-(m-2)D^B_{ijk}.
\end{equation}
We want to characterize \eqref{codazzi: prima cond int} in terms of the existence of a suitable Codazzi tensor.

\begin{proposition}\label{codazzi: prop: Z e prima cond}
	Let $(M,g)$ be an $m$-dimensional Riemannian manifold and $\p:(M,g)\to (N,h)$ a smooth map with $(N,h)$ a second Riemannian manifold. Let $\alpha \in \RR,\alpha \neq 0, f\in C^\infty(M)$. Assume that, for some $\lambda\in C^\infty(M)$, the tensor
	\begin{align*}
		Z=\ric^\p+\hess f-\di f\otimes \di f-\lambda g
	\end{align*}
	be a Codazzi tensor. Then it holds
	\begin{equation}\label{codazzi: prima cond int riscritta}
		0 = C^\p_{ijk}+f_lW^\p_{lijk}-\frac{\alpha}{m-1}\pa{\p^a_{ll}-\p^a_lf_l}\pa{\p^a_k\delta_{ij}-\p^a_j\delta_{ik}} -D^A_{ijk}-(m-2)D^B_{ijk}
	\end{equation}
	on $(M,g)$.
\end{proposition}
The proof of Proposition \ref{codazzi: prop: Z e prima cond} relies on the following
\begin{lemma}\label{codazzi:lemma: D bar e prima cond}
	Let $Z$ be defined as in \eqref{codazzi: tensore zeta}. Set
	\begin{align}\label{codazzi: D barrato}
		\overline{D}_{ijk}=Z_{ij,k}-Z_{ik,j}-\frac{1}{m-1}\sq{\pa{Z_{ll,k}-Z_{lk,l}}\delta_{ij}-\pa{Z_{ll,j}-Z_{lj,l}}\delta_{ik}}.
	\end{align}
	Then we have
	\begin{equation}\label{codazzi: prima cond int con D bar}
		D^A_{ijk} +(m-2)D^B_{ijk}+\overline{D}_{ijk} = C^\p_{ijk}+f_lW^\p_{lijk}-\frac{\alpha}{m-1}\pa{\p^a_{ll}-\p^a_lf_l}\pa{\p^a_k\delta_{ij}-\p^a_j\delta_{ik}}.
	\end{equation}
\end{lemma}

\begin{remark}
	From \eqref{codazzi: prima cond int con D bar} we see that $\overline{D}$ is independent from the choice of $\lambda$.
\end{remark}

\begin{proof}[Proof of Lemma \ref{codazzi:lemma: D bar e prima cond}]
	We expand \eqref{codazzi: D barrato} using \eqref{codazzi: tensore zeta}. From \eqref{codazzi: tensore zeta}
	\begin{align*}
		Z_{ij,k}=R^\p_{ij,k}+f_{ijk}-f_{ik}f_j-f_if_{jk}-\lambda_k\delta_{ij},
	\end{align*}
	so that
	\begin{align}\label{codazzi: tracci di zeta derivata}
		Z_{ll,k}=S^\p_k+\pa{\Delta_f f}_k-m\lambda_k
	\end{align}
	and 
	\begin{align*}
		Z_{lk,l}=R^\p_{lk,l}+f_{lkl}-(\Delta f)f_k-f_lf_{lk}-\lambda_k.
	\end{align*}
	Using the $\p$-Schur identity and the Ricci commutation relations
	\[
		Z_{lk,l} = \frac{1}{2}S^\p_k-\alpha \p^a_{ll}\p^a_k+(\Delta f)_k+f_lR_{lk}-(\Delta f)f_k -f_lf_{lk}-\lambda_k.
	\]
	From the definition of $\ric^\p$ we get
	\begin{equation}\label{codazzi: divergenza di zeta}
		Z_{lk,l} = \frac{1}{2}S^\p_k-\alpha\p^a_k (\p^a_{ll}-\p^a_lf_l)+(\Delta f)_k+f_tl^\p_{lk}-(\Delta f)f_k-f_lf_{lk}-\lambda_k.
	\end{equation}
	Taking the difference of \eqref{codazzi: tracci di zeta derivata} and \eqref{codazzi: divergenza di zeta} we have
	\[
		Z_{ll,k}-Z_{lk,l}=\frac{1}{2}S^\p_k-2f_lf_{lk}-(m-1)\lambda_k-f_lR^\p_{lk}+\alpha\p^a_k(\p^a_{ll}-\p^a_lf_l)+(\Delta f)f_k+f_lf_{lk}
	\]
	so that
	\begin{equation}\label{codazzi: traccia di zeta meno diver di zeta}
		Z_{ll,k}-Z_{lk,l} = \frac{1}{2}S^\p_k-f_lf_{lk}-(m-1)\lambda_k-f_lR^\p_{lk} + \alpha\p^a_k(\p^a_{ll}-\p^a_lf_l)+(\Delta f)f_k; 
	\end{equation}
	then \eqref{codazzi: D barrato} becomes
	\begin{align*}
		\overline{D}_{ijk} = & \; R^\p_{ij,k}-R^\p_{ik,j}+f_{ijk}-f_{ikj}-f_{ik}f_j+f_{ij}f_k\\
		& -\lambda_k\delta_{ij}+\lambda_j\delta_{ik}-\frac{1}{m-1}\pa{\frac{1}{2}S^\p_k-f_lf_{lk}+(\Delta f)f_k}\delta_{ij}\\
		& -\frac{1}{m-1}\pa{-f_lR^\p_{lk}-(m-1)\lambda_k+\alpha \p^a_k(\p^a_{ll}-\p^a_lf_l)}\delta_{ij}\\
		& +\frac{1}{m-1}\pa{\frac{1}{2}S^\p_j-f_lf_{lj}+(\Delta f)f_j-f_lR^\p_{lj}-(m-1)\lambda_j+\alpha \p^a_j(\p^a_{ll}-\p^a_lf_l)}\delta_{ik},
	\end{align*}
	that is,
	\begin{align}\label{codazzi: riscrittura di D bar}
		\overline{D}_{ijk} = & \; C^\p_{ijk}+f_{ijk}-f_{ikj}-f_{ik}f_j+f_{ij}f_k\\ \nonumber
		&-\frac{1}{m-1}\sq{(\Delta f)f_k-f_lf_{lk}-f_lR^\p_{lk}+\alpha\p^a_k\pa{\p^a_{ll}-\p^a_lf_l}}\delta_{ij}\\
		&+\frac{1}{m-1}\sq{(\Delta f)f_j-f_lf_{lj}-f_lR^\p_{lj}+\alpha\p^a_j\pa{\p^a_{ll}-\p^a_lf_l}}\delta_{ik}.\nonumber
	\end{align}
	Using the Ricci commutation relations and the definition of $W^\p$ we get
	\begin{align}\label{codazzi: roba}
		f_{ijk}-f_{ikj} = & \; f_lR_{lijk}\\ \nonumber
		= & \; f_lW^\p_{lijk}+\frac{1}{m-2}\pa{f_jR^\p_{ik}-f_kR^\p_{ij}+f_lR^\p_{lj}\delta_{ik}-f_lR^\p_{lk}\delta_{ij}}\\
		&-\frac{S^\p}{(m-1)(m-2)}\pa{f_j\delta_{ik}-f_k\delta_{ij}}. \nonumber
	\end{align}
	Inserting \eqref{codazzi: roba} in \eqref{codazzi: riscrittura di D bar} we find
	\begin{align*}
		\overline{D}_{ijk} = & \; C^\p_{ijk}+f_lW^\p_{lijk}+\frac{1}{m-2}\pa{f_jR^\p_{ik}-f_kR^\p_{ij}+f_lR^\p_{lj}\delta_{ik}-f_lR^\p_{lk}\delta_{ij}}\\
		&-\frac{S^\p}{(m-1)(m-2)}\pa{f_j\delta_{ik}-f_k\delta_{ij}}-f_{ik}f_j+f_{ij}f_k\\
		&-\frac{1}{m-1}\sq{(\Delta f)f_k-f_lf_{lk}-f_lR^\p_{lk}+\alpha\p^a_k\pa{\p^a_{ll}-\p^a_lf_l}}\delta_{ij}\\
		&+\frac{1}{m-1}\sq{(\Delta f)f_j-f_lf_{lj}-f_lR^\p_{lj}+\alpha\p^a_j\pa{\p^a_{ll}-\p^a_lf_l}}\delta_{ik}.
	\end{align*}
	Simplifying and rearranging
	\begin{align*}
		\overline{D}_{ijk} = & \; C^\p_{ijk}+f_lW^\p_{lijk}+\frac{1}{m-2}\sq{f_jR^\p_{ik}-f_kR^\p_{ij}-\frac{f_l}{m-1}\pa{R^\p_{lk}\delta_{ij}-R^\p_{lj}\delta_{ik}}}\\
		&-\frac{S^\p}{(m-1)(m-2)}\pa{f_j\delta_{ik}-f_k\delta_{ij}}-f_{ik}f_j+f_{ij}f_k\\
		&-\frac{1}{m-1}\pa{-f_lf_{lk}\delta_{ij}+f_lf_{lj}\delta_{ik}+(\Delta f)f_k\delta_{ij}-(\Delta f)f_j\delta_{ik}}\\
		&-\frac{\alpha}{m-1}\pa{\p^a_{ll}-\p^a_lf_l}\pa{\p^a_k\delta_{ij}-\p^a_j\delta_{ik}}.
	\end{align*}
	From the definitions of $D^A$ and $D^B$ we deduce
	\[
		D^A_{ijk} +(m-2)D^B_{ijk}+\overline{D}_{ijk} = C^\p_{ijk}+f_lW^\p_{lijk}-\frac{\alpha}{m-1}\pa{\p^a_{ll}-\p^a_lf_l}\pa{\p^a_k\delta_{ij}-\p^a_j\delta_{ik}}	
	\]
	which is \eqref{codazzi: prima cond int con D bar}.
\end{proof}
\begin{proof}[Proof of Proposition  \ref{codazzi: prop: Z e prima cond}]
	From the definition of $\overline{D}$ it is clear that $\overline{D}\equiv 0$ if $Z$ is Codazzi so that \eqref{codazzi: prima cond int con D bar} becomes \eqref{codazzi: prima cond int riscritta}.
\end{proof}
It is now clear how \eqref{codazzi: prima cond int} generalizes \eqref{Deduzione: C-phi-PF versione 3 1}: the former holds when there exists a function $\lambda\in C^\infty(M)$ such that $Z$ is Codazzi, while the latter holds only when $\lambda=\frac{1}{2m}\pa{S^\p+2\Delta_f f+2U(\p)+2\mu}$ makes $Z$ Codazzi. \\

\section{More on Riemann compatibility}\label{Appendice B}
\subsection{Riemann compatibility and $\p$-curvatures}\label{Appendice B 1}
We already saw in \eqref{riem comp: weyl + commut} that a Riemann compatible tensor is also Weyl compatible and it commutes with the Ricci tensor. We want to show that, quite interestingly, on a C-$\p$-PF a similar result holds for the $\p$-curvatures. First, we have the following

\begin{proposition}\label{riem comp: prop: phi* h commuta con P}
	Let $(M,g)$ be an $m$-dimensional Riemannian manifold and $\p: (M,g)\to (N,h)$ a smooth map that targets another Riemannian manifold. For some $f\in C^\infty(M), \alpha \in \RR, \alpha \neq 0$, set
	\begin{align*}
		P:=\ric^\p+\hess f-\di f\otimes \di f.
	\end{align*}
	Assume that system \eqref{Deduzione: C-phi-PF: versione con le D} holds on $M$,  
	for some $\mu,p\in C^\infty(M), U\in C^\infty(N)$. Then $\p^*h$ commutes with $P$.
\end{proposition}
\begin{proof}
	From the definition \eqref{Deduzione D A definizione} of $D^A$ we immediately get
	\begin{align}\label{riem comp: prop phi comm: D A mjk }
		f_lD^A_{ljk}=\frac{1}{m-1}\pa{f_lR^\p_{lj}f_k-f_lR^\p_{lk}f_j}.
	\end{align}
	Indeed,
	\begin{align*}
		f_lD^A_{ljk} = & \; \frac{1}{m-2}\bigg[f_lR^\p_{lj}f_k-f_lR^\p_{lk}f_j+\frac{1}{m-1}(f_lR^\p_{lk}f_j-f_lR^\p_{lj}f_k) -\frac{S^\p}{m-1}(f_jf_k-f_kf_j)\bigg]\\
		= & \; \frac{1}{m-2}\sq{\frac{m-2}{m-1}(f_lR^\p_{lj}f_k-f_lR^\p_{lk}f_j)}\\
		= & \; \frac{1}{m-1}(f_lR^\p_{lj}f_k-f_lR^\p_{lk}f_j).
	\end{align*}
	In the same fashion, we have
	\begin{align}\label{riem comp: prop phi comm: D B mjk}
		f_lD^B_{ljk}=\frac{1}{m-1}\pa{f_lf_{lk}f_j-f_lf_{lj}f_k}.
	\end{align}
	We contract \eqref{Deduzione: C-phi-PF: versione con le D 1} with $f_i$ and we use \eqref{riem comp: prop phi comm: D A mjk } and \eqref{riem comp: prop phi comm: D B mjk} to get
	\begin{align*}
		0 = & \; f_lC^\p_{ljk}+f_sf_lW^\p_{sljk}-\frac{1}{m-1}U^a\p^a_kf_j+\frac{1}{m-1}U^a\p^a_jf_k-f_lD^A_{ljk}-(m-2)f_lD^B_{ljk}\\
		= & \; f_lC^\p_{ljk}-\frac{1}{m-1}U^a\p^a_kf_j+\frac{1}{m-1}U^a\p^a_jf_k-\frac{1}{m-1}(f_lR^\p_{lj}f_k-f_lR^\p_{lk}f_j)\\
		&-\frac{m-2}{m-1}\pa{f_lf_{lk}f_j-f_lf_{lj}f_k},
	\end{align*}
	that is,
	\begin{equation}\label{riem comp: prop phi comm: C phi mjk}
		f_lC^\p_{ljk} = \frac{1}{m-1}\bigg[U^a\p^a_kf_j-U^a\p^a_j f_k+f_l(R^\p_{lj}f_k-R^\p_{lk}f_j) +(m-2)f_l\pa{f_{lk}f_j-f_{lj}f_k}\bigg].
	\end{equation}
	We will also need the following:
	\begin{align}\label{riem comp: prop phi comm: gradc mu e f commutano}
		\mu_j f_i-\mu_if_j=0.
	\end{align}
	Take the covariant derivative of \eqref{Deduzione: C-phi-PF: versione con le D 3}
	\begin{align*}
		\mu_j f_i+p_jf_i+(\mu+p)f_{ij}-p_{ij}=0
	\end{align*}
	and use \eqref{Deduzione: C-phi-PF: versione con le D 3} into the latter to obtain
	\begin{align*}
		\mu_j f_i=-(\mu+p)f_if_j-(\mu+p)f_{ij}+p_{ij}.
	\end{align*}
	Since the right hand side of the equation above is symmetric in $i$ and $j$, the left hand side has to be so, and we obtain \eqref{riem comp: prop phi comm: gradc mu e f commutano}.
	Now, take the covariant derivative of \eqref{Deduzione: C-phi-PF: versione con le D}	
	\begin{align*}
		0=&-\frac{m-1}{m}f_{llij}+\frac{m-2}{m}f_lf_{lij}+\frac{m-2}{m}f_{li}f_{lj}+\pa{\Delta f}_j f_i\\
		&+\pa{\Delta f}f_{ij}+\frac{1}{2m}S^\p_{ij}-f_lR^\p_{li,j}-R^\p_{li}f_{lj}+\frac{1}{m}U^a\p^a_{ij}\\
		&+\frac{1}{m}U^{ab}\p^a_i\p^b_j-\frac{m-1}{m}\mu_{ij};
	\end{align*}
	skew symmetrize the equation above, observing that $f_tf_{tij}$ is symmetric by the Ricci commutation identities,
	\begin{align*}
		0=\pa{\Delta f}_j f_i-\pa{\Delta f}_i f_j-f_l (R^\p_{li,j}-R^\p_{lj,i})-R^\p_{li}f_{lj}+R^\p_{lj}f_{li}.
	\end{align*}
	From the definition of $C^\p$ we get
	\begin{equation}\label{riem comp: prop phi comm: conto 1}
		0 = \pa{\Delta f}_j f_i-\pa{\Delta f}_i f_j-f_lC^\p_{lij}-\frac{1}{2(m-1)}\pa{S^\p_j f_i-S^\p_if_j}-R^\p_{li}f_{lj}+R^\p_{lj}f_{li}.
	\end{equation}
	Insert \eqref{riem comp: prop phi comm: C phi mjk} in \eqref{riem comp: prop phi comm: conto 1}
	\begin{align}\label{riem comp: prop phi comm: conto 2}
		0 = & \; \pa{\Delta f}_j f_i-\pa{\Delta f}_i f_j-\frac{1}{m-1}\bigg[U^a\p^a_j f_i-U^a\p^a_if_j+f_l \pa{R^\p_{li}f_j-R^\p_{lj}f_i} \\
		&+(m-2)f_l\pa{f_{lj}f_i-f_{li}f_j}\bigg]-\frac{1}{2(m-1)}\pa{S^\p_j f_i-S^\p_if_j} -R^\p_{li}f_{lj}+R^\p_{lj}f_{li}. \nonumber
	\end{align}
	Multiply \eqref{Deduzione: C-phi-PF: versione con le D 2} by $\frac{m}{m-1}$ and rearrange to deduce
	\[
		\pa{\Delta f}_i = \frac{m-2}{m-1}f_{l}f_{li}+\frac{m}{m-1}\pa{\Delta f}f_i+\frac{1}{2(m-1)}S^\p_i-\frac{m}{m-1}f_lR^\p_{li} +\frac{1}{m-1}U^a\p^a_i-\mu_i.
	\]
	Use the latter in \eqref{riem comp: prop phi comm: conto 2} to obtain
	\begin{align*}
		0 = & \; \frac{m-2}{m-1}f_lf_{lj}f_i+\frac{m}{m-1}\pa{\Delta f}f_j f_i+\frac{1}{2(m-1)}S^\p_j f_i-\frac{m}{m-1}f_lR^\p_{lj}f_i +\frac{1}{m-1}U^a\p^a_j f_i-\mu_j f_i\\
		& -\frac{m-2}{m-1}f_lf_{li}f_j-\frac{m}{m-1}\pa{\Delta f}f_i f_j-\frac{1}{2(m-1)}S^\p_i f_j+\frac{m}{m-1}f_lR^\p_{li}f_j -\frac{1}{m-1}U^a\p^a_i f_j+\mu_i f_j\\
		& -\frac{1}{m-1}\bigg[U^a\p^a_j f_i-U^a\p^a_if_j+f_l (R^\p_{li}f_j-R^\p_{lj}f_i) + (m-2)f_l (f_{lj}f_i-f_{li}f_j) \bigg] \\
		& -\frac{1}{2(m-1)}(S^\p_j f_i-S^\p_if_j) -R^\p_{li}f_{lj}+R^\p_{lj}f_{li}. 
	\end{align*}
	Simplifying and using \eqref{riem comp: prop phi comm: gradc mu e f commutano} we get
	\begin{align*}
		0=-f_lR^\p_{lj}f_i+f_l R^\p_{li}f_j-R^\p_{li}f_{lj}+R^\p_{lj}f_{li}.
	\end{align*}
	Since $\ric^\p$ clearly commutes with itself, the latter gives the commutativity of $\ric^\p$ and $P$. From Proposition \ref{riem comp: prop: C phi PF e riem comp}, we have that $P$ is Riemann compatible so that in particular it commutes with the Ricci tensor. Combining these two informations we deduce that $P$ and $\p^*h$ commute and we are done.
\end{proof}

\begin{proposition}\label{riem comp: prop: phi weyl compatibilità}
	Let $(M,g)$ be an $m$-dimensional C-$\p$-PF. Set
	\begin{align*}
		P:=\ric^\p+\hess f-\di f\otimes \di f.
	\end{align*}
	Then we have
	\begin{subequations}\label{riem comp: prop phi weyl: phi weyl e comm con phi ric}
		\begin{empheq}[left={\empheqlbrace}]{alignat=2}
			&	W^\p_{sikl}P_{js}+W^\p_{silj}P_{ks}+W^\p_{sijk}P_{ls}=0, \label{riem comp: prop phi weyl: phi weyl e comm con phi ric 1}\\
			&P_{js}R^\p_{si}=P_{is}R^\p_{sj}, \label{riem comp: prop phi weyl: phi weyl e comm con phi ric 2}
		\end{empheq}
	\end{subequations}
	that is, $P$ is $\p$-Weyl compatible and it commutes with $\ric^\p$.
\end{proposition}

\begin{proof}
	From Proposition \ref{riem comp: prop: C phi PF e riem comp} we get that $P$  commutes with $\ric$, while Proposition \ref{riem comp: prop: phi* h commuta con P} tells us that $P$ commutes with $\ric^\p$ so that we have \eqref{riem comp: prop phi weyl: phi weyl e comm con phi ric 2}.
	By Proposition \ref{riem comp: prop: C phi PF e riem comp}, we obtain
	\begin{align}\label{riem comp: porp phi weyl: weyl comp}
		0=P_{ls}W_{sijk}+P_{ks}W_{silj}+P_{js}W_{sikl}.
	\end{align}
	From \eqref{phi curv: W phi e W} we have
	\begin{align*}
		W_{lijk} = & \; W^\p_{lijk}-\frac{\alpha}{m-2}\pa{\p^a_l\p^a_j\delta_{ik}-\p^a_l\p^a_k\delta_{ij}+\p^a_i\p^a_k\delta_{lj}-\p^a_i\p^a_j\delta_{lk}}\\
		& +\frac{\alpha}{(m-2)(m-1)}\abs{\di \p}^2\pa{\delta_{lj}\delta_{ik}-\delta_{lk}\delta_{ij}}.
	\end{align*}
	Using the latter in \eqref{riem comp: porp phi weyl: weyl comp} we deduce
	\begin{align*}
		0 = & \; P_{ls}W^\p_{sijk}-\frac{\alpha}{m-2}\pa{P_{ls}\p^a_s\p^a_j\delta_{ik}-P_{ls}\p^a_s\p^a_k\delta_{ij}+P_{lj}\p^a_i\p^a_k-P_{lk}\p^a_i\p^a_j}\\
		&+\frac{\alpha}{(m-2)(m-1)}\abs{\di \p}^2\pa{P_{jl}\delta_{ik}-P_{kl}\delta_{ij}}\\
		&+P_{ks}W^\p_{silj}-\frac{\alpha}{m-2}\pa{P_{ks}\p^a_s\p^a_l\delta_{ij}-P_{ks}\p^a_s\p^a_j\delta_{il}+P_{lk}\p^a_i\p^a_j-P_{jk}\p^a_i\p^a_l}\\
		&+\frac{\alpha}{(m-2)(m-1)}\abs{\di \p}^2\pa{P_{kl}\delta_{ij}-P_{kj}\delta_{il}}\\
		&+P_{js}W^\p_{sikl}-\frac{\alpha}{m-2}\pa{P_{js}\p^a_s\p^a_k\delta_{il}-P_{js}\p^a_s\p^a_l\delta_{ik}+P_{kj}\p^a_i\p^a_l-P_{lj}\p^a_i\p^a_k}\\
		&+\frac{\alpha}{(m-2)(m-1)}\abs{\di \p}^2\pa{P_{jk}\delta_{il}-P_{jl}\delta_{ik}}.\\
	\end{align*}
	Using the fact that $P$ and $\p^*h$ commute and simplifying we obtain \eqref{riem comp: prop phi weyl: phi weyl e comm con phi ric 1}.
\end{proof}
\subsection{Riemann compatibility and 4-dimensional Riemannian geometry}\label{Appendice B 2}
We recall some well-know facts about differential topology and the geometry of a Riemannian 4-manifolds. This will allow us to strengthen some relatively recent results of Mantica and Molinari and, at the same time, to characterize the injectivity of the $W^*$ operator in dimension 4, giving more depth to the results in Section \ref{Sect 5: riem comp}.\\
Consider the homomorphism
\begin{align*}
	R:\Lambda^2 TM\to \text{End }TM
\end{align*}
that sends $X\wedge Y\in \Lambda^2TM$ to the endomorphism of $TM$ of local components
\begin{align*}
	R(X\wedge Y)_{\ j}^{i}=X^kY^tR^{i}_{\ jkt}.
\end{align*}
Then, for any positive integer $k$, we can define a $4k$-covariant tensor field $\omega_k$  on $M$ by setting, for all $X_1,...,X_{4k}\in TM$
\begin{align*}
	\omega_k(X_1,...,X_{4k})=\text{Trace}\sq{R(X_1\wedge X_2)\circ R(X_3\wedge X_4)\circ ...\circ R(X_{4k-1}\wedge X_{4k})}
\end{align*}
where the composition and the trace operation are taken with respect to $\text{End } TM$. \\
For $4k\leq m$, the $k$-th \emph{Pontryagin form}, $\Omega_k(R)$,  of $M$ is defined to be the total anti-symmetrization of $\omega_k.$ \\
The next theorem shows how the existence of a Riemann compatible tensor can constrain the geometry of $(M,g)$.
\begin{theorem}[Theorem 5.3 of \cite{Mantica2012RiemannCT}]\label{riem comp: teo pontryagin}
	Let $(M,g)$ be an $m$-dimensional Riemannian manifold. Suppose that there exist on $M$ a Riemann compatible, 2-covariant symmetric tensor $P$ and a point $x\in M$ such that $P$ has $m$ distinct eigenvalues at $x$. Then every Pontryagin form of $M$ vanishes at $x$.
\end{theorem}
As we will see, we can relax the assumptions of Theorem \ref{riem comp: teo pontryagin} when $m=4$. It is well known, by the celebrated Hirzebruch signature formula, that the integral of the first Pontryagin  form $\Omega_1(R)$ on a closed 4-manifold $(M,g)$ coincides with the topological signature $\tau(M)$ of $M$, that is, with the signature of the intersection form of the second, real, singular co-homology group $H^2(M,\RR)$ of $M$. Elaborating on this fact, we are going to prove the following    
\begin{theorem}\label{riem comp: teo signature}
	Let $(M,g)$ be a compact, orientable Riemannian manifold of dimension $4$ and let $P$ be a Weyl compatible, 2-covariant, symmetric tensor on $M$.
	If, on a dense open subset of $M$, we have
	\begin{align}\label{riem comp: teo signature: ipotesi autovalori di P}
		P\neq \frac{\tr_g P}{4}g,
	\end{align}
	then the signature $\tau(M)$ of $M$ is zero.
\end{theorem}
\begin{remark}
	Theorems \ref{riem comp: teo pontryagin} and \ref{riem comp: teo signature} generalize two results of Derdzinski and Shen, see \cite{DerdzinskiShen}, valid under the assumption that $P$ be Codazzi.
\end{remark}
\begin{remark}
	In the following, we will actually prove a stronger statement than that of Theorem \ref{riem comp: teo signature}, that is, we will prove that, under the given assumptions, the first Pontryagin form $\Omega_1(R)$ of $(M,g)$ vanishes identically on $M$. Of course, when $m=4$, $\Omega_1(R)$ is the only Pontryagin form of $(M,g)$ so that we have actually generalized Theorem \ref{riem comp: teo pontryagin}: the same conclusion is now obtained under assumption \eqref{riem comp: teo signature: ipotesi autovalori di P} which is equivalent to the requirement that the endomorphism of $TM$ induced by $P$ has more than one eigenvalue on a dense open subset of $M$.
	We decided to formulate Theorem \ref{riem comp: teo signature} in terms of $\tau(M)$ rather than $\Omega_1(R)$ because of the deeper geometric meaning of the former compared to the latter and in order to make contact with \cite{DerdzinskiShen}.
\end{remark}
We now recall some standard facts regarding the algebraic structures of the Weyl tensor in dimension 4, in order to provide
 simple conditions on the Weyl tensor under which $W^*_{|_{S^2_0}}$ is injective.
When $m=4$ it is well-known that the space of $2$-forms on $M$ decomposes in the spaces of self-dual and anti-self dual forms, i.e.
\begin{align*}
	\Lambda^2(M)=\Lambda^2_+(M)\oplus \Lambda^2_-(M).
\end{align*}
Moreover, the Weyl tensor also decomposes in self-dual and anti-self dual parts, which take the form
\begin{align*}
	W^+:=\frac{1}{2}\pa{W+W^*}, \ \ W^-:=\frac{1}{2}\pa{W-W^*}.
\end{align*}
These define endomorphisms
\begin{align*}
	W^+_{|_{\Lambda^2_+}}:\Lambda^2_+\to \Lambda^2_+, \ \ W^-_{|_{\Lambda^2_-}}:\Lambda^2_-\to \Lambda^2_- 
\end{align*}
which assign to a 2-form $\omega \in \Lambda^2_{\pm}$ of components $\omega_{ij}$ the two form $W^{\pm}_{|_{\Lambda^2_{\pm}}}(\omega)$ of components
\begin{align*}
	W^{\pm}_{|_{\Lambda^2_{\pm}}}(\omega)_{ij}=W^{\pm}_{ijkt}\omega_{kt}.
\end{align*}
With these notations in mind, it is proved in \cite{GoverA.Rod2007FCC} that the injectivity of $W^*_{|_{S^2_0}}$ holds under the assumptions
\begin{align}\label{riem comp: ipotesi Gover Nagy}
	W^+=0, \ \ \det W^-_{|_{\Lambda^2_-}}\neq 0,
\end{align}
where $\det W^-_{|_{\Lambda^2_-}}$ is the determinant of $W^-_{|_{\Lambda^2_-}}$. 
More in depth, it is proved there, see also \cite[Section 4]{Gursky2021CurvatureOT}, that the endomorphism $W^*_{|_{S^2_0}}$ is completely determined by the endomorphisms $W^+_{|_{\Lambda^2_+}}$ and $W^-_{|_{\Lambda^2_-}}$.
Indeed, consider at a given point $p\in M$, an orthonormal basis $\omega^1,\omega^2,\omega^3$ of $\Lambda^2_+(T^*_pM)$ made of eigenforms of $W^+_{|_{\Lambda^2_+}}$ of respective eigenvalues $\lambda^+_1,\lambda^+_2,\lambda^+_3$. Similarly, consider at $p$ an orthonormal basis $\eta^1,\eta^2,\eta^3$ of $\Lambda^2_-(T^*_pM)$ made of eigenforms of $W^-_{|_{\Lambda^2_-}}$ of respective eigenvalues $\lambda^-_1,\lambda^-_2,\lambda^-_3$.
Then it is proved in Item iv) of \cite[Proposition 4.1]{Gursky2021CurvatureOT} that the nine 2-covariant tensors
\begin{align*}
	h^{\alpha,\beta}_{ij}:=\omega^\alpha_{ip}\eta^{\beta}_{pj}, \ \ \alpha,\beta \in\{1,2,3\}
\end{align*}
are symmetric, trace-free and they form a basis of $S^2_0(T^*_pM)$.
Moreover, $W_{|_{S^2_0}}$ is diagonalized on this basis and the eigenvalue corresponding to $h^{\alpha,\beta}$ is
\begin{align*}
	\lambda^+_{\alpha}+\lambda^-_{\beta},
\end{align*}
see \cite[Proposition 4.3]{Gursky2021CurvatureOT}.
Similarly, $W^*_{|_{S^2_0}}$ is diagonalized by $\{h^{\alpha,\beta}\}_{\alpha,\beta=1,2,3}$ and $h^{\alpha,\beta}$ has eigenvalue
\begin{align*}
	\lambda_\alpha^+-\lambda_\beta^-.
\end{align*}
From this we see that $W^*_{|_{S^2_0}}$ fails to be injective at $p$ if and only if, for some $\alpha,\beta\in\{1,2,3\}$ we have
\begin{align*}
	\lambda_{\alpha}^+=\lambda_{\beta}^-.
\end{align*}
In particular, this cannot happen if \eqref{riem comp: ipotesi Gover Nagy} holds.
\begin{remark}
	It is possible to give a more precise characterization of the relation between $W^+$ and $W^-$ on a four-manifold admitting a non-trivial Weyl compatible tensor. Indeed, we have the validity of the following 
	\begin{lemma}\label{riem comp: lemma spectas uguali}
		Let $(M,g)$ be a Riemannian four-manifold and let $P$ be a Weyl compatible tensor on $M$. Then, at any point $x\in M$ such that
		\begin{align*}
			P\neq \frac{\tr_g P}{4}g
		\end{align*}
		holds at $x$, we have that the spectra of $W^+_{|_{\Lambda^2_+}}$ and $W^-_{|_{\Lambda^2_-}}$ coincide, with equal multiplicities.
	\end{lemma}
	The above result has been proved in \cite[Lemma 2]{DerdzinskiBundles} under the stronger assumption that $P$ be Codazzi.
	The same argument used there works in this situation if one applies  \cite[Proposition 2.4]{Mantica2012WeylCT} instead of  \cite[Theorem 1]{DerdzinskiShen}.
\end{remark}
\begin{proof}[Proof of Theorem \ref{riem comp: teo signature}]
	The proof is a simple application of Lemma \ref{riem comp: lemma spectas uguali}. Indeed, it is well known, see equation $(2.21)$ of \cite{Bourguignon}, that, on an orientable $4$-manifold, the first Pontryagin form is
	\begin{align*}
		\Omega_1(R)=\pa{\abs{W^+}^2-\abs{W^-}^2}\epsilon,
	\end{align*}
	where $\epsilon$ is again the volume form of $(M,g)$.
		Under the assumptions of Theorem \ref{riem comp: teo signature} and using Lemma \ref{riem comp: lemma spectas uguali} we obtain $\abs{W^+}^2=\abs{W^-}^2$ on a dense open subset of $M$, so that $\Omega_1(R)\equiv 0$ on $M$.
	From the Hirzebruch signature formula one gets
	\begin{align*}
		\tau(M)=\int_M \pa{\abs{W^+}^2-\abs{W^-}^2}=0
	\end{align*}
and we are done.
\end{proof}


\begin{thebibliography}{99}
	\bibitem{AMR}
	L. J. Alìas, P. Mastrolia and M. Rigoli, \emph{Maximum principles and geometric applications}, Springer Monographs in Mathematics, Springer, Cham, 2016. MR 3445380
	
	\bibitem{Altas2025}
	E. Altas  and B. Tekin,
		\emph{Vanishing of conserved charges in Cotton gravity},
		Phys. Rev. D
		\textbf{111},
		no. 2,  L021503,
		(2025).
\bibitem{Altas2025Killing}
	E. Altas  and B. Tekin,
	\emph{Consistency problems of conformal Killing gravity},
	Phys. Rev. D \textbf{111},
	no. 6, 064083,
	(2025).
	\bibitem{Ambrozio}
	L. Ambrozio, \emph{On static three-manifolds with positive scalar curvature}, Journal of Differential Geometry \textbf{107} (2015).
	
	\bibitem{Anselli_2021} A. Anselli, \emph{Bach and Einstein's equations in presence of a field}, Int. J. Geom. Methods Mod. Phys. \textbf{18} (2021), no. 5, paper no. 2150077, 68 pp.
	
	\bibitem{ACR} A. Anselli, G. Colombo and M. Rigoli, \emph{On the geometry of Einstein-type structures}, Nonlinear Anal. \textbf{204} (2021), paper no. 112198, 84 pp.
\bibitem{PhysRevD.104.088501} P. Bargueño,
	 \emph{Comment on ``Emergence of the Cotton tensor for describing gravity''}, Phys. Rev. D \textbf{104}, no. 8, 104.088501	(2021).

\bibitem{BorghiniMazzieri1}
S. Borghini and L. Mazzieri,\emph{On the mass of static metrics with positive cosmological	constant: {I}}, Classical and Quantum Gravity \textbf{35} (2018) no. 12.
	
	\bibitem{Bourguignon} J. P. Bourguignon, \emph{Les variétés de dimension 4 à signature non nulle dont la courbure est
	harmonique sont d’Einstein}, Invent. Math. \textbf{63} n.2 (1981) 263-286.
   \bibitem{phiSPFST} L. Branca, G.Colombo, P. Mastrolia, F. Mastropietro and M. Rigoli, \emph{On the geometry of $\p$-Static Perfect Fluid Space-Times}, preprint, submitted for publication.
\bibitem{Cao2013}
	H.D. Cao and Q. Chen,
	\emph{On {B}ach-flat gradient shrinking {R}icci solitons},
	Duke Math. J.,
	\textbf{162},
	(2013),
	no. 6,
	1149--1169.
	

	\bibitem{CatMastMontRig}
	G. Catino, P. Mastrolia, D. D. Monticelli and M. Rigoli, \emph{On the geometry
		of gradient Einstein-type manifolds}, Pacific J. Math. 286 (2017), no. 1, 39–67. MR 3582400
	\bibitem{Feng2024}
		P. Chen and J. Feng,
			\emph{Cosmological constant as an integration constant},
	Eur. Phys. J. C 
		\textbf{84},
		no. 12,
		(2024).
\bibitem{Chrusciel}
P. T. Chruściel, \emph{Remarks on rigidity of the de Sitter metric}, homepage.univie.ac.at/piotr.chrusciel/papers/
deSitter/deSitter2.pdf

	\bibitem{Clement:2023tyx}
	G. Clément  and K. Nouicer,
	\emph{Cotton gravity is not predictive},
	Phys. Lett. B
	\textbf{856}, 138947,
	(2024).
	
	
	\bibitem{CMR2022} G. Colombo, L. Mari and M. Rigoli, \emph{Einstein-type structures, Besse's conjecture, and a uniqueness result for a $\varphi $-CPE metric in its conformal class}, J. Geom. Anal. \textbf{32} (2022), no. 11, paper no. 267, 32 pp.
	
\bibitem{Ribeiro2} J. Costa, R. Diógenes, N. Pinherio and E. Ribeiro, \emph{Geometry of static perfect fluid space-time}, (2023), 2306.00225.
	\bibitem{Ribeiro1} F. Coutinho, R. Diógenes, B. Leandro, E: Ribeiro, \emph{Static perfect fluid space-time on compact manifolds}, Class. Quantum Grav. \textbf{37}, (2020), no. 1, 015003.
	
\bibitem{DerdzinskiBundles}
 A. Derdzinski  , \emph{Riemannian metrics with harmonic curvature on 2-sphere bundles
over compact surfaces}, Bull. Soc. Math. France \textbf{116} (1988), 133–156
\bibitem{DerdzinskiShen}  A. Derdzinski and C. L. Shen, \emph{Codazzi tensor fields, curvature and Pontryagin forms}, Proc.
London Math. Soc. \textbf{47} n.3 (1983) 15-26.	
	

	\bibitem{EL} J. Eells and L. Lemaire, \emph{A report on harmonic maps}, Bull. London Math. Soc. \textbf{10} (1978), no. 1, 1--68.
	
	\bibitem{EL2} J. Eells and L. Lemaire, \emph{Two reports on harmonic maps}, World Scientific Publications, River Edge, NJ, 1995.
	\bibitem{Fardoun1997HarmonicMW} A. Fardoun and A. Ratto, \emph{Harmonic maps with potential}, Calc. Var. Partial Differential
	Equations \textbf{5} (1997), no. 2, 183–197. MR 1433176
	\bibitem{FRRharmonicPot} A. Fardoun, A. Ratto and R. Regbaoui, \emph{On the heat flow for harmonic maps with
	potential}, Ann. Global Anal. Geom. \textbf{18} (2000), no. 6, 555–567. MR 1800592
\bibitem{FischerMarsden1}
A. E. Fischer and J. E. Marsden, \emph{Manifolds of Riemannian metrics with prescribed scalar curvature}, Bulletin (new series) of the American Mathematical Society \textbf{80} (1974) no. 3, pg. 479-484.
\bibitem{FischerMarsden2}
A. E. Fischer and J. E. Marsden, \emph{Deformations of the scalar curvature}, Duke mathematical journal, \textbf{42} (1975) no. 3, pg. 519-547.
	
	\bibitem{GoverA.Rod2007FCC} A. R. Gover and P.-A. Nagy, \emph{Four-dimensional conformal c-spaces}, Quarterly journal of mathematics \textbf{58} (2007), no. 4, 443--462.
	
	\bibitem{GoverA.Rod2006OtcE} A. R. Gover and P. Nurowski, \emph{Obstructions to conformally Einstein metrics in $n$ dimensions}, J. Geom. Phys. \textbf{56} (2006), no. 3, 450--484.
	
	\bibitem{Gursky2021CurvatureOT} M. J. Gursky, X. Cao and H. Tran, \emph{Curvature of the second kind and a conjecture of Nishikawa}, Comm. Math. Helvetici (2021).
	\bibitem{Harada} J. Harada, \emph{Emergence of the Cotton tensor for describing gravity}, Phys. Rev. D \textbf{103} (2021), no. 11, paper no. L121502.
	\bibitem{Harada2022} J. Harada, \emph{Cotton gravity and 84 galaxy rotation curves}, Phys. Rev. D \textbf{106},
		no. 6. 064044, (2022)
\bibitem{Hawking} S. W. Hawking and G. F. R. Ellis, \emph{The large scale structure of space-time, anniversary
ed., Cambridge Monographs on Mathematical Physics}, Cambridge University Press,
Cambridge, 2023, With a foreword by Abhay Ashtekar. MR 4615777
\bibitem{HwangYun2}
S. Hwang and G. Yun, \emph{Vacuum Static Spaces with Vanishing of Complete Divergence of Weyl Tensor}, The Journal of Geometric Analysis (2020) \textbf{31}.
\bibitem{Kobayashi}
O. Kobayashi, \emph{A differential equation arising from scalar curvature function}, J. Math.
Soc. Japan \textbf{34}, No. 4 (1982), 665-675.

\bibitem{KO}
	O. Kobayashi and M. Obata,
	\emph{Conformally-flatness and static space-time},
	Manifolds and {L}ie groups ({N}otre {D}ame, {I}nd., 1980),
	 Progr. Math.,
	\textbf{14},
	197--206,
	(1981),
\bibitem{Lafontaine}
J. Lafontaine, \emph{Sur la geometrie d’une generalisation de l’equation d’Obata}, J. Math.
Pures Appliquees \textbf{62} (1983), 63-72

\bibitem{Lemaire1977OnTE} L. Lemaire, \emph{On the existence of harmonic maps}, University of Warwick, 1977, Unpublished.

\bibitem{Lindblom} L. Lindblom, \emph{Some properties of static general relativistic stellar models}, J. Math. Phys. \textbf{21}, (1980), 6, 1455-1459.

    \bibitem{List2008EvolutionOA} B. List, \emph{Evolution of an extended Ricci flow system}, Comm. Anal. Geom. \textbf{16} (2008), no. 5, 1007--1048.


	\bibitem{Mantica2012RiemannCT} C. A. Mantica and L. G. Molinari, \emph{Riemann compatible tensors}, Colloq. Math \textbf{128}
	(2012), 197–210.
	\bibitem{Mantica2012WeylCT}
	C. A. Mantica and L. G. Molinari, \emph{Weyl compatible tensors}, Int. J. Geom. Meth. Mod. Phys.
	\textbf{11} (2014) 1450070 (15 pp).
\bibitem{ManticaCottonG}
 C. A. Mantica and L. G. Molinari, \emph{Codazzi tensors and their
spacetimes, and Cotton gravity}, Gen. Relativ. Gravit. \textbf{55}, 62
(2023).
\bibitem{mantica2025noteconservedcurrentsstatic}
	C.A. Mantica and L. G. Molinari,
	\emph{Note on conserved currents in static Conformal Killing Gravity}, 
	(2025)
	eprint={2504.12840},
	url={https://arxiv.org/abs/2504.12840}, 

\bibitem{Masood}
A. Masood-ul-Alam, \emph{Proof that static stellar models are spherical}, General Relativity and Gravitation \textbf{39}, (2007), 55-85.

\bibitem{sussman2024responsecritiquecottongravity}
 C. A. Mantica and L. G. Molinari and S. Nájera and R. A Sussman,
	\emph{Response to a critique of "Cotton Gravity"},
	(2024),
	https://arxiv.org/abs/2401.10479. 

\bibitem{MR} L. Marini and M. Rigoli, \emph{On the geometry of $\varphi$-curvatures}, J. Math. Anal. Appl. \textbf{483} (2020), no. 2, paper no. 123657, 22 pp.
\bibitem{QingWan}
J. Qing and W. Yuan, \emph{A note on static spaces and related problems}, J. Geom. Phys.
\textbf{74} (2013), 18-27.

\bibitem{Ratto} A. Ratto, \emph{Harmonic maps with potential}, Rend. Circ. Mat. Palermo (\textbf{2}) Suppl. (1997),
no. 49, 229–242. MR 1603011
\bibitem{Shen}
Y. Shen, \emph{A note on Fischer–Marsden’s conjecture}, Proc. Amer. Math. Soc. \textbf{125} (1997), pg. 901-905.

\end{thebibliography}
\end{document}